\documentclass[11pt,a4paper]{amsart}
\usepackage{graphicx}
\usepackage{xspace}
\usepackage{hyperref}
\usepackage[mathscr]{euscript}
\usepackage{amsmath}
\usepackage{amssymb}
\usepackage{amsthm}
\usepackage[top=1.5in,bottom=1.5in,left=1in,right=1in]{geometry}

\usepackage{thmtools,thm-restate}
\usepackage{tikz}
\usepackage{mathabx}
\usepackage{enumitem}
\usetikzlibrary{decorations.pathreplacing,hobby,backgrounds,trees,arrows,automata,shapes,positioning,fit,patterns,calc,matrix}

\newcommand{\cont}[1]{\ensuremath{\mathit{alph}(#1)}\xspace}

\newcommand{\wh}[1]{\ensuremath{\widehat{#1}}\xspace}

\newcommand{\fupp}[2]{\ensuremath{\lceil #1 \rceil_{#2}}\xspace}

\newcommand{\eupp}[1]{\fupp{#1}{\eta}}

\newcommand{\nat}{\mathbb{N}\xspace}

\newcommand{\veps}{\ensuremath{\varepsilon}\xspace}

\newcommand{\Opo}{\ensuremath{\textup{Op}}\xspace}
\newcommand{\Op}[1]{\ensuremath{\mathop{\textup{Op}(#1)}}}
\newcommand{\bool}[1]{\ensuremath{\textup{Bool}(#1)}\xspace}
\newcommand{\pol}[1]{\ensuremath{\textup{Pol}(#1)}\xspace}
\newcommand{\bpol}[1]{\ensuremath{\textup{BPol}(#1)}\xspace}

\newcommand{\polo}{\ensuremath{\textup{Pol}}\xspace}
\newcommand{\bpolo}{\ensuremath{\textup{BPol}}\xspace}

\newcommand{\polp}[2]{\ensuremath{\textup{Pol}_{#1}(#2)}\xspace}
\newcommand{\copolp}[2]{\ensuremath{\textit{co\textup{-}}\!\textup{Pol}_{#1}(#2)}\xspace}
\newcommand{\capolp}[2]{\ensuremath{\polp{#1}{#2} \cap \copolp{#1}{#2}}\xspace}

\newcommand{\bpolp}[2]{\ensuremath{\textup{BPol}_{#1}(#2)}\xspace}

\newcommand{\Bs}{\ensuremath{\mathscr{B}}\xspace}
\newcommand{\Cs}{\ensuremath{\mathscr{C}}\xspace}
\newcommand{\Ds}{\ensuremath{\mathscr{D}}\xspace}

\newcommand{\Gs}{\ensuremath{\mathscr{G}}\xspace}

\newcommand{\Hb}{\ensuremath{\mathbf{H}}\xspace}

\newcommand{\Kb}{\ensuremath{\mathbf{K}}\xspace}
\newcommand{\Lb}{\ensuremath{\mathbf{L}}\xspace}

\newcommand{\Vb}{\ensuremath{\mathbf{V}}\xspace}

\newcommand{\inv}{\ensuremath{^{-1}}}

\newcommand{\at}{\ensuremath{\textup{AT}}\xspace}

\newcommand{\stzer}{\ensuremath{\textup{ST}}\xspace}
\newcommand{\dotzer}{\ensuremath{\textup{DD}}\xspace}

\newcommand{\fo}{\ensuremath{\textup{FO}}\xspace}

\usepackage[scaled=1.]{dsserif}

\newcommand{\sic}[1]{\ensuremath{\Sigma_{#1}}\xspace}
\newcommand{\bsc}[1]{\ensuremath{\Bs\Sigma_{#1}}\xspace}

\newcommand{\md}{\ensuremath{\textup{MOD}}\xspace}

\newcommand{\grp}{\ensuremath{\textup{GR}}\xspace}

\newcommand{\vari}{prevariety\xspace}

\newcommand{\varis}{prevarieties\xspace}
\newcommand{\Varis}{Prevarieties\xspace}

\newcommand{\pvari}{positive prevariety\xspace}

\DeclareMathOperator{\uclos}{\uparrow}

\newcommand{\Jrel}{\ensuremath{\mathrel{\mathscr{J}}}\xspace}
\newcommand{\Hrel}{\ensuremath{\mathrel{\mathscr{H}}}\xspace}
\newcommand{\Rrel}{\ensuremath{\mathrel{\mathscr{R}}}\xspace}
\newcommand{\Lrel}{\ensuremath{\mathrel{\mathscr{L}}}\xspace}

\newcommand{\Jord}{\ensuremath{\leqslant_{\mathscr{J}}}\xspace}
\newcommand{\Hord}{\ensuremath{\leqslant_{\mathscr{H}}}\xspace}
\newcommand{\Rord}{\ensuremath{\leqslant_{\mathscr{R}}}\xspace}
\newcommand{\Lord}{\ensuremath{\leqslant_{\mathscr{L}}}\xspace}

\newcommand{\Jords}{\ensuremath{<_{\mathscr{J}}}\xspace}
\newcommand{\Hords}{\ensuremath{<_{\mathscr{H}}}\xspace}
\newcommand{\Rords}{\ensuremath{<_{\mathscr{R}}}\xspace}
\newcommand{\Lords}{\ensuremath{<_{\mathscr{L}}}\xspace}

\newcommand{\frI}{\ensuremath{\mathbb{I}}\xspace}
\newcommand{\frB}{\ensuremath{\mathbb{B}}\xspace}

\newcommand{\smarrow}[1]{\ensuremath{\xrightarrow{\smash{#1}}}\xspace}

\newcommand{\infsig}[1]{\ensuremath{\frI_{#1}}\xspace}

\newcommand{\infsigc}{\infsig{\Cs}}

\tikzset{every state/.style={draw=blue!50!green,very thick,fill=blue!50!green!20}}
\tikzset{statesub/.style={state,minimum size=1.3cm,inner sep=1pt}}
\tikzset{pattstate/.style={state,draw=red!50!yellow,line width=2pt,fill=red!50!yellow!20}}
\tikzset{pdotstate/.style={state,minimum size=0.75cm,inner sep=0.5pt,draw=red!50!yellow,line
		width=2pt,dashed,fill=red!50!yellow!20}}
\tikzset{lstate/.style={state,minimum size=0.5cm,inner sep=0.5pt}}
\tikzstyle{trans}=[shorten >= 1pt,thick,->]

\makeatletter
\tikzstyle{initial by arrow}=   [after node path=
{
	{
		[to path=
		{
			[->,double=none,shorten >= 1pt,thick,every initial by arrow]
			([shift=(\tikz@initial@angle:\tikz@initial@distance)]\tikztostart.\tikz@initial@angle)
			node [shape=coordinate,anchor=\tikz@initial@anchor,draw] {\tikz@initial@text}
			-- (\tikztostart)}]
		edge ()
	}
}]
\makeatother

\makeatletter
\tikzstyle{accepting by arrow}=   [after node path=
{
	{
		[to path=
		{
			[->,double=none,shorten >= 1pt,thick,every accepting by arrow]
			(\tikztostart) --
			([shift=(\tikz@accepting@angle:\tikz@accepting@distance)]\tikztostart.\tikz@accepting@angle)
			node [shape=coordinate,anchor=\tikz@accepting@anchor,draw] {\tikz@accepting@text}
		}
		]
		edge ()
	}
}]
\makeatother

\AtBeginDocument{}

\theoremstyle{plain}
\newtheorem{theorem}{Theorem}
\newtheorem{corollary}[theorem]{Corollary}
\newtheorem{fact}[theorem]{Fact}

\newtheorem{lemma}[theorem]{Lemma}
\newtheorem{remark}[theorem]{Remark}

\begin{document}

\title[Dot-depth three]{Dot-depth three, return of the \Jrel-class}
\author{Thomas Place \and Marc Zeitoun}
\email{firstname.name@labri.fr}
\address{Univ. Bordeaux, CNRS, Bordeaux INP, LaBRI, UMR 5800, F-33400 Talence, France}

\begin{abstract}
  We look at \emph{concatenation hierarchies} of classes of regular languages. Each such hierarchy is determined by a single  class, its basis: level $n$ is built by applying the Boolean polynomial closure operator (BPol), $n$ times to the basis. A prominent and difficult open question in automata theory is to decide membership of a regular language in a given level. For instance, for the historical dot-depth hierarchy, the decidability of membership is only known at levels \emph{one} and \emph{two}.

  We give a generic algebraic characterization of the operator~\bpolo. This characterization implies that for any concatenation hierarchy, if $n$ is at least two, membership at level $n$ reduces to a more complex problem, called \emph{covering}, for the previous level, $n-1$. Combined with earlier results on covering, this implies that \mbox{membership} is decidable for dot-depth \emph{three} and for level \emph{two} in most of the prominent hierarchies in the literature. For instance, we obtain that the levels \emph{two} in both the modulo hierarchy and the group hierarchy have decidable membership.
\end{abstract}
\maketitle

\section{Introduction}
\label{sec:intro}
\noindent\textbf{Context.}
This article addresses one of the most fundamental questions in automata theory, which emerged in the~1970s: the dot-depth problem. It originates from an algorithm by Schützenberger~\cite{schutzsf} that determines whether a regular language is \emph{star-free}. Star-free languages are specific regular languages obtained from singletons using three regularity-preserving operations: union, concatenation, and complement (on the other hand, star-free expressions disallow Kleene star). Schützenberger's result initiated a line of research aiming to understand subclasses of regular languages defined by high-level formalisms (\emph{e.g.}, a restricted set of regular expressions). Moreover, it highlighted a natural parameter of star-free languages, which measures their complexity: the number of alternations between complement and concatenation required to express~them.

Schützenberger's work led Brzozowski and Cohen~\cite{BrzoDot} to define a hierarchy of classes of languages known as the \emph{dot-depth}. Roughly speaking, languages at level $n$ in this hierarchy are those that can be defined by a star-free expression with at most $n-1$ alternations between complement and concatenation. Brzozowski and Knast showed that the dot-depth is infinite (\emph{i.e.}, it does not collapse)~\cite{BroKnaStrict}. The significance of the dot-depth hierarchy is further evidenced by its logical formulation, discovered by Thomas~\cite{ThomEqu}: level $n$ in the dot-depth corresponds to the class $\bsc{n}$ of all languages that can be defined in first-order logic under an appropriate signature, with at most $n$ blocks of quantifiers of the form $\exists^*$ or~$\forall^*$.

\smallskip
The dot-depth actually inspired a general family of hierarchies called the \emph{concatenation hierarchies}. The simplest way to define them is through \emph{operators}. An operator \Opo takes as input a class \Cs of regular languages and produces as output a larger class $\Op{\Cs}$. In order to build a concatenation hierarchy, we require two of them. First, the \emph{polynomial closure}  \pol{\Cs} of a class~\Cs consists of all finite unions of concatenation products $L_1a_1L_1\cdots a_nL_n$ where the languages $L_i$ belong to \Cs and the $a_i$'s are letters. Second, the \emph{Boolean closure} \bool{\Cs} of a class \Cs is its closure under all Boolean operations. Finally, ``\bpolo'' is the composition of these two operators: $\bpol\Cs=\bool{\pol\Cs}$ for each class \Cs. All concatenation hierarchies follow the same construction pattern: it starts from a base class, level~zero; then, level $n+1$ is built from level~$n$ by applying the ``\bpolo'' operator once. Thus, each hierarchy is entirely determined by its level zero. For the dot-depth hierarchy, level zero consists of four languages: the empty one, the singleton consisting of the empty word, and their complements. Other prominent hierarchies include the \emph{Straubing-Thérien hierarchy}~\cite{StrauConcat,TheConcat} whose basis consists of only the empty and universal languages, the \emph{group hierarchy}~\cite{MargolisP85} whose basis consists of the languages recognized by a permutation automaton~\cite{permauto} or the \emph{modulo hierarchy}~\cite{ChaubardPS06} whose basis consist of the ``modulo languages'', which are those for which membership of a word is determined by its length modulo some integer.

The logical characterization by quantifier alternation hierarchies is also a generic feature of concatenation hierarchies. Each such hierarchy corresponds to the alternation hierarchy in first-order logic equipped with a signature determined by its basis~\cite{PZ:generic18}. For instance, the Straubing-Thérien hierarchy~\cite{PPOrder} is associated with the signature consisting only of the linear order ``$<$'' on positions and the unary alphabetical predicates $a(x)$, $b(x)$,\ldots\ In order to capture the dot-depth, one needs to add the successor predicate ``$+1$''~\cite{ThomEqu}. The modulo hierarchy~\cite{ChaubardPS06} is also a natural example. Its signature consists of the linear order, the alphabetical predicates and the modular predicates, which test position's \mbox{values modulo an integer}.

\smallskip In summary, the levels of concatenation hierarchies have a dual characterization: by operators on the one hand, by logic on the~other. This underscores  that the minimum level containing a language is a \emph{robust} parameter. Intuitively, it measures the complexity required to express a language. This intuition is formalized by the fact that the non-elementary lower bound for checking satisfiability in first-order logic comes from quantifier alternation~\cite{stockphd}. This has motivated researchers to tackle the original Schützenberger membership problem for individual level of concatenation hierarchies: do we have an algorithm to test whether an input language belongs to such a level? If so, we say that the level is \emph{decidable}.

\smallskip\noindent\textbf{State of the art.}
Historically, this question has drawn the attention of numerous researchers, see~\cite{jep-dd45} for a survey. Initially, it was mostly addressed for \emph{specific levels} in the dot-depth or Straubing-Thérien hierarchies. Nevertheless, due to its inherent complexity, advances have been gradual. Simon~\cite{simonphd,simonthm} proposed an algorithm for level~one in the Straubing-Thérien hierarchy. This is an influential result and there are many alternate proofs~\cite{pinvarbook,stbc1proof,almeidabc1proof,higginsbc1proof,klimabc1proof}. Knast~\cite{knast83} also tackled level~one, but for the dot-depth hierarchy (see also~\cite{Therien88,MR1647225,KufleitnerL12} for alternate proofs). Straubing then designed a key reduction~\cite{StrauVD} that shows that a level $n$ is decidable in the Straubing-Thérien hierarchy if and only if the same holds for the dot-depth. Additionally, level~one has been proven to be decidable in the group hierarchy. This breakthrough was achieved~\cite{henckell:hal-00019815,krpgbg} in a purely algebraic context, leveraging a challenging result by Ash~\cite{Ash91} (see \cite{jeppgbg} for a survey). Finally, there are also several results concerning the lower levels of the modulo hierarchy~\cite{ChaubardPS06,KufleitnerW15}

\smallskip
Following these first results, the decidability of dot-depth \emph{two} became a prominent open problem in automata theory. Many partial results, conjectures and over-approximations were proposed~\cite{conjdd2-blanchet1,conjdd2-blanchet2,phdpw,Cowan_1993,StrauDD2Conf,StraubingDD2Journal,Weil_1989,Straubing_1992,Weil_1993,pwdelta2,Pin_2001,AK2009,AK2010} until decidability was finally proved in~\cite{pzqalt,pzjacm19}. The algorithms presented in \cite{pzqalt} and~\cite{pzjacm19} differ slightly, but they both rely on a more general problem than membership, called \emph{separation}, which is decidable for a subclass of dot-depth two. Intuitively, solving \Cs-separation for a class \Cs, brings information that can be exploited for studying classes built on top of \Cs (typically through operators). A striking example of this claim is the \emph{generic} reduction~\cite{PZ:generic18} from \pol{\Cs}-\emph{membership} to \Cs-\emph{separation} (when \Cs has mild properties). It is based on a generic characterization of \pol{\Cs} on the \emph{syntactic monoid} of the input language (a canonical algebraic recognizer that can be computed from the input). The class \Cs defines a binary relation on the syntactic monoid whose computation boils down to \Cs-separation. The characterization states that the input language belongs to \pol{\Cs} if and only if every pair of this relation satisfies a generic equation depending only on \polo (and not on \Cs).

\medskip\noindent\textbf{Contributions.} We investigate the \bpolo \emph{operator}, as opposed to focusing on \bpol\Cs for a \emph{specific} class~\Cs. This is desirable in order to obtain generic results and to pinpoint which hypotheses on the input class \Cs are necessary for deciding \bpol{\Cs}-membership (thus abstracting away the irrelevant peculiarities of~\Cs). This approach is not new, see for instance~\cite{Pin_2001}. Our aim is to obtain a transfer result from  \Cs-separation to \bpol{\Cs}-membership, similar to the one existing for \polo. Unfortunately, \bpolo is much harder to deal with than \polo. Indeed, in practice, the classes which are closed under \emph{language concatenation} (such as \pol{\Cs}) are simpler to handle than those which are not (such as \bpol{\Cs} in general).

\smallskip
Our first contribution is a \emph{generic} algebraic characterization of \bpolo (for input classes satisfying mild properties). Its presentation is similar to the aforementioned characterization of \polo. It uses two equations that the syntactic monoid of a language satisfies if and only if the language in is \bpol{\Cs}. Each one parameterized by a relation defined from \Cs. The first of these relations is the one used in the characterization of \pol{\Cs}: its computation boils down to \Cs-separation. Unfortunately, the second relation is based on a new \emph{ad hoc} problem for~\Cs whose decidability status is not clear in general. Consequently, our first characterization does not immediately yield a \bpol{\Cs}-membership procedure for most classes \Cs. Nonetheless, we do get an interesting simple corollary: a new effective characterization of level \emph{two} in the Straubing-Thérien hierarchy. While such a characterization was known~\cite{pzjacm19}, this new one is much simpler. In particular, this yields a new proof for the decidability of dot-depth \emph{two} by the aforementioned reduction of Straubing~\cite{StrauVD}.

\smallskip
Our second contribution is a specialized characterization of \bpolo for the particular inputs \Cs which are themselves of the form \bpol{\Ds}, for some other class \Ds. In other words, \Cs is a \emph{non-zero} level in an arbitrary concatenation hierarchy and \bpol{\Cs} is a level above \emph{two} in the same hierarchy. In this case, we are able to reformulate the second equation of our first generic characterization. We end-up with a new equation parameterized by a relation based on the \Cs-\emph{covering} problem (a generalization of separation to more than just two input languages). Therefore, the main corollary of this second characterization is that \bpol{\bpol{\Ds}}-\emph{membership} reduces to \bpol{\Ds}-\emph{covering}. As the latter problem has already been extensively studied for many classes \bpol{\Ds}, we obtain as a corollary new positive decidability results of membership for several~levels.

\smallskip\noindent\textbf{Applications.} When \Gs consists only of \emph{group languages} and has decidable \emph{separation}, it is shown in~\cite{pzconcagroup, PlaceZ22} that level \emph{one} in the hierarchies of bases \Gs and $\Gs^+$ has decidable \emph{covering}. Here, $\Gs^+$ consists of all languages of the form $L\setminus\{\varepsilon\}$ and $L\cup\{\varepsilon\}$, with $L\in \Gs$. Combined with our results, this implies that for such classes \Gs, level \emph{two} in the hierarchies of bases \Gs and $\Gs^+$ has decidable \emph{membership}. Most of the bases that are considered in the literature are of this kind. In particular, since their bases have decidable separation~\cite{Ash91,pzgr}, we obtain that level \emph{two} in both the \emph{modulo hierarchy} and the \emph{group hierarchy} has decidable membership.

Our other main applications concern the Straubing-Thérien hierarchy and the dot-depth. It was shown in~\cite{pzbpolcj} that level \emph{two} has decidable \emph{covering} in both hierarchies. Combined with our results, this implies that level \emph{three} in both the Straubing-Thérien hierarchy and the dot-depth has decidable membership.

\smallskip\noindent\textbf{Organization of the paper.}  In Section~\ref{sec:prelims}, we introduce preliminary definitions. We define concatenation hierarchies in Section~\ref{sec:hiera}. Section~\ref{sec:pairs} is devoted to relations that we use in our characterizations. The two characterizations are presented in Section~\ref{sec:bpolgen} and Section~\ref{sec:uptwo}.

\section{Preliminaries}
\label{sec:prelims}

\subsection{Words, classes and regular languages}

We fix an arbitrary finite alphabet $A$ for the whole paper. As usual, $A^*$ denotes the set of all words over $A$, including the empty word~\veps. We let $A^{+}=A^{*}\setminus\{\veps\}$. For $u,v \in A^*$, we write $uv$ the word obtained by concatenating $u$ and~$v$. A \emph{language} is a subset of $A^*$. Abusing terminology, we often write $u$ for the singleton $\{u\}$. We lift concatenation to languages:  for $K,L\subseteq A^*$, we let $KL=\{uv \mid u \in K \text{\;and\;} v \in L\}$.

A \emph{class of languages} \Cs is a set of languages. We say that \Cs is a \emph{lattice} when it is closed under union and intersection, $\emptyset \in \Cs$ and $A^* \in \Cs$. Also, a \emph{Boolean algebra} is a lattice closed under complement. Finally, if for all $L \in \Cs$ and $w \in A^*$, the languages $w\inv L =  \{u \in A^* \mid wu \in L\}$ and $Lw\inv =  \{u \in A^* \mid uw \in L\}$ both belong to~\Cs, we say that \Cs is closed under  \emph{quotients}.

A \emph{\pvari} (resp.\ \emph{\vari}) is a lattice (resp.\ Boolean algebra) closed under quotients, containing only \emph{regular languages}. These are the languages that can be equivalently defined by finite automata or finite monoids. We use the latter definition.

\smallskip
\noindent
{\bf Monoids.} A \emph{monoid} is a set $M$  endowed with an associative multiplication $(s,t)\mapsto st$ which has an identity  $1_M$, \emph{i.e.}, such that ${1_M}\cdot s=s\cdot {1_M}=s$ for every~$s \in M$. A \emph{group} is a monoid $G$ such that for all $g\in G$, there exists $g^{-1}$ such that $gg^{-1}=g^{-1}g = 1_G$.

An \emph{idempotent} of a monoid $M$ is an element $e \in M$ such that $ee = e$. We write $E(M) \subseteq M$ for the set of all idempotents in $M$. It is folklore that for each \emph{finite} semigroup~$M$, there is a natural number $\omega(M)$ (denoted by $\omega$ when $M$ is understood) such that for every $s \in M$, the element $s^\omega$ is an idempotent.

An \emph{ordered monoid} is a pair $(M,\leq)$ where $M$ is a monoid and $\leq$ is a partial order on $M$, compatible with multiplication: for all $s,s',t,t' \in M$, if $s \leq t$ and $s' \leq t'$, then $ss'\leq tt'$. A subset $F \subseteq M$ is an upper set (for~$\leq$) if it is upward closed for $\leq$: for all $s,t\in M$, if $s \in F$ and $s \leq t$, then $t \in F$.   By convention, we view every \emph{unordered} monoid $M$ as the ordered one $(M,=)$ whose ordering is \emph{equality}. In this case, \emph{all} subsets of $M$ are upper sets. Thus, a definition for ordered monoids makes sense for unordered ones.

Finally, we use the Green relations~\cite{green} defined on monoids. We briefly recall them. Given a monoid $M$ and $s,t \in M$,
\[
  \begin{array}{ll@{\ }l}
    s \Jord t & \text{when} & \text{there exist $x,y\in M$ such that $s=xty$}, \\
    s \Lord t & \text{when} & \text{there exists $x \in M$ such that $s = xt$}, \\
    s \Rord t & \text{when} & \text{there exists $y \in M$ such that $s = ty$}, \\
    s \Hord t & \text{when} & \text{$s \Lord t$ and $s \Rord t$}.
  \end{array}
\]
These relations are preorders (\emph{i.e.}, reflexive and transitive). We write \Jords, \Lords, \Rords and \Hords for their strict variants (\emph{e.g.}, $s \Jords t$ when $s \Jord t$ but $t \not\Jord s$). We write \Jrel, \Lrel, \Rrel and \Hrel for the associated equivalences (\emph{e.g.}, $s \Jrel t$ when $s \Jord t$ and $t \Jord s$). The following lemmas are standard (see \emph{e.g.}, \cite{pingoodref}).

\begin{restatable}{lemma}{jlr} \label{lem:jlr}
  Let $M$ be a finite monoid and $s,t \in M$ such that $t \Jord s$. If $s \Rord t$, then $s \Rrel t$. Moreover, if $s \Lord t$ then $s \Lrel t$.
\end{restatable}

\begin{restatable}{lemma}{htogroup} \label{lem:htogroup}
	Let $M$ be a finite monoid, $e \in E(M)$ and $s \in M$ be such that $e \Hrel s$. Then, $s^\omega = e$.
\end{restatable}

\noindent
{\bf Regular languages.} Note that $A^*$ is a monoid: its multiplication is word concatenation, and its identity is \veps. We consider morphisms $\alpha: A^* \to (M,\leq)$ where $(M,\leq)$ is an arbitrary ordered monoid. That is, $\alpha:A^*\to M$ is a map such that $\alpha(\varepsilon)=1_M$ and $\alpha(uv)=\alpha(u)\alpha(v)$ for all $u,v\in A^*$. We say that $L \subseteq A^*$ is \emph{recognized} by the morphism $\alpha$ when there is an \emph{upper set} $F \subseteq M$ for $\leq$ such that $L= \alpha\inv(F)$.

\begin{restatable}{remark}{convequ}
  Recall that we view any unordered monoid $M$ as the ordered one $(M,=)$. Thus, a language $L \subseteq A^*$ is recognized by a morphism $\alpha: A^* \to M$ if there exists an \emph{arbitrary} set $F \subseteq M$ (all subsets are upper sets for ``$=$'') such that $L = \alpha\inv(F)$.
\end{restatable}

It is standard that the regular languages are those which can be recognized by a morphism into a \emph{finite} monoid.

\begin{restatable}{remark}{finmono} \label{rem:finmono}
	The only infinite monoid that we consider is~$A^*$. We shall implicitly assume that every other monoid $M,N,\dots$ is finite.
\end{restatable}

\subsection{Membership, separation and covering}

We define three decision problems. Each of them depends on an arbitrary fixed class~\Cs and take regular languages as input. We use them as mathematical tools for investigating~\Cs.

The first problem is \emph{\Cs-membership}. It takes a regular language $L$ as input and asks whether $L$ belongs to \Cs.

The second problem is \emph{\Cs-separation}. It takes as input two regular languages $L_1$ and $L_2$ and asks if $L_1$  is \emph{\Cs-separable} from $L_2$: does there exist $K \in \Cs$ such that $L_1 \subseteq K$ and $L_2 \cap K = \emptyset$?

\begin{restatable}{remark}{sepgenmemb} \label{rem:sepgenmemb}
	Separation generalizes membership. Indeed, for any regular language $L$, we have $L \in \Cs$ if and only if
	$L$ is \Cs-separable from $A^* \setminus L$ (which is also regular).
\end{restatable}

The third problem is \emph{\Cs-covering}. It takes as input a regular language~$L_1$ and a finite set  $\Lb_2$ of regular languages, and asks whether the pair $(L_1,\Lb_2)$ is \emph{\Cs-coverable}: does there exist a \Cs-cover \Kb of $L_1$ (\emph{i.e.}, \Kb is a \emph{finite} set of languages in \Cs such that $L \subseteq \bigcup_{K \in \Kb} K$) such that for~all $K\in\Kb$, there exists $L \in \Lb_2$ satisfying $K \cap L = \emptyset$?

\begin{restatable}{remark}{covgensep} \label{rem:covgensep}
  When \Cs is a lattice, this generalizes separation: $L_1$ is \Cs-separable from $L_2$ if and only if $(L_1,\{L_2\})$ is \Cs-coverable.
\end{restatable}

Finally, the definition of \Cs-covering is simplified when~\Cs is a \emph{Boolean algebra}. Given a finite set of languages~\Lb, we say that \Lb is \emph{\Cs-coverable} to mean that $(A^*,\Lb)$ is \Cs-coverable. The following simple lemma is proved in~\cite{pzcovering2}.

\begin{restatable}{lemma}{cobool}\label{lem:covbopl}
  Let \Cs be a Boolean algebra, $L_1 \subseteq A^*$ be a language and $\Lb_2$ be a finite set of languages. Then, $(L_1,\Lb_2)$ is \Cs-coverable if and only if $\{L_1\} \cup \Lb_2$ is \Cs-coverable.
\end{restatable}

\subsection{\Cs-morphisms}

We turn to a key mathematical tool (see~\cite{pzupol,pzupol2} for the proofs). Given a \pvari \Cs, a morphism $\eta: A^*\to (N,\leq)$ into a finite ordered monoid
is a \emph{\Cs-morphism} if it is \emph{surjective} and all languages recognized by $\eta$ belong to \Cs. If \Cs is a \vari (\emph{i.e.}, \Cs is also closed under complement), it suffices to consider \emph{unordered} monoids (by convention, we view them as ordered by equality).

\begin{restatable}{lemma}{cmorphbool} \label{lem:cmorphbool}
	Let \Cs be a \vari and let $\eta: A^* \to (N,\leq)$ be a morphism. The two following conditions are equivalent:
	\begin{enumerate}
		\item $\eta: A^* \to (N,\leq)$ is a \Cs-morphism.
		\item $\eta: A^* \to (N,=)$ is a \Cs-morphism.
	\end{enumerate}
\end{restatable}

\begin{restatable}{remark}{cmbool} \label{rem:cmbool}
	Another formulation of Lemma~\ref{lem:cmorphbool} is that when \Cs is a \vari, whether a morphism $\eta: A^* \to (N,\leq)$ is a \Cs-morphism does not depend on the ordering $\leq$. Thus, we shall implicitly use Lemma~\ref{lem:cmorphbool} and drop the ordering when dealing with \varis.
\end{restatable}

We use \Cs-morphisms as mathematical tools. In particular, we need the following result.

\begin{restatable}{proposition}{genocm}\label{prop:genocm}
  Let \Cs be a \pvari and let $L_1,\dots,L_k \in \Cs$. There exists a \Cs-morphism $\eta: A^* \to (N,\leq)$ recognizing $L_1,\dots,L_k$.
\end{restatable}

We also use \Cs-morphisms to tackle membership. It is known that every regular language $L$ is~recognized by a canonical morphism $\alpha_L: A^*  \to (M_L,\leq_L)$, called its \emph{syntactic morphism}, which can be computed from any representation of $L$ (see \emph{e.g.},~\cite{pingoodref}). It is standard that for every \pvari \Cs, we have $L \in \Cs$ if and only if $\alpha_L$ is a \Cs-morphism. Altogether, this yields the following proposition.

\begin{restatable}{proposition}{synmemb} \label{prop:synmemb}
	Let \Cs be a \pvari. There is~an effective reduction from \Cs-membership to the problem of deciding whether a morphism $\alpha: A^* \to (M,\leq)$ is a \Cs-morphism.
\end{restatable}

\section{Concatenation hierarchies}
\label{sec:hiera}
Concatenation hierarchies are built uniformly from a single input class using two operators, which we first present.

\subsection{Polynomial closure and Boolean closure}

The \emph{polynomial closure of a class \Cs} is the class \pol{\Cs} consisting of all \emph{finite unions} of \emph{marked products} $L_0a_1L_1 \cdots a_nL_n$ where $n \in \nat$ $a_1,\dots,a_n \in A$ and $L_0,\dots,L_n \in \Cs$. It is not clear from the definition whether the class \pol{\Cs} has robust properties, even when \Cs does. It was shown by Arfi~\cite{arfi87,arfi91} that if~\Cs is a \vari, then \pol{\Cs} is a \pvari. This result was later strengthened by Pin~\cite{jep-intersectPOL}: if \Cs is a \pvari, then so is \pol{\Cs} (see also~\cite{PZ:generic18}).

\begin{restatable}{theorem}{polc}\label{thm:polc}
  Let \Cs be a \pvari. Then, \pol{\Cs} is a \pvari closed under language concatenation.
\end{restatable}

In general, the classes \pol{\Cs} are \emph{not} closed under complement. Hence, it is natural to combine polynomial closure with another operator. Given a class \Ds, the \emph{Boolean closure} of \Ds, written \bool{\Ds}, is the \emph{least} Boolean algebra containing~\Ds. For each input class \Cs, we write \bpol{\Cs} for \bool{\pol{\Cs}}. Since quotients commute with Boolean operations, we have the following corollary of Theorem~\ref{thm:polc}.

\begin{restatable}{corollary}{bpolc}\label{cor:bpolc}
  If \Cs is a \pvari, \bpol{\Cs} is a \vari.
\end{restatable}

Observe that by definition, a language in \bool{\Cs} is built from finitely many languages in \Cs. This yields the following lemma for \bool{\Cs}-morphisms.

\begin{restatable}{lemma}{bpolm} \label{lem:bpolm}
  Let \Cs be a \pvari and $\alpha: A^* \to M$ be a \bool{\Cs}-morphism. There exists a morphism $\beta: Q \to M$  and a \Cs-morphism $\gamma: A^* \to (Q,\leq)$ such that $\alpha = \beta \circ \gamma$.
\end{restatable}

\begin{proof}
  By hypothesis on $\alpha$, there exists a finite set \Lb of languages in \Cs such that every language recognized by $\alpha$ is a Boolean combination of languages in \Lb. Since \Cs is a \pvari, Proposition~\ref{prop:genocm} yields a \Cs-morphism $\gamma: A^* \to (Q,\leq)$ recognizing every language $L \in \Lb$. It remains to construct $\beta: Q \to M$. For every $q \in Q$, we fix $w_q \in \gamma\inv(q)$ and define $\beta(q) = \alpha(w_q)$. We prove that $\beta$ is a morphism. This boils down to proving that $\gamma(u) = \gamma(v) \Rightarrow \alpha(u) = \alpha(v)$ for all $u,v \in A^*$. Indeed, since $\gamma(w_{1_Q}) = 1_Q = \gamma(\veps)$, we would get $\alpha(w_{1_Q}) = \alpha(\veps)$, and therefore $\beta(1_Q) = 1_M$. Similarly, for $q,r \in Q$, since $\gamma(w_{qr}) = qr = \gamma(w_qw_r)$, we would get $\alpha(w_{qr}) = \alpha(w_qw_r)$, which yields $\beta(qr) =   \beta(q)\beta(r)$.

  Let now $u,v \in A^*$ such that $\gamma(u) = \gamma(v)$, and let us prove that $\alpha(u) = \alpha(v)$. Since $\gamma$ recognizes every language $L\in\Lb$, we have $u \in L \Leftrightarrow v \in L$ for all $L \in \Lb$. Hence, for every Boolean combination $H$ of languages of \Lb, we have $u \in H \Leftrightarrow v \in H$. Since every language recognized by $\alpha$ is such a Boolean combination, we obtain $\alpha(u) = \alpha(v)$.

  It remains to prove that $\alpha=\beta\circ\gamma$. Let $w\in A^*$ and $q=\gamma(w)$. We have $\gamma(w)=\gamma(w_q)$. Therefore, $\alpha(w)=\alpha(w_q)=\beta(q)=\beta(\gamma(w))$, as desired.
\end{proof}

The classes \bpol{\Cs}, which are central to this paper, are difficult to handle. This is because they are \emph{not} closed under concatenation in general. This makes it hard to specify languages in \bpol{\Cs}.  Nonetheless, we shall use the following weak concatenation principle based on \bpol{\Cs}-covers, proved in~\cite[Lemma~3.4]{pzbpolcj}.

\begin{lemma}\label{lem:bpconcat}
  Let \Cs be a \vari and $L_0,\dots,L_n \in \pol{\Cs}$. For $i \leq n$, let $\Kb_i$ be a \bpol{\Cs}-cover of\/ $L_i$. Then, there exists a \bpol{\Cs}-cover \Hb of\/ $L_0 \cdots L_n$ such that for every $H \in \Hb$, there exists $K_i \in \Kb_i$ for all $i \leq n$, such that $H \subseteq K_0 \cdots K_n$.
\end{lemma}

We shall also need the following variant of Lemma~\ref{lem:bpconcat}, which considers \emph{marked} concatenation. It is presented in~\cite[Lemma~3.6]{pzbpolcj}.

\begin{lemma}\label{lem:bpconcatm}
  Let \Cs be a \vari, $L_0,\dots,L_n \in \pol{\Cs}$ and let $a_1,\dots,a_n \in A$. For $i \leq n$, let $\Kb_i$ be a \bpol{\Cs}-cover of\/ $L_i$. Then, there is a \bpol{\Cs}-cover \Hb of\/ $L_0a_1L_1 \cdots a_nL_n$ such that for every $H \in \Hb$, there exists $K_i \in \Kb_i$ for all $i \leq n$, such that $H \subseteq K_0a_1K_1 \cdots a_nK_n$.
\end{lemma}

Lemmas~\ref{lem:bpconcat} and ~\ref{lem:bpconcatm} only apply to very specific languages to be covered: such a language must be a concatenation of languages in \pol{\Cs} that are themselves already covered. To build these specific languages, we shall use another result which we present now.

\smallskip

We start with definitions. Let $\eta: A^* \to N$ be a morphism (such as a \Cs-morphism) and $n \in \nat$. An \emph{$\eta$-pattern of length $n$} is a tuple $\widebar{p}=(w_1,e_1,\dots,w_n,e_n,w_{n+1})$ of $(A^+ \times E(N))^n \times A^*$ satisfying the three following conditions:
\begin{enumerate}
  \item\label{itm:pat1} If $n \geq 1$, then $w_1, \dots,w_{n+1} \in A^+$ are \emph{all nonempty},
  \item\label{itm:pat2} $e_i \Jrel \eta(w_1 \cdots w_{n+1})$ for $1 \leq i \leq n$,
  \item\label{itm:pat3} $\eta(w_i)e_i\! =\! \eta(w_i)$ and $e_i\eta(w_{i+1})\! =\! \eta(w_{i+1})$ for $1\! \leq\! i\! \leq\! n$.
\end{enumerate}
If $\alpha: A^* \to M$ is another morphism, we say that $\widebar{p}$ is \emph{$\alpha$-guarded} if for all $1 \leq i \leq n$, there exists $g_i \in E(M)$ such that $\eta\inv(e_i) \cap \alpha\inv(g_i) \neq \emptyset$, $\alpha(w_i)g_i=\alpha(w_i)$ and $g_i\alpha(w_{i+1})=\alpha(w_{i+1})$. Note that $\widebar{p}$ is $\eta$-guarded if $\eta$ is surjective. Note also that when $n = 0$, Conditions~\eqref{itm:pat1} to~\eqref{itm:pat3} hold trivially. Therefore, $(w)$ is an $\eta$-pattern for every $w \in A^*$.

We now introduce a notation. Let $\widebar{p} = (w_1,e_1,\dots,w_n,e_n,w_{n+1})$ be an  $\eta$-pattern. We write $\pi(\widebar{p})=w_1\cdots w_{n+1}\in A^*$. Also, we let $\uclos_{\eta} \widebar{p}\subseteq A^*$ be the following language:
\[
  \uclos_{\eta} \widebar{p} = w_1\ \eta\inv(e_1)\ \cdots\ w_n\ \eta\inv(e_n)\ w_{n+1}.
\]
When $\widebar{p}$ is of length $n = 0$, we let $\uclos_{\eta} \widebar{p}  = \{w_1\}$.

We turn to the statement. Given a morphism $\eta: A^* \to N$ and a \Jrel-class $J \subseteq N$, it can be used to build a cover of a language included in $\eta\inv(J)$. Yet, we need a strong hypothesis on $J$. We say that a \Jrel-class $J \subseteq N$ is $\eta$-alphabetic to indicate that for every $a \in A$, if $J \Jord \eta(a)$, then $\eta(a) \in J$.

\begin{restatable}{proposition}{core}\label{prop:core}
  Let $\eta: A^*\to N$ and $\alpha: A^*\to M$ be two morphisms, $J \subseteq N$ be an $\eta$-alphabetic \Jrel-class and $H \subseteq \eta\inv(J)$. Then, there exists a finite set $P$ of $\alpha$-guarded $\eta$-patterns such that $\pi(\widebar{p}) \in H$ for all $\widebar{p} \in P$ and $\{\uclos_{\eta} \widebar{p} \mid \widebar{p} \in P\}$ is a cover of~$H$.
\end{restatable}

\begin{remark}
  The proof of Proposition~\ref{prop:core} is not immediate. The core idea is to define an ordering on a suitable set of $\eta$-patterns and to show that this ordering is well. The proof is inspired by the ``minimal bad sequence'' argument of Nash-Williams~\cite{nash-williams63}. The set $P$ is then chosen as a set of minimal patterns for this ordering.
\end{remark}

Before proving Proposition~\ref{prop:core}, we introduce some terminology on $\eta$-patterns. Let $\eta: A^* \to N$ be a morphism and $\widebar{p}$ be an $\eta$-pattern. For $n \in \nat$, we define the \emph{$n$-splits} of $\widebar{p}$. Intuitively, they are tuples $(\widebar{p}_1,e_1,\dots,\widebar{p}_n,e_n,\widebar{p}_{n+1})$ where $e_1,\dots,e_n \in E(N)$ and $\widebar{p}_1,\dots,\widebar{p}_n$ are $\eta$-patterns, which are such that ``flattening'' $(\widebar{p}_1,e_1,\dots,\widebar{p}_n,e_n,\widebar{p}_{n+1})$ yields $\widebar{p}$. Let us formally define $n$-splits, by induction on $n$. If $n = 0$, $(\widebar{p})$ is the only $0$-split of $\widebar{p}$. If $n = 1$, we say that $(\widebar{p}_1,e_1,\widebar{p}_2)$ is a $1$-split of $\widebar{p}$ if $\widebar{p}_1 = (u_1,f_1,\dots,u_k,f_k,u_{k+1})$, $\widebar{p}_2= (v_1,g_1,\dots,v_\ell,g_\ell,v_{\ell+1})$ and,
\[
  \widebar{p} = (u_1,f_1,\dots,u_k,f_k,u_{k+1},e_1,v_1,g_1,\dots,v_\ell,g_\ell,v_{\ell+1}).
\]
Finally, when $n \geq 2$, we say that $(\widebar{p}_1,e_1,\dots,\widebar{p}_n,e_n,\widebar{p}_{n+1})$ is an $n$-split of $\widebar{p}$ when there exists an $\eta$-pattern $\widebar{q}$ such that $(\widebar{q},e_n,\widebar{p}_{n+1})$ is a $1$-split of $\widebar{p}$ and $(\widebar{p}_1,e_1,\dots,\widebar{p}_{n-1},e_{n-1},\widebar{p}_n)$ is an $(n-1)$-split of $\widebar{q}$.

Finally, a \emph{split} of $\widebar{p}$ is an $n$-split of $\widebar{p}$ for some $n \in \nat$. Abusing terminology, we often write $\widebar{p} = (\widebar{p}_1,e_1,\dots,\widebar{p}_n,e_n,\widebar{p}_{n+1})$ to indicate that $(\widebar{p}_1,e_1,\dots,\widebar{p}_n,e_n,\widebar{p}_{n+1})$ is a split of $\widebar{p}$. We may now recall the statement of Proposition~\ref{prop:core}.

When $\widebar{p}$ is an $\eta$-pattern, we abuse notation and write $\eta(\widebar{p})$ instead of $\eta(\pi(\widebar{p}))$. Clearly, when  $(\widebar{p}_1,e_1,\dots,\widebar{p}_n,e_n,\widebar{p}_{n+1})$ is a split of $\widebar{p}$, we have $\eta(\widebar{p})=\eta(\widebar{p}_1)\cdots\eta(\widebar{p}_n)\eta(\widebar{p}_{n+1})$.

\begin{proof}
  We work with a slightly different property than the one of the statement. Observe that given an $\eta$-pattern $\widebar{p}$, it may happen that $\pi(\widebar{p}) \not\in \uclos_{\eta} \widebar{p}$. We associate with $\widebar{p}$ an alternate language $\uclos_* \widebar{p}$ containing $\pi(\widebar{p})$. It will be easier to manipulate in the proof. Let $\widebar{p} = (w_1,e_1,\dots,w_n,e_n,w_{n+1})$ be an $\eta$-pattern. We write $\uclos_* \widebar{p} \subseteq A^*$ for the following language:
  \[
    \uclos_* \widebar{p} = w_1 (\eta\inv(e_1) \cup \{\veps\})\ \cdots\ w_n(\eta\inv(e_n) \cup \{\veps\})w_{n+1}.
  \]
  This definition ensures that $\pi(\widebar{p}) \in \uclos_* \widebar{p}$ for every $\eta$-pattern~$\widebar{p}$. We connect this notion to the original languages $\uclos_{\eta} \widebar{p}$ in the following fact.

  \begin{fact} \label{fct:concpat}
    Let $\widebar{p}$ be an $\alpha$-guarded $\eta$-pattern. There exists finitely many $\alpha$-guarded $\eta$-patterns $\widebar{p}_1,\dots,\widebar{p}_k$ such that $\pi(\widebar{p}_i) = \pi(\widebar{p})$ for every $i \leq k$ and $\uclos_* \widebar{p} = \bigcup_{0 \leq i \leq k} \uclos_{\eta} \widebar{p}_i$.
  \end{fact}

  \begin{proof}
    We write $\widebar{p} = (w_1,e_1,\dots,w_n,e_n,w_{n+1})$. If $n = 0$, then $\uclos_* \widebar{p} = \uclos_{\eta} \widebar{p} = \{w_1\}$. Hence, it suffices to chose $k = 1$ and $\widebar{p}_1 = \widebar{p}$. Assume now that $n \geq 1$. Let $I \subseteq \{1,\dots,n\}$. We associate an $\eta$-pattern $\widebar{p}_I$ that we build from $\widebar{p}$. Let $I = \{i_1,\dots,i_h\}$ with $i_1 < \cdots < i_h$. Moreover, let $i_0 = 0$ and $i_{h+1} = n+1$. For every $j$ such that $1\leq j\leq h+1$, we let $v_j = w_{i_{j-1}+1} \cdots w_{i_j}$. Finally, we define $\widebar{p}_I = (v_1,e_{i_1}, \dots,v_h,e_{i_h},v_{h+1})$. It is easy to check that $\pi(\widebar{p}_I)$ is an $\alpha$-guarded $\eta$-pattern such that $\pi(\widebar{p}_I) = \pi(\widebar{p})$. Moreover, since concatenation distributes over union, we obtain $\uclos_* \widebar{p} = \bigcup_{I \subseteq \{1,\dots,n\}} \uclos_{\eta} \widebar{p}_I$.
  \end{proof}

  We construct a finite set $P$ of $\alpha$-guarded $\eta$-patterns $\widebar{p}$ such that $\pi(\widebar{p}) \in H$ and $\{\uclos_* \widebar{p} \mid \widebar{p} \in P\}$ is a cover of $H$. It will then be straightforward to use Fact~\ref{fct:concpat} to construct the set described in Proposition~\ref{prop:core}, which will complete the proof.

  \smallskip

  Let us focus on the construction of~$P$. First, we identify a special \emph{infinite} set $Q$ of $\eta$-patterns. We shall then define the desired set $P$ as a subset of $Q$.  Let $k = 1+|M \times N|^2$. Then, $Q$ is the set of all $\alpha$-guarded $\eta$-patterns $(w_1,e_1,\dots,w_n,e_n,w_{n+1})$ such that $|w_i|\leq 2k$ for every $i\leq n+1$. We complete the definition with a key lemma. It implies that the $\eta$-patterns in $Q$ suffice to cover a language contained in $\eta\inv(J)$, such as $H$.

  \begin{lemma} \label{lem:wformed}
    Let $w \in \eta\inv(J)$. There is $\widebar{p} \in Q$ such that $w = \pi(\widebar{p})$.
  \end{lemma}

  \begin{proof}
    We prove will a stronger result. We define $R$ as the set of all $\eta$-patterns of the form $\widebar{p}=(w_1,e_1,\dots,w_n,e_n,w_{n+1}) \in Q$ that additionally satisfy $|w_{n+1}| \leq k$. Clearly, $R\subseteq Q$. We prove that for every $w\in\eta\inv(J)$, there exists $\widebar{p} \in R$ such that $w=\pi(\widebar{p})$. In the proof, we consider the monoid $M\times N$ equipped with the componentwise multiplication. We write $\gamma: A^*\to M\times N$ for the morphism defined by $\gamma(u) = (\alpha(u),\eta(u))$ for all $u \in A^*$.

    We proceed by induction on the length of $w$. If $|w| \leq k$, it suffices to choose $\widebar{p} = (w) \in R$. Assume now that $|w| > k$. We get $a_1,\dots,a_k \in A$ and $w' \in A^+$ such that $w = w'a_1 \cdots a_k$. Since $k > |M \times N|^2$, the pigeonhole principle yields $i,j$ such that $1 \leq i<j\leq k$, we have $\gamma(a_1\cdots a_i) =  \gamma(a_1\cdots a_j)$ and $\gamma(a_{i+1}\cdots a_k)=\gamma(a_{j+1}\cdots a_k)$. We now define $x=a_1 \cdots a_i$ and $y = a_{i+1} \cdots a_k$. By definition $|x|\leq k$ and $|y| \leq k$. Moreover, we have $w = w'xy$. Consider $(g,e) = (\gamma(a_{i+1} \cdots a_j))^\omega\in E(M \times N)$. By definition of~$i$~and~$j$, we have $\gamma(x)(g,e) = \gamma(x)$  and $(g,e)\gamma(y)=\gamma(y)$. Clearly, $w' \in \eta\inv(J)$ since $w'$ is a nonempty prefix of $w\in\eta\inv(J)$ and we know that $J$ is $\eta$-alphabetic. Since $|w'| < |w|$, induction yields an $\alpha$-guarded $\eta$-pattern $\widebar{p}' = (w_1,e_1,\dots,w_n,e_n,w_{n+1}) \in R$ such that $w' = \pi(\widebar{p}')$. Let $\widebar{p} =  (w_1,e_1,\dots,w_n,e_n,w_{n+1}x,e,y)$. Now, the tuple $\widebar{p}$ is an $\alpha$-guarded $\eta$-pattern. Indeed, Condition~\eqref{itm:pat1} holds by definition. Condition~\eqref{itm:pat2} holds because $w \in \eta\inv(J)$ and $J$ is $\eta$-alphabetic. Condition~\eqref{itm:pat3} and the fact that $\widebar{p}$ is $\alpha$-guarded hold by definition of the idempotent $(g,e) \in E(M \times N)$. Moreover, since $w = w'xy$ and $w' = \pi(\widebar{p}')$, it is immediate that $w = \pi(\widebar{p})$. Finally, $|w_i| \leq 2k$ for every $i \leq n$ and $|w_{n+1}| \leq k$ since $\widebar{p}'\in R$. Since $|x| \leq k$, we deduce that $|w_{n+1}x| \leq 2k$. Finally, we have $|y|\leq k$ by construction. Hence, $\widebar{p}\in R$, by definition of~$R$.
  \end{proof}

  We now define an \emph{ordering} on $Q$. Its definition is based on splits, introduced above. For $\widebar{p},\widebar{q} \in Q$, we write $\widebar{p} \preceq \widebar{q}$ when there exist $\ell \in \nat$, $\widebar{p}_1,\dots,\widebar{p}_{\ell+1} \in Q$,  $\widebar{q}_{1},\dots,\widebar{q}_{\ell} \in Q$ and $f_1,\dots,f_\ell \in E(N)$ satisfying the three following conditions:
  \begin{enumerate}
    \item for $1 \leq i \leq \ell$, we have $\eta(\widebar{q}_{i}) = f_i$,
    \item $\widebar{p} = (\widebar{p}_1,f_1,\ \dots\ ,\widebar{p}_{\ell},f_\ell,\widebar{p}_{\ell+1})$,
    \item $\widebar{q} = (\widebar{p}_1, f_1,\widebar{q}_{1},f_1,\ \dots\ ,\widebar{p}_{\ell},f_{\ell},\widebar{q}_{\ell},f_{\ell},\widebar{p}_{\ell+1})$.
  \end{enumerate}
  When $\ell = 0$, the definition simply states that $\widebar{p} = \widebar{p}_1 = \widebar{q}$. It can be verified that $\preceq$ is a partial order.

  \begin{fact} \label{fct:preord}
    Let $\widebar{p},\widebar{q} \in Q$ such that $\widebar{p} \preceq \widebar{q}$. Then, $\uclos_* \widebar{q} \subseteq \uclos_* \widebar{p}$.
  \end{fact}

  \begin{proof}
    Let $\ell \in \nat$, $\widebar{p}_1,\dots,\widebar{p}_{\ell+1} \in Q$,  $\widebar{q}_{1},\dots,\widebar{q}_{\ell} \in Q$ and $f_1,\dots,f_\ell$ be idempotents of $N$, as in the definition of $\widebar{p} \preceq \widebar{q}$. Now consider $w \in \uclos_* \widebar{q}$. By definition, this yields $u_i \in  \uclos_* \widebar{p}_i$ for  $i \leq \ell+1$, $v_i \in  \uclos_* \widebar{q}_{i}$ for $i \leq \ell$ and $x_i,y_i \in \eta\inv(f_i) \cup \{\veps\}$ for $i \leq \ell$ such that,
    \[
      w = u_1x_1v_1y_1 \cdots u_\ell x_\ell v_\ell y_\ell u_{\ell+1}.
    \]
    Since $v_i \in  \uclos_* \widebar{q}_{i}$ for all $i \leq \ell$, we have $\eta(v_i) = \eta(\widebar{q}_{i}) = f_i$ by hypothesis. It follows that $\eta(x_iv_iy_i) = f_i$. Therefore, we have $x_iv_iy_i \in \eta\inv(f_i)$. This implies that $w \in \uclos_* \widebar{p}$, since we also have $u_i \in  \uclos_* \widebar{p}_i$ for all $i \leq \ell+1$ and $\widebar{p} = (\widebar{p}_1,f_1,\ \dots\ , \widebar{p}_{\ell},f_\ell,\widebar{p}_{\ell+1})$.
  \end{proof}

  The proof of Proposition~\ref{prop:core} is based on the following lemma. It states that the ordering $\preceq$ on $Q$ is well.

  \begin{lemma}\label{lem:polcov:wellorder}
    Let $(\widebar{p}_i)_{i \in \nat}$ be an infinite sequence of $\eta$-patterns in $Q$. There exist $i,j \in \nat$ such that $i<j$ and $\widebar{p}_i \preceq \widebar{p}_j$.
  \end{lemma}

  Before we prove Lemma~\ref{lem:polcov:wellorder}, let us apply it to build our cover of the language $H$ and complete the main argument. Let $Q_H \subseteq Q$ be the set of all $\eta$-patterns $\widebar{p} \in Q$ such that $\pi(\widebar{p}) \in H$. We say that an $\eta$-pattern $\widebar{p} \in Q_H$ is \emph{minimal} if there exists no other $\eta$-pattern $\widebar{q} \in Q_H$ such that $\widebar{q} \preceq \widebar{p}$. Let then $F \subseteq Q_H$ be the set of all minimal $\eta$-patterns $\widebar{p} \in Q_H$. By definition, $\widebar{q} \not\preceq \widebar{p}$ for every $\widebar{p},\widebar{q} \in F$ such that $\widebar{p} \neq \widebar{q}$. Hence, Lemma~\ref{lem:polcov:wellorder} implies that $F$ is a finite set. We define $\Kb = \{\uclos_* \widebar{p} \mid \widebar{p} \in F\}$. Since $F \subseteq Q_H$, it is immediate by definition that $\pi(\widebar{p}) \in H$ for every $\widebar{p} \in F$. Therefore, it remains to prove that \Kb is a cover of $H$. We fix a word $w \in H$ and exhibit $K \in \Kb$ such that $w \in K$. Since $H \subseteq \eta\inv(J)$ by hypothesis, Lemma~\ref{lem:wformed} yields $\widebar{q} \in Q$ such that $\pi(\widebar{q})=w$. Moreover, since $w \in H$, we have $\widebar{q} \in Q_H$. Since $\preceq$ is a well-order, it has no strictly decreasing sequence. Therefore, there exists a minimal $\eta$-pattern $\widebar{p} \in Q_H$ such that $\widebar{p} \preceq \widebar{q}$. By definition, $\widebar{p} \in F$ and $\uclos_* \widebar{p} \in \Kb$. Since $\pi(\widebar{q}) = w$, we have $w \in \uclos_* \widebar{q}$. Finally, since $\widebar{p} \preceq \widebar{q}$, Fact~\ref{fct:preord} implies that $w \in \uclos_* \widebar{p}$ as well, which concludes the argument.

  \medskip It remains to prove Lemma~\ref{lem:polcov:wellorder} to finish the proof of Proposition~\ref{prop:core}. We do so by adapting the standard argument of Nash-Williams~\cite{nash-williams63}. We say that a sequence $(\widebar{p}_i)_{i \in \nat}$ of $\eta$-patterns is \emph{bad} if $\widebar{p}_i \in Q$ for every $i$ and $\widebar{p}_i \not\preceq \widebar{p}_j$ for every $i < j$. We have to prove that there exists no bad sequence. Proceeding by contradiction, assume that there exists a bad sequence. We first use this hypothesis to construct a specific one: we build by induction a specific sequence $(\widebar{p}_i)_{i \in \nat}$ such that for every $i \in \nat$, $\widebar{p}_0,\dots,\widebar{p}_i$ can be continued into a bad~sequence.

  Choose $\widebar{p}_0 \in Q$ such that $\widebar{p}_0$ can be continued into a bad sequence. Such an $\eta$-pattern must exist by the assumption that there exists a bad sequence. Assume now that $\widebar{p}_0,\dots,\widebar{p}_i$ have been chosen so that $\widebar{p}_0,\dots,\widebar{p}_i$ can be continued into a bad sequence. We define $\widebar{p}_{i+1} \in Q$ as an $\eta$-pattern of \emph{minimal length} such that $\widebar{p}_0,\dots,\widebar{p}_i,\widebar{p}_{i+1}$ can be continued into a bad sequence. This defines $(\widebar{p}_i)_{i \in \nat}$. Observe that it is necessarily bad. Indeed, otherwise, we would have $i<j$ such that $\widebar{p}_i \preceq \widebar{p}_j$, contradicting the hypothesis that $\widebar{p}_0,\dots,\widebar{p}_i$ can be continued into a bad sequence. In particular, observe that $(\widebar{p}_i)_{i \in \nat}$ contains $\eta$-patterns with arbitrarily large length (since, by definition of Q, there are only finitely many $\eta$-patterns of length $n$ in $Q$ for each $n \in \nat$).

  Let $\ell=|E(N) \times E(M)| \times |N|+1$. By the pigeonhole principle, there exist infinitely many indices $1 \leq i_1 < i_2 < \cdots$ such that $\widebar{p}_{i_k} = (\widebar{q},f,\widebar{r}_{i_k})$, where,
  \begin{itemize}
    \item $\widebar{q}\in Q$ is an $\eta$-pattern of length $\ell$,
    \item $f \in E(N)$ is an idempotent,
    \item $\widebar{r}_{i_k}$ is an $\eta$-pattern in $Q$  for each $k \geq 1$.
  \end{itemize}
  Let $\widebar{q} = (w_1,e_1,\dots,w_\ell,e_\ell,w_{\ell+1})$. Since $\widebar{q}$ belongs to $Q$, it is $\alpha$-guarded. Hence, we obtain idempotents $g_1,\dots,g_\ell \in E(M)$ such that for~all $1 \leq i \leq \ell$, we have $\eta\inv(e_i) \cap \alpha\inv(g_i) \neq \emptyset$, $\alpha(w_i)g_i = \alpha(w_i)$ and $g_i\alpha(w_{i+1}) = \alpha(w_{i+1})$. Now, since  $\ell=|E(N)\times E(M)|\times |N|+1$, the pigeonhole principle yields $|N|+1$ indices $h_1,\dots, h_{|N|+1} \leq \ell$ such that $e_{h_1} = \cdots = e_{h_{|N|+1}}$ and $g_{h_1} = \cdots = g_{h_{|N|+1}}$. We now use the pigeonhole principle again to get $i,j$ such that $1 \leq i < j \leq |N|+1$ and $\eta(w_1 \cdots w_{h_i}) = \eta(w_1 \cdots w_{h_j})$. We define,
  \[
    \widebar{q}' = (w_1,e_1,\dots,w_{h_i},e_{h_i},w_{h_j+1},e_{h_j+1},\dots,e_\ell,w_{\ell+1}).
  \]
  Since $\eta(w_1 \cdots w_{h_i}) = \eta(w_1 \cdots w_{h_j})$, $e_{h_i} = e_{h_j}$ and $g_{h_i} = g_{h_j}$, we get that $\widebar{q}'$ is an $\alpha$-guarded $\eta$-pattern in $Q$. Moreover $\eta(\widebar{q}) = \eta(\widebar{q}')$. Furthermore, it is clear that the length of $\widebar{q}'$ is strictly smaller than that of $\widebar{q}$ since $i < j$. Finally, observe that $\widebar{q}' \preceq \widebar{q}$. Indeed, consider $e = e_{h_i} = e_{h_j} \in E(N)$. We show that $\eta(w_{h_i+1} \cdots w_{h_j})= e$. This implies that $\widebar{q}' \preceq \widebar{q}$ by definition of $\preceq$. Let $s = \eta(w_1 \cdots w_{h_i})$ and $t = \eta(w_{h_i+1} \cdots w_{h_j})$. By definition of $i$ and $j$, we have $st = s$. The definition of $\eta$-patterns also yields $s \Jrel e$, $se = s$ and $et = t$.  The fact that $se = s$ also implies that $s \Lord e$ and we get $s \Lrel e$ from Lemma~\ref{lem:jlr}. This yields $x \in N$ such that $xs = e$. We get $t = et = xst = xs = e$, which completes the proof that $\widebar{q}' \preceq \widebar{q}$.

  For each $k \geq 1$, let $\widebar{p}_{i_k}' = (\widebar{q}',f,\widebar{r}_{i_k})$. One can check that $\widebar{p}_{i_k}'$ is indeed an $\eta$-pattern in~$Q$, because this is the case for $\widebar{p}_{i_k} = (\widebar{q},f,\widebar{r}_{i_k})$. We now prove that the infinite sequence  $\widebar{p}_0,\dots,\widebar{p}_{i_1-1},\widebar{p}_{i_1}',\widebar{p}_{i_2}',\widebar{p}_{i_3}',\dots$ is bad. In particular, this means that $\widebar{p}_0,\dots,\widebar{p}_{i_1-1},\widebar{p}_{i_1}'$ can be continued into a bad sequence. This is a contradiction: by definition of $\widebar{q}'$ from $\widebar{q}$, we know that the length of $\widebar{p}_{i_1}'$ is strictly smaller than the one of $\widebar{p}_{i_1}$. This is not possible since $\widebar{p}_{i_1} \in Q$ is defined as an $\eta$-pattern of \emph{minimal length} such that $\widebar{p}_0,\dots,\widebar{p}_{i_1-1},\widebar{p}_{i_1}$ can be continued into a bad sequence.

  It remains to prove that $\widebar{p}_0,\dots,\widebar{p}_{i_1-1},\widebar{p}_{i_1}',\widebar{p}_{i_2}',\widebar{p}_{i_3}',\dots$ is bad. Since $(\widebar{p}_i)_{i \in \nat}$ is bad, we know that for $i < j \leq i_1-1$, we have $\widebar{p}_i \not\preceq \widebar{p}_j$. We now prove that $\widebar{p}_i \not\preceq \widebar{p}_{i_k}'$ for $i \leq i_1-1$ and $k \geq 1$. By contradiction, assume that $\widebar{p}_i \preceq \widebar{p}_{i_k}'$. Since $\widebar{q}' \preceq \widebar{q}$, we have $\widebar{p}_{i_k}' \preceq \widebar{p}_{i_k}$. Hence, we get $\widebar{p}_i \preceq \widebar{p}_{i_k}$, contradicting the hypothesis that $(\widebar{p}_i)_{i \in \nat}$ is bad.  Finally, we show that given $h,k$ such that $1 \leq h < k$, we have $\widebar{p}_{i_h}' \not\preceq \widebar{p}_{i_k}'$. By contradiction assume that $\widebar{p}_{i_h}' \preceq\widebar{p}_{i_k}'$. By definition, this means that $(\widebar{q}',f,\widebar{r}_{i_h})\preceq(\widebar{q}',f,\widebar{r}_{i_k})$. Hence, the definition of $\preceq$ implies that there are two cases.
  \begin{itemize}[leftmargin=*]
    \item First, it may happen that $\widebar{r}_{i_h} \preceq  \widebar{r}_{i_k}$. Since $\widebar{p}_{i_h} = (\widebar{q},f,\widebar{r}_{i_h})$ and $\widebar{p}_{i_k} = (\widebar{q},f,\widebar{r}_{i_h})$, it follows that $\widebar{p}_{i_h} \preceq \widebar{p}_{i_k}$. Again, this contradicts the hypothesis that $(\widebar{p}_i)_{i \in \nat}$  is bad.
    \item Otherwise, there are $\widebar{u},\widebar{v} \in Q$ such that $\widebar{q}' \preceq (\widebar{q}',f,\widebar{u})$, $\widebar{r}_{i_h} \preceq \widebar{v}$ and $\widebar{r}_{i_k} = (\widebar{u},f,\widebar{v})$. Observe that $\eta(\widebar{u}) = f$. Indeed, we have $f \Jrel \eta(\widebar{q}')$,  $\eta(\widebar{q}') f = \eta(\widebar{q}')$ and $f \eta(\widebar{u}) = \eta(\widebar{u})$ by definition of $\eta$-patterns. The fact that $\eta(\widebar{q}') f = \eta(\widebar{q}')$ implies $\eta(\widebar{q}') \Lord f$ and Lemma~\ref{lem:jlr} yields $\eta(\widebar{q}') \Lrel f$. We get $s \in N$ such that $s\eta(\widebar{q}') = f$. Finally, since we have $\widebar{q}' \preceq (\widebar{q}',f,\widebar{u})$, we obtain $\eta(\widebar{q}') = \eta(\widebar{q}') \eta(\widebar{u})$. Altogether, we get $\eta(\widebar{u}) = f \eta(\widebar{u}) = s\eta(\widebar{q}')\eta(\widebar{u}) = s\eta(\widebar{q}') = f$.

          Since $f = \eta(\widebar{u})$ and $\widebar{r}_{i_h} \preceq \widebar{v}$, the definition of $\preceq$ implies that $(\widebar{q},f,\widebar{r}_{i_h}) \preceq (\widebar{q},f,\widebar{u},f,\widebar{v})$. Since $\widebar{r}_{i_k} = (\widebar{u},f,\widebar{v})$, this exactly says that $(\widebar{q},f,\widebar{r}_{i_h}) \preceq (\widebar{q},f,\widebar{r}_{i_k})$. That is, we have $\widebar{p}_{i_h} \preceq \widebar{p}_{i_k}$. Again, this contradicts the hypothesis that $(\widebar{p}_i)_{i \in \nat}$  is bad.
  \end{itemize}
  This completes the proof.
\end{proof}

\subsection{Concatenation hierarchies}

Let \Cs be a \vari. For all $n \in \nat$, we define two classes \polp{n}{\Cs} and \bpolp{n}{\Cs}. We let $\polp{0}{\Cs}\! =\! \bpolp{0}{\Cs}\! =\! \Cs$. Then, for every $n \geq 1$, we define inductively $\polp{n}{\Cs} = \pol{\bpolp{n-1}{\Cs}}$ and $\bpolp{n}{\Cs} = \bpol{\bpolp{n-1}{\Cs}} = \bool{\polp{n}{\Cs}}$. We say that the class \bpolp{n}{\Cs} is \emph{level $n$ in the concatenation hierarchy of basis \Cs}. Note that we have $\bpolp{n}{\Cs} \subseteq \polp{n+1}{\Cs} \subseteq \bpolp{n+1}{\Cs}$ for $n \in \nat$.

\begin{restatable}{remark}{hallevels} \label{rem:halflevels}
  We are mainly interested in ``full'' levels \bpolp{n}{\Cs}. However, the classes \polp{n}{\Cs} are important tools in the paper.
\end{restatable}

\noindent
Historically, the first concatenation hierarchy is the dot-depth (Brzozowski, Cohen~\cite{BrzoDot}) of basis $\dotzer = \{\emptyset,\{\veps\},A^+,A^*\}$. Thomas~\cite{ThomEqu} proved that the dot-depth corresponds exactly to the \emph{quantifier alternation hierarchy} of first-order logic with the linear order and the successor ($\fo(<,+1)$). More precisely, $\polp{n}{\dotzer}= \sic{n}(<,+1)$ and $\bpolp{n}{\dotzer} = \bsc{n}(<,+1)$ for all $n \in \nat$. This logical characterization extends to arbitrary concatenation hierarchies~\cite{PZ:generic18}: for every \vari~\Cs, there exists a set \infsigc of first-order predicates such that $\polp{n}{\Cs}= \sic{n}(\infsigc)$ and $\bpolp{n}{\Cs} = \bsc{n}(\infsigc)$ for all $n \in \nat$.

\smallskip
\noindent {\bf Group based hierarchies.} \emph{\Varis} consisting of group languages are prominent bases. \emph{Group languages} are those that can be recognized by a morphism into a finite \emph{group}. A \emph{group \vari} is a \vari \Gs that contains only group languages. We use bases that are group \varis as well as slight extensions thereof. For each class \Gs, we let $\Gs^+$ be the class consisting of all languages $\{\veps\} \cup L$ and $A^+ \cap L$ for $L \in \Gs$. One may verify that if \Gs is a \vari, then so is $\Gs^+$. Also, $\Gs^+$ contains the languages $\{\veps\}$ and $A^+$.

Let us present three standard group \varis. The trivial \vari $\stzer = \{\emptyset,A^*\}$ yields a hierarchy introduced by Straubing~\cite{StrauConcat} and Thérien~\cite{TheConcat}. It characterizes the alternation hierarchy of first-order logic equipped with only the linear order ($\fo(<)$). Also, $\stzer^+= \dotzer$: the hierarchy of basis $\stzer^+$ is the dot-depth. Another example is the \vari \md of \emph{modulo languages} (membership of a word in the language depends only on its length modulo some fixed integer). The hierarchies of bases \md and $\md^+$ characterize the alternation hierarchies of $\fo(<,MOD)$ and $\fo(<,+1,MOD)$, where ``\emph{MOD}'' is the set of modular predicates, see \emph{e.g.}, \cite{ChaubardPS06}. Finally, let us mention the \emph{group hierarchy} of Margolis and Pin~\cite{MargolisP85}. Its basis is the class \grp of \emph{all} group languages.

There are many papers investigating membership, separation and covering for particular hierarchies of this kind (see \emph{e.g.}, \cite{simonthm,knast83,henckell:hal-00019815,pwdelta2,ChaubardPS06,cmmptsep,pzjacm19,Zetzsche18}). Yet, proving that they are decidable for all levels remains a longstanding open question in automata theory. In the paper, we shall use the following generic theorem of~\cite{pzconcagroup,PlaceZ22}, which applies to level one in an arbitrary group based hierarchy.

\begin{restatable}{theorem}{grpgen}\label{thm:grpgen}
  Let \Gs be a group \vari with decidable separation. Then, \bpol{\Gs} and \bpol{\Gs^+} have decidable covering.
\end{restatable}

This result can be pushed further for the Straubing-Thérien hierarchy and the dot-depth. The following is proved in~\cite{pzbpolcj}.

\begin{restatable}{theorem}{stdot}\label{thm:stdot}
  \bpolp{2}{\stzer} and \bpolp{2}{\dotzer} have decidable covering.
\end{restatable}

We prove a generic characterization of \bpolo. It implies that if a \emph{non-zero} full level has decidable \emph{covering}, then the next full level has decidable \emph{membership}. By Theorem~\ref{thm:grpgen}, it follows that if \Gs is a group \vari with decidable separation, then \bpolp{2}{\Gs} and \bpolp{2}{\Gs^+} have decidable membership. Also, Theorem~\ref{thm:stdot} implies that  \bpolp{3}{\stzer} and \bpolp{3}{\dotzer} have decidable membership.

\section{\texorpdfstring{\Cs-pairs, \Cs-sets and \Cs-swaps}{C-pairs, C-sets and C-swaps}}
\label{sec:pairs}
We present tools used in our characterization. With a class~\Cs and a morphism $\alpha: A^* \to M$, we associate three notions: the \Cs-\emph{pairs} in $M^2$, the \Cs-\emph{sets} in $2^M$ and the \emph{\Cs-swaps} in $M^6$. 

\subsection{\texorpdfstring{The \Cs-pairs associated to a morphism}{The C-pairs associated to a morphism}}

Let \Cs be  a class and $\alpha: A^* \to M$ be  a morphism. We say that $(s,t) \in M^2$ is a \emph{\Cs-pair for $\alpha$} when $\alpha\inv(s)$ is \emph{not} \Cs-separable from $\alpha\inv(t)$. When $\alpha$ is understood, we simply say that $(s,t)$ is a \Cs-pair.  Since the set of \Cs-pairs is a subset of $M^2$, it is finite. Also, having a \Cs-separation algorithm in hand suffices to compute all \Cs-pairs associated with an input morphism.

\begin{restatable}{lemma}{septopairs} \label{lem:septopairs}
  Let \Cs a class with decidable separation. Given a morphism $\alpha: A^*\to M$, one can compute all \Cs-pairs for $\alpha$.
\end{restatable}

The \Cs-pairs for $\alpha$ define a \emph{binary relation} on $M$, which is reflexive when $\alpha$ is surjective (this is always the case in the paper). Also, it is symmetric if \Cs is closed under complement. On the other hand, the \Cs-pairs relation is \emph{not} transitive in~general.

We turn to a simple but crucial lemma. When \Cs is a \emph{\pvari}, it characterizes the \Cs-pairs in terms of \Cs-morphisms. Let $\eta: A^* \to (N,\leq)$ and $\alpha: A^* \to M$ be two morphisms. We say that $(s,t) \in M^2$ is an \emph{$\eta$-pair} for~$\alpha$ when there exist $u, v \in A^*$ such that \mbox{$\eta(u) \leq \eta(v)$}, $s = \alpha(u)$ and $t = \alpha(v)$. The statement is as~follows.

\begin{restatable}{lemma}{pairmor} \label{lem:pairmor}
	Let \Cs be a \pvari and let $\alpha: A^* \to M$ be a morphism. For every \Cs-morphism $\eta: A^* \to (N,\leq)$, each \Cs-pair for $\alpha$ is also an $\eta$-pair. Moreover, there exists a \Cs-morphism $\eta: A^* \to (N,\leq)$ such that the \Cs-pairs for $\alpha$ are exactly the $\eta$-pairs.
\end{restatable}

\begin{proof}
  For the first assertion, let $\eta: A^* \to (N,\leq)$ be a \Cs-morphism and let $(s,t) \in M^2$ be a \Cs-pair for $\alpha$. Let $F \subseteq N$ be the set of all elements $r \in N$ such that $\eta(u) \leq r$ for some $u \in \alpha\inv(s)$. By definition, $F$ is an upper set for the ordering $\leq$ on $N$. Hence, we have $\eta\inv(F) \in \Cs(A)$ since $\eta$ is a \Cs-morphism. Moreover, it is immediate from the definition of $F$ that $\alpha\inv(s) \subseteq \eta\inv(F)$. Since $(s,t)$ is a \Cs-pair (which means that $\alpha\inv(s)$ cannot be separated from $\alpha\inv(t)$ by a language of \Cs), it follows that $\eta\inv(F) \cap \alpha\inv(t) \neq \emptyset$. This yields $v \in A^*$ such that $\eta(v) \in F$ and $\alpha(v) = t$. Finally, the definition of $F$ yields $u \in A^*$ such that $\eta(u) \leq \eta(v)$ and $\alpha(u) = s$, concluding the proof of the first assertion.

  We turn to the second assertion. Let $P \subseteq M^2$ be the set of all pairs $(s,t) \in M^2$ that are \emph{not} \Cs-pairs. For each $(s,t) \in P$, there exists $K_{s,t} \in \Cs(A)$ which separates $\alpha\inv(s)$ from $\alpha\inv(t)$. Proposition~\ref{prop:genocm} yields a \Cs-morphism $\eta: A^* \to (N,\leq)$ such that every language $K_{s,t}$ for $(s,t) \in P$ is recognized by $\eta$. It remains to prove that for every $u,v \in A^*$, if $\eta(u) \leq \eta(v)$, then $(\alpha(u),\alpha(v))$ is a \Cs-pair. We prove the contrapositive. Assuming that $(\alpha(u),\alpha(v))$ is a \emph{not} a \Cs-pair, we show that $\eta(u)\not\leq \eta(v)$. By hypothesis, $(\alpha(u),\alpha(v)) \not\in P$ which means that $K_{\alpha(u),\alpha(v)} \in \Cs(A)$ is defined and separates $\alpha\inv(\alpha(u))$ from $\alpha\inv(\alpha(v))$. Thus, $u \in K_{\alpha(u),\alpha(v)}$ and $v \not\in K_{\alpha(u),\alpha(v)}$. Since $K_{\alpha(u),\alpha(v)}$ is recognized by $\eta$, this implies $\eta(u) \not\leq \eta(v)$, concluding the proof.
\end{proof}

We also have the following key corollary of Lemma~\ref{lem:pairmor}:  the \Cs-pair relation is compatible with multiplication.

\begin{corollary} \label{cor:pairmult}
  Let \Cs be a \pvari, let $\alpha: A^* \to M$ be a morphism and let $(s_1,t_1), (s_2,t_2) \in M^2$ be two \Cs-pairs. Then $(s_1s_2,t_1t_2)$ is also a \Cs-pair.

\end{corollary}

The \Cs-pairs are used to characterize several operators. We shall use them to handle \bpolo. In the proof, we shall also need the characterization of \polo proved in~\cite[Theorem~54]{PZ:generic18}.

\begin{restatable}{theorem}{polcar}\label{thm:polcar}
  Consider a \pvari \Cs and a surjective morphism $\alpha: A^* \to (M,\leq)$. Then, $\alpha$ is a \pol{\Cs}-morphism if and only if the following property holds:
  \begin{equation}\label{eq:polcar}
    s^{\omega+1} \leq s^\omega t s^\omega \text{ for all \Cs-pairs $(s,t) \in M^2$ for $\alpha$}.
  \end{equation}
\end{restatable}

Theorem~\ref{thm:polcar} implies that \pol{\Cs} has decidable membership when \Cs has decidable separation. We also state a useful corollary.

\begin{restatable}{corollary}{polcpairs} \label{cor:polcpairs}
  Let \Cs be a \pvari and $\alpha: A^* \to M$ a morphism. For every $e \in E(M)$ and $s \in M$, if $(e,s)$ is a \Cs-pair for $\alpha$, then $(e,ese)$ is a\/ \pol{\Cs}-pair for $\alpha$.
\end{restatable}

\begin{proof}
  It follows from Lemma~\ref{lem:pairmor} that there exists a \pol{\Cs}-morphism $\gamma: A^* \to (Q,\leq)$ such that the \pol{\Cs}-pairs for $\alpha$ are exactly the $\gamma$-pairs for $\alpha$. Furthermore, we obtain by Lemma~\ref{lem:pairmor} a \Cs-morphism $\eta: A^* \to (N,\leq)$ such that the \Cs-pairs for $\gamma$ are exactly the $\eta$-pairs for $\gamma$.

  Since $(e,s)$ is a \Cs-pair for $\alpha$ and since $\eta$ is a \Cs-morphism, we may apply Lemma~\ref{lem:pairmor} again to obtain $u,v \in A^*$ such that $\eta(u) \leq \eta(v)$, $\alpha(u) = e$ and $\alpha(v) = s$. Let $q = \gamma(u)$ and $r = \gamma(v)$. By definition, $(q,r) \in Q^2$ is an $\eta$-pair for $\gamma$, which means that it is a \Cs-pair by definition of $\eta$. Since $\gamma$ is a \pol{\Cs}-morphism, Theorem~\ref{thm:polcar} yields $q^{\omega+1} \leq q^\omega r q^\omega$, \emph{i.e.}, $\gamma(u^{k+1}) \leq \gamma(u^kvu^k)$ for $k = \omega(Q)$. By definition of $\gamma$, it follows that $(\alpha(u^{k+1}),\alpha(u^kvu ^k))$ is a \pol{\Cs}-pair. Since $e \in E(M)$, this exactly says that $(e,ese)$ is a \pol{\Cs}-pair for $\alpha$, as~desired.
\end{proof}

\subsection{\texorpdfstring{The \Cs-sets associated to a morphism}{The C-sets associated to a morphism}}

Roughly speaking, \Cs-sets are the generalization of \Cs-pairs adapted to covering. We only use \Cs-sets when \Cs is a \emph{Boolean algebra}. In this case, by Lemma~\ref{lem:covbopl}, we may consider the variant in which the language that needs to be covered is $A^*$.

Consider a class \Cs and a morphism $\alpha: A^* \to M$. We say that a subset $T \subseteq M$ is a \emph{\Cs-set} for $\alpha$ if and only if the set $\{\alpha\inv(t) \mid t \in T\}$ is \emph{not} \Cs-coverable. As for \Cs-pairs, the set of all \Cs-sets for $\alpha$ is finite: it is a subset of $2^M$. Furthermore, having a \Cs-covering algorithm in hand is clearly enough to compute all \Cs-sets for an input morphism~$\alpha$.

\begin{restatable}{lemma}{covtosets}\label{lem:covtosets}
  Let \Cs be a class with decidable covering. Given a morphism $\alpha: A^*\to M$, one can compute all \Cs-sets for~$\alpha$.
\end{restatable}

Let us present a few properties. First, we have the following fact which is immediate from the definition of covering.

\begin{restatable}{fact}{setinc} \label{fct:setinc}
  Let \Cs be a class, $\alpha: A^* \to M$ be a morphism and $T \subseteq M$ be a \Cs-set for $\alpha$. Every $S \subseteq T$ is also a \Cs-set for $\alpha$.
\end{restatable}

When \Cs is a \vari, we can characterize the \Cs-sets associated to a morphism with a single \Cs-morphism. Consider two morphisms $\eta: A^* \to N$ and $\alpha: A^* \to M$. We say that $T \subseteq M$ is an \emph{$\eta$-set} for $\alpha$ to indicate that there exists $F  \subseteq  A^*$ such that $\eta(w) = \eta(w')$ for every $w,w' \in F$ and $T = \alpha(F)$.

\begin{restatable}{lemma}{setmor} \label{lem:setmor}
	Let \Cs be a \vari  and $\alpha: A^* \to M$ be a morphism. For every \Cs-morphism $\eta: A^* \to N$, each \Cs-set for $\alpha$ is also an $\eta$-set. There exists a \Cs-morphism $\eta: A^* \to N$ such that the \Cs-sets for $\alpha$ are exactly the $\eta$-sets.
\end{restatable}

\begin{proof}
  For the first assertion, we fix a \Cs-set $T \subseteq M$ and a \Cs-morphism $\eta: A^* \to N$. We prove that $T$ is an $\eta$-set. We define $\Kb=\{\eta\inv(s) \mid r \in N\}$. Since $\eta$ is a \Cs-morphism, \Kb is a \Cs-cover of $A^*$. Hence, since $T$ is a \Cs-set (which means that $\{\alpha\inv(t) \mid t \in T\}$ is \emph{not} \Cs-coverable), there exists $K\in\Kb$ such that $K \cap \alpha\inv(t) \neq \emptyset$ for every $t \in T$. Thus, for each $t \in T$, we get $w_t \in K \cap \alpha\inv(t)$. Let $F=\{w_t\mid t\in T\}$. By definition $T=\alpha(F)$. Finally, by definition of \Kb, we have $K= \eta\inv(r)$ for some $r\in N$. Hence, $\eta(w) = r$ for every $w \in F$, which completes the proof of the first assertion.

  We turn to the second assertion. We first introduce some terminology: given two sets of languages \Kb and \Lb, we say that \Kb is \emph{separating} for \Lb when for every $K\in\Kb$, there exists $L\in\Lb$ such that $K\cap L=\emptyset$. Note that by definition, \Lb is \Cs-coverable if and only if there exists a \Cs-cover of $A^*$ that is separating for~\Lb.

  Let now $P \subseteq 2^M$ be the set of all subsets $T$ of $M$ that are \emph{not} \Cs-sets. By definition, for every $T \in P$, there exists a \Cs-cover $\Kb_{T}$ of $A^*$ that is separating for $\{\alpha\inv(t) \mid t \in T\}$. Proposition~\ref{prop:genocm} yields a \Cs-morphism $\eta: A^* \to N$ recognizing all languages $K\in \Kb_{T}$, for $T \in P$. It remains to prove that if $T \subseteq M$ is an $\eta$-set for $\alpha$, then it is a \Cs-set for $\alpha$ (the converse implication is immediate from the first assertion). By contradiction, assume that $T$ is \emph{not} a \Cs-set, \emph{i.e.}, that $T \in P$. By definition of covering, it follows that $T \neq \emptyset$. Also, since $T \in P$, there exists a \Cs-cover $\Kb_{T}$ of $A^*$ that is separating for $\{\alpha\inv(t) \mid t \in T\}$. Moreover, since $T$ is an $\eta$-set, we get $F\subseteq A^*$ such that $\eta(w)=\eta(w')$ for every $w,w'\in F$ and $T=\alpha(F)$. Since $T\neq\emptyset$, we have $F \neq \emptyset$ as well. Since $\Kb_{T}$ is a cover of $A^*$, we get $K \in \Kb_{T}$ such that $K \cap F \neq \emptyset$. Thus, since $K$ is recognized by $\eta$ and $\eta(w) = \eta(w')$ for every $w,w' \in F$, it follows that $F \subseteq K$. Since $T = \alpha(F)$, we obtain $K \cap \alpha\inv(t) \neq \emptyset$ for every $t \in T$. It follows that $\Kb_{T}$ is \emph{not} separating for $\{\alpha\inv(t) \mid t \in T\}$, contradicting our hypothesis.
\end{proof}

\subsection{\texorpdfstring{The \Cs-swaps associated to a morphism}{The C-swaps associated to a morphism}}

This last relation is \emph{ad hoc}. It is not based on a ``Boolean problem'' such as separation or covering. We directly use \Cs-morphisms.

Given two morphisms $\alpha: A^* \to M$ and $\eta: A^* \to N$, we define tuples $(q,r,s,t,e,f) \in M^6$ called the $\eta$-swaps~for~$\alpha$. When using these objects later, the elements $e,f\in M$ will be idempotents (hence the notation). However, this is not required for the definition. We say that $(q,r,s,t,e,f) \in M^6$ is an \emph{$\eta$-swap for $\alpha$} when there exist words $u_q \in \alpha\inv(q)$, $u_r \in \alpha\inv(r)$, $u_s \in \alpha\inv(s)$, $u_t \in \alpha\inv(t)$, $u_e \in \alpha\inv(e)$ and $u_f \in \alpha\inv(f)$ satisfying the two following conditions:
\begin{enumerate}
  \item The elements $\eta(u_e)$ and $\eta(u_f)$ are idempotents of $N$.
  \item We have $\eta(u_ru_eu_s) = \eta(u_f)$ and $\eta(u_qu_fu_t) = \eta(u_e)$.
\end{enumerate}
Consider a \vari \Cs. We define the \Cs-swaps for $\alpha$ as the tuples $(q,r,s,t,e,f) \in M^6$ that are $\eta$-swaps for \emph{every \Cs-morphism $\eta: A^* \to N$}. We conclude this section with a technical result.

\begin{restatable}{lemma}{witness} \label{lem:witness}
  Let \Cs be a \vari and let $\alpha: A^* \to M$ be a morphism. There xs a \Cs-morphism $\eta: A^* \to N$ such that the \Cs-swaps for $\alpha$ are exactly the $\eta$-swaps and the \Cs-pairs for $\alpha$ are exactly the $\eta$-pairs.
\end{restatable}

\begin{proof}
  Let $P \subseteq M^6$ be the set of all tuples that are \emph{not} a \Cs-swap and $\{\widebar{p}_1,\dots,\widebar{p}_k\} = P$. By definition, we know that for each $i \leq k$, there exists a \Cs-morphism $\eta_i: A^* \to N_i$ such that $\widebar{p}_i \in M^6$ is not an $\eta_i$-swap. Moreover, Lemma~\ref{lem:pairmor} yields a \Cs-morphism $\gamma: A^* \to Q$ such that the \Cs-pairs are exactly the $\gamma$-pairs. Consider the monoid $N_1 \times \cdots N_k \times Q$ equipped with the componentwise multiplication. Let now $\beta: A^* \to N_1 \times \cdots N_k \times Q$ be the morphism defined by $\beta(w)= (\eta_1(w),\dots,\eta_k(w),\gamma(w))$ for $w \in A^*$. Finally, let $\eta: A^* \to N$ be the surjective restriction of $\beta$. One may now verify that $\eta$ satisfies the desired properties.
\end{proof}

\section{General characterization}
\label{sec:bpolgen}
In this section, we characterize the \bpol{\Cs}-morphisms.

\begin{restatable}{theorem}{cargen} \label{thm:cargen}
  Let \Cs be a \vari and let $\alpha: A^* \to M$ be a surjective morphism. Then, $\alpha$ is a \bpol{\Cs}-morphism if and only if the two following equations hold:
  \begin{alignat}{1}
    \begin{array}{c}
      (eset)^{\omega+1} = (eset)^{\omega}et(eset)^{\omega} \text{ for every $t \in M$}\\
      \text{and every \Cs-pair $(e,s) \in E(M) \times M$}.
    \end{array}\label{eq:cacp} \\
    \begin{array}{c}
      (eqfre)^\omega (esfte)^\omega = (eqfre)^\omega qft (esfte)^\omega \text{ for every} \\
      \text{\Cs-swap $(q,r,s,t,e,f) \in M^6$ such that $e,f \in  E(M)$}.
    \end{array} \label{eq:swap}
  \end{alignat}
\end{restatable}

Theorem~\ref{thm:cargen} does \emph{not} have immediate general consequences for \bpol{\Cs}-membership. It implies that the problem is decidable when both \Cs-pairs and \Cs-swaps can be computed. On the one hand, the calculation of \Cs-pairs boils down to \Cs-separation. On the other hand, we do not know whether the computation of \Cs-swaps reduces to a natural problem. This is why we cannot deduce a general corollary from Theorem~\ref{thm:cargen}. Yet, we prove in Section~\ref{sec:uptwo} that if \Cs is itself of the form \bpol{\Ds} for a \vari \Ds, then~\eqref{eq:swap} is equivalent to an equation depending on \emph{\Cs-sets} only (which can be computed when \Cs-covering is decidable, by Lemma~\ref{lem:covtosets}).

\begin{remark}
  Equation~\eqref{eq:cacp} was previously used in~\cite{pzupol2}: alone, it characterizes the larger class \capolp{2}{\Cs} consisting of all languages $L$ such that $L \in \polp{2}{\Cs}$ and $A^* \setminus L \in \polp{2}{\Cs}$.
  The inclusion $\bpol{\Cs}\subseteq\capolp{2}{\Cs}$ implies that all \bpol{\Cs}-morphisms indeed satisfy~\eqref{eq:cacp}. We shall give a direct proof in the paper.
\end{remark}

\subsection{A first application: the class \bpol{\at}}

We may already get an effective characterization for a key class: level \emph{two} in the Straubing-Thérien hierarchy (\bpolp{2}{\stzer}). Obtaining such a characterization remained an important open problem for many years until it was finally solved in~\cite{pzqalt}. Theorem~\ref{thm:cargen} yields a new simpler characterization of \bpolp{2}{\stzer}. Let \at be the class of \emph{alphabet testable languages}: the Boolean combinations of languages $A^* a A^*$ where $a \in A$ is a letter. Pin and Straubing~\cite{pin-straubing:upper} proved that $\bpolp{2}{\stzer} = \bpol{\at}$ (see also~\cite{PZ:generic18} for a more recent proof).  Hence, we apply Theorem~\ref{thm:cargen} in the special case when $\Cs = \at$.

For $w \in A^*$, let $\cont{w} \subseteq A$ be the set of all letters occurring in~$w$ (\emph{i.e.}, the least $B\subseteq A$ such that $w \in B^*$). It is folklore and simple to verify that a morphism $\eta: A^* \to N$ is an \at-morphism if and only if $\cont{u} = \cont{v} \Rightarrow \eta(u) = \eta(v)$ for all $u,v \in A^*$.  With this result in hand, one can check that for \Cs= \at, Equation~\eqref{eq:swap} is equivalent to Equation~\eqref{eq:bpolat2} below, which uses \at-sets instead of \at-swaps. More precisely, Theorem~\ref{thm:cargen} yields the following corollary.

\begin{restatable}{corollary}{bpolat} \label{cor:bpolat}
	Let $\alpha: A^* \to M$ be a surjective morphism. Then, $\alpha$ is a \bpol{\at}-morphism if and only if the following equations hold:
	\begin{alignat}{1}
		\begin{array}{c}
			(eset)^{\omega+1} = (eset)^{\omega}et(eset)^{\omega} \text{ for every $t \in M$}\\
			\text{and every \at-pair $(e,s) \in E(M) \times M$}.
		\end{array}\label{eq:bpolat1} \\
		\begin{array}{c}
			(eqfre)^\omega (esfte)^\omega = (eqfre)^\omega qft (esfte)^\omega \text{ for every}\\
			\text{\at-set $\{q,r,s,t,e,f\} \subseteq M$ such that $e,f \in  E(M)$}.
		\end{array} \label{eq:bpolat2}
	\end{alignat}
\end{restatable}

Thus, \bpol{\at}-membership boils down to \at-covering, which suffices to compute \at-pairs and \at-sets. Deciding \at-covering is simple: a finite set of languages \Lb is \emph{not} \at-coverable if and only if there is a sub-alphabet $B \subseteq A$ such that each $L \in \Lb$ contains a word $w$ satisfying $\cont{w} = B$. Hence, the class $\bpolp{2}{\stzer} = \bpol{\at}$ is decidable. This reproves the result of~\cite{pzqalt}.

\begin{remark}
	Interestingly,  the properties in Corollary~\ref{cor:bpolat} appeared in previous papers (though precisely establishing the link involves quite a bit of technical work). Equation~\eqref{eq:bpolat2} is equivalent to a property that Pin and Weil~\cite{pwdelta2,pwmalcev} conjectured to be a characterization of \bpolp{2}{\stzer}. Equation~\eqref{eq:bpolat1} is equivalent to an upper bound of \bpolp{2}{\stzer} proved by Almeida and Kl\'ima~\cite{AK2010} (the link is made with Lemma~\ref{lem:eqformu} below which reformulates~\eqref{eq:bpolat1}). Actually, it was conjectured in~\cite{AK2010} that the two put together characterize \bpolp{2}{\stzer} (a conjecture proved by Corollary~\ref{cor:bpolat}).
\end{remark}

\noindent
We now turn to the proof of Theorem~\ref{thm:cargen}.

\subsection{Proof of Theorem~\ref{thm:cargen}: ``only if'' direction}

Assume that $\alpha$ is a \bpol{\Cs}-morphism. We show that it satisfies~\eqref{eq:cacp} and~\eqref{eq:swap}. That it satisfies~\eqref{eq:cacp} follows from~\cite[Theorem~6.7]{pzupol}, which states that~\eqref{eq:cacp} characterizes the class $\polp{2}{\Cs}\cap \textup{co-}\polp{2}{\Cs}$ consisting of all languages $L$ such that both $L$ and $A^*\setminus L$ belong to $\polp{2}{\Cs}$. Now by definition, $\bpol{\Cs} \subseteq \polp{2}{\Cs}\cap \textup{co-}\polp{2}{\Cs}$. Let us provide a direct~proof.

Since \Cs is a \vari, \pol\Cs is a \pvari by Theorem~\ref{thm:polc}. Therefore, we may apply Lemma~\ref{lem:bpolm}. It yields a \pol{\Cs}-morphism $\gamma: A^* \to (Q,\leq)$ as well as a morphism $\beta: Q \to M$ such that $\alpha = \beta \circ \gamma$. Moreover, Lemma~\ref{lem:pairmor} yields a \Cs-morphism $\eta: A^*\to N$ such that the $\eta$-pairs for $\gamma$ are exactly the \Cs-pairs for $\gamma$. We let $n = \omega(M) \times \omega(N) \times \omega(Q)$.

\medskip
\noindent
{\bf Equation~\eqref{eq:cacp}.}  We fix a \Cs-pair $(e,s) \in E(M) \times M$ for $\alpha$ and $t \in M$. We have to show that $(eset)^{\omega+1} = (eset)^\omega et (eset)^\omega$. Since $\eta$ is a \Cs-morphism, Lemma~\ref{lem:pairmor} yields $u,v \in A^*$ such that $\eta(u) = \eta(v)$, $\alpha(u) = e$ and $\alpha(v) = s$. Let $x = u^n$ and $y = u^{n-1}v$. We have $\gamma(x) \in E(Q)$, $\alpha(x) = e$ and $\alpha(y) = es$. Let $f = \gamma(x) \in E(Q)$ and $q = \gamma(y)$. Note that $\beta(f) = e$ and $\beta(q) = es$. We also let $r \in Q$ such that $\beta(r) = t \in M$. Since $\eta(x) = \eta(y)$, we know that $(f,q)$ is a \Cs-pair for $\gamma$ by definition of $\eta$.  Since $\gamma$ is a \pol{\Cs}-morphism, this yields $f \leq fqf$ by Theorem~\ref{thm:polcar}. Thus,
\[
  (fqfr)^\omega fr (fqfr)^\omega \leq (fqfr)^{\omega+1}
\]
Conversely, since $(f,q)$ is a \Cs-pair, we may apply Corollary~\ref{cor:pairmult}. It implies that $(fr,fqfr)$ is a \Cs-pair for $\gamma$ as well. Thus, Theorem~\ref{thm:polcar} yields the following inequality:
\[
  (fqfr)^{\omega+1} \leq (fqfr)^\omega fr (fqfr)^\omega.
\]
Altogether, we get $(fqfr)^{\omega+1} = (fqfr)^\omega fr (fqfr)^\omega$. We now apply $\beta$  to get~\eqref{eq:cacp}, as desired.

\medskip
\noindent
{\bf Equation~\eqref{eq:swap}.} We fix a \Cs-swap $(q,r,s,t,e,f) \in M^6$ with $e,f \in E(M)$. We have to prove that $(eqfre)^\omega (esfte)^\omega = (eqfre)^\omega qft (esfte)^\omega$.

Since $\eta$ is a \Cs-morphism, the definition of a \Cs-swap yields $u_x \in \alpha\inv(x)$ for $x \in \{q,r,s,t,e,f\}$ with $\eta(u_e),\eta(u_f) \in E(N)$, $\eta(u_ru_eu_s) = \eta(u_f)$ and $\eta(u_qu_fu_t) = \eta(u_e)$. Let $y_e = u_e^n$ and $y_f=u_f^n$. By definition, we have $\gamma(y_e),\gamma(y_f)\in E(Q)$, $\alpha(y_e) = e$ and $\alpha(y_f) = f$. Finally, since $\eta(u_e), \eta(u_f) \in E(N)$, we also know that $\eta(y_e) = \eta(u_e)$ and $\eta(y_f) = \eta(u_f)$. In particular, we obtain that $\eta(u_ry_eu_s) = \eta(y_f)$ and $\eta(u_qy_fu_t) = \eta(y_e)$. We let,
\[
  \begin{array}{lll}
    x & = & (y_eu_qy_fu_ry_e)^n (y_eu_sy_fu_ty_e)^n, \\
    y & = & (y_eu_qy_fu_ry_e)^n u_qy_fu_t (y_eu_sy_fu_ty_e)^n.
  \end{array}
\]
It is enough to prove that $\gamma(x)=\gamma(y)$. Indeed, since $\alpha=\beta\circ\gamma$, this will yield $\alpha(x) = \alpha(y)$, which is exactly~\eqref{eq:swap}. Let $\hat{e} = \gamma(y_e)$, $\hat{f} = \gamma(y_f)$, $\hat{q} = \gamma(u_q)$, $\hat{r} = \gamma(u_r)$, $\hat{s} = \gamma(u_s)$ and $\hat{t} = \gamma(u_t)$. Since $\eta(y_e) = \eta(u_qy_fu_t)$, the definition of $\eta$ implies that $(\hat{e},\hat{q}\hat{f}\hat{t})$ is a \Cs-pair for $\gamma$. Hence, since $\gamma$ is a \pol{\Cs}-morphism, Theorem~\ref{thm:polcar} yields $\hat{e} \leq \hat{e}\hat{q}\hat{f}\hat{t}\hat{e}$. We may now multiply by $(\hat{e}\hat{q}\hat{f}\hat{r}\hat{e})^n$ on the left and $(\hat{e}\hat{s}\hat{f}\hat{t}\hat{e})^n$ on the right to obtain,
\[
  (\hat{e}\hat{q}\hat{f}\hat{r}\hat{e})^n(\hat{e}\hat{s}\hat{f}\hat{t}\hat{e})^n \leq (\hat{e}\hat{q}\hat{f}\hat{r}\hat{e})^n \hat{q}\hat{f}\hat{t}(\hat{e}\hat{s}\hat{f}\hat{t}\hat{e})^n.
\]
Hence, $\gamma(x) \leq \gamma(y)$. Moreover, the definitions yield,
\[
  \eta(y_f) = \eta(u_ry_e(y_e u_qy_fu_ry_e)^{n-1}(y_eu_sy_fu_ty_e)^{n-1}y_eu_s).
\]
Thus, $(\hat{f},\hat{r}\hat{e}(\hat{e}\hat{q}\hat{f}\hat{r}\hat{e})^{n-1}(\hat{e}\hat{s}\hat{f}\hat{t}\hat{e})^{n-1}\hat{e}\hat{s})$ is a \Cs-pair for $\gamma$.  Since $\gamma$ is a \pol{\Cs}-morphism and $\hat{f} \in E(Q)$, Theorem~\ref{thm:polcar} then yields the inequality $\hat{f} \leq \hat{f}\hat{r}\hat{e}(\hat{e}\hat{q}\hat{f}\hat{r}\hat{e})^{n-1}(\hat{e}\hat{s}\hat{f}\hat{t}\hat{e})^{n-1}\hat{e}\hat{s}\hat{f}$. We may now multiply it by $(\hat{e}\hat{q}\hat{f}\hat{r}\hat{e})^n\hat{e}\hat{q}$ on the left and $\hat{t}\hat{e}(\hat{e}\hat{s}\hat{f}\hat{t}\hat{e})^n$ on the right. Since $n$ is a multiple of $\omega(Q)$, this yields,
\[
  (\hat{e}\hat{q}\hat{f}\hat{r}\hat{e})^n \hat{q}\hat{f}\hat{t}(\hat{e}\hat{s}\hat{f}\hat{t}\hat{e})^n \leq (\hat{e}\hat{q}\hat{f}\hat{r}\hat{e})^n(\hat{e}\hat{s}\hat{f}\hat{t}\hat{e})^n.
\]
Hence, $\gamma(y) \leq \gamma(x)$. We have proved $\gamma(x) = \gamma(y)$, as desired.

\subsection{Proof of Theorem~\ref{thm:cargen}: ``if'' direction, preliminaries}

We now concentrate on the more difficult ``if'' implication. Assume that $\alpha$ satisfies~\eqref{eq:cacp} and~\eqref{eq:swap}. We~prove that it is a \bpol{\Cs}-morphism. We first reformulate~\eqref{eq:cacp} using \pol{\Cs}-pairs instead of \Cs-pairs. This follows from~\cite[Theorems 6.4 and~6.7]{pzupol2}. Here, we give a direct proof. 

\begin{restatable}{lemma}{eqformu} \label{lem:eqformu}
  Let \Cs be a \vari and let $\alpha: A^* \to M$ be a surjective morphism. Then, $\alpha$ satisfies~\eqref{eq:cacp} if and only if it satisfies the following:
  \begin{equation} \label{eq:capolcp}
    s^{\omega+1} = s^\omega t s^\omega \text{ for all \pol{\Cs}-pairs $(t,s) \in M^2$}.
  \end{equation}
\end{restatable}

\begin{proof}
  We start with the ``if'' direction. Assume that $\alpha$ satisfies~\eqref{eq:capolcp}. We show that~\eqref{eq:cacp} holds as well. Let $(e,s) \in M^2$ be a \Cs-pair and $t \in M$. We have to show that,
  \[
    (eset)^{\omega+1} = (eset)^{\omega+1} et (eset)^{\omega}.
  \]
  Corollary~\ref{cor:polcpairs} implies that $(e,ese)$ is a \pol{\Cs}-pair. It then follows from Corollary~\ref{cor:pairmult} that $(et,eset)$ is a \pol{\Cs}-pair. The desired equality is now immediate from~\eqref{eq:capolcp}.

  \smallskip
  We turn to the ``only if'' direction of Lemma~\ref{lem:eqformu}, which requires more work. Assume that $\alpha$ satisfies~\eqref{eq:cacp}. We prove that~\eqref{eq:capolcp} holds. First, we overapproximate the \pol{\Cs}-pairs using a preorder. We let $\preceq$ be the least preorder on  $M$ such that for every $e \in E(M)$ and every $s,x,y \in M$ such that $(e,s)$ is a \Cs-pair, we have $xey \preceq xesey$. We have the following lemma.

  \begin{restatable}{lemma}{trclos} \label{lem:trclos}
    For all \pol{\Cs}-pairs $(q,r) \in M^2$, we have $q \preceq r$.
  \end{restatable}

  \begin{proof}
    One can verify that $\preceq$ is a precongruence on $M$, \emph{i.e.}, if $q \preceq r$ and $q' \preceq r'$, then $qq' \preceq rr'$. Let $\sim$ be the equivalence on $M$ generated by $\preceq$: we have $q \sim r$ if and only if $q \preceq r$ and $r \preceq q$. Since $\preceq$ is a precongruence, $\sim$ is a congruence. Therefore, the quotient $N = {M}/{\sim}$ is a monoid and the ordering induced by $\preceq$ on $N$ makes $(N,\leq)$ an ordered monoid. We let $\delta: M \to N$ be the morphism associating its equivalence class to every element in $M$. Finally, we let $\eta=\delta\circ\alpha: A^*\to(N,\leq)$. We use Theorem~\ref{thm:polcar} to prove that $\eta$ is a \pol{\Cs}-morphism. For a  \pol{\Cs}-pair $(q,r) \in M^2$ for $\alpha$, it will follow that from Lemma~\ref{lem:pairmor} that there exist $u,v \in A^*$ such that $\eta(u) \leq \eta(v)$, $\alpha(u) = q$ and $\alpha(v)= r$. By definition of $\eta$, this yields $\delta(q) \leq \delta(r)$. Hence, $q \preceq r$ as desired.

    It remains to prove that $\eta$ is a \pol{\Cs}-morphism. In view of Theorem~\ref{thm:polcar}, we have to show that for all \Cs-pairs $(p,t) \in N^2$ for $\eta$, we have $p^{\omega+1} = p^{\omega}tp^\omega$. Since $(p,t)$ is a \Cs-pair, one may verify that there exists a \Cs-pair $(q,r) \in M^2$ for $\alpha$ such that $\delta(q) = p$ and $\delta(r) = t$. Let $k = \omega(M)$. It follows from Corollary~\ref{cor:pairmult} that $(q^k,q^{k-1}r)$ is a \Cs-pair for $\alpha$. Since $q^k \in E(M)$, it follows that $q^k \preceq q^{2k-1}rq^k$ by definition. We can now multiply by $q$ on the left which yields $q^{\omega+1} \preceq q^{\omega} r q^{\omega}$. We apply $\delta$ on both sides to get $p^{\omega+1} = p^{\omega}tp^\omega$ as desired.
  \end{proof}

  We are now ready to conclude the proof of Lemma~\ref{lem:eqformu}. We have to prove that $\alpha$ satisfies~\eqref{eq:capolcp}. Let $(t,s) \in M^2$ be a \pol{\Cs}-pair for $\alpha$. We prove that $s^{\omega+1} = s^\omega t s^\omega$. Lemma~\ref{lem:trclos} yields $t \preceq s$. By definition of $\preceq$, this yields $s_0,\dots,s_n \in M$ such that $s_0=s$, $s_n=t$ and for $0 < i \leq n$, we have $s_{i-1} =xeqey$ and $s_i = xey$ for some $e \in E(M)$ and $q,x,y \in M$ such that $(e,q)$ is a \Cs-pair. We use induction on $i$ to prove that $s^{\omega+1} = s^\omega s_i s^\omega$ for every $i \leq n$. The case $i = n$ will then yield the desired result.

  When $i = 0$, the result is trivial since $s_0 = s$. Assume now that $i \geq 1$. By definition, we have $s_{i-1} =xeqey$ and $s_i = xey$ for some $e \in E(M)$ and $q,x,y \in M$ such that $(e,q)$ is a \Cs-pair. Moreover, $s^{\omega+1} = s^\omega s_{i-1} s^\omega$ by induction. Thus, $s^{\omega+1} = s^\omega xeqey s^\omega$. It follows that,
  \[
    \begin{array}{lll}
      s^{\omega+2} & = & (s^{\omega+1})^{\omega+2}, \\
              &= &(s^\omega xeqey s^\omega)^{\omega+2}, \\
              & = & s^\omega x (eqeys^\omega x)^{\omega+1}eqeys^\omega.
    \end{array}
  \]
  Since $(e,q)$ is a \Cs-pair, it follows from~\eqref{eq:cacp} that,
  \[
    (eqeys^\omega x)^{\omega+1} = (eqeys^\omega x)^{\omega} eys^\omega x(eqeys^\omega x)^{\omega}.
  \]
  Altogether, we obtain,
  \[
    \begin{array}{lll}
      s^{\omega+2} & = &s^\omega x (eqeys^\omega x)^{\omega} eys^\omega x(eqeys^\omega x)^{\omega}eqeys^\omega,\\
              & = &(s^\omega xeqeys^\omega)^\omega xey (s^\omega xeqey s^\omega)^{\omega+1}, \\
              & = & (s^\omega s_{i-1}s^\omega)^\omega s_i  (s^\omega s_{i-1} s^\omega)^{\omega+1}.
    \end{array}
  \]
  Since $s^{\omega+1} = s^\omega s_{i-1} s^\omega$, we get $s^{\omega+2} = s^\omega s_i s^{\omega+1}$. It now suffices to multiply by $s^{\omega-1}$ on the right to get $s^{\omega+1} = s^\omega s_i s^\omega$, which completes the proof.
\end{proof}

\newcommand{\mpr}[1]{\ensuremath{[#1]_M}\xspace}
\newcommand{\npr}[1]{\ensuremath{[#1]_N}\xspace}

We now introduce some terminology. Lemma~\ref{lem:witness} yields a \Cs-morphism $\eta: A^*\to N$ such that the $\eta$-swaps are exactly the \Cs-swaps and the $\eta$-pairs are exactly the \Cs-pairs. We let $Q = \{(\eta(w),\alpha(w)) \mid w \in A^*\} \subseteq N \times M$, which is a monoid for the componentwise multiplication. For $q=(t,s) \in Q$, we let $\npr{q}=t \in N$ and $\mpr{q} = s \in M$. Let $\tau: A^* \to Q$ be the morphism such that $\tau(w) = (\eta(w),\alpha(w))$ for all $w\in A^*$. Clearly, all languages recognized by $\alpha$ are also recognized by $\tau$. Thus, it suffices to show that $\tau$ is a \bpol{\Cs}-morphism. We first present a few simple properties of $\tau$. Using the fact that~$\alpha$ satisfies~\eqref{eq:swap}, one deduces the following lemma.

\begin{restatable}{lemma}{rswap} \label{lem:rswap}
  Let $q,r,s,t \in Q$ and $e,f \in E(Q)$ be elements such that $\npr{qft} = \npr{e}$ and $\npr{res} = \npr{f}$. Then, we have,
  \[
    (eqfre)^\omega (esfte)^\omega = (eqfre)^\omega qft (esfte)^\omega.
  \]
\end{restatable}

\begin{proof}
  Since $\npr{qft} = \npr{e}$ and $e \in E(Q)$ (which implies that $\npr{e} \in E(N)$), it is immediate that,
  \[
    \npr{(eqfre)^\omega (esfte)^\omega} = \npr{(eqfre)^\omega qft (esfte)^\omega}.
  \]
  Moreover, since the \Cs-swaps for $\alpha$ are exactly the $\eta$-swaps, the hypothesis that $\npr{qft} = \npr{e}$ and $\npr{res} = \npr{f}$ implies that $(\mpr{q},\mpr{r},\mpr{s},\mpr{t},\mpr{e},\mpr{f})$ is a \Cs-swap for $\alpha$. Hence, we obtain from~\eqref{eq:swap} that,
  \[
    \mpr{(eqfre)^\omega (esfte)^\omega} = \mpr{(eqfre)^\omega qft (esfte)^\omega}.
  \]
  Hence, we get $(eqfre)^\omega (esfte)^\omega = (eqfre)^\omega qft (esfte)^\omega$, as desired.
\end{proof}

We now reformulate~\eqref{eq:cacp} (or rather~\eqref{eq:capolcp}, which is equivalent by Lemma~\ref{lem:eqformu}). We say that a pair $(q,r)\in Q^2$ is a \emph{good pair} if and only if $(\mpr{q},\mpr{r})\in M^2$ is a \pol{\Cs}-pair for $\alpha$ and $\npr{q} = \npr{r}$. One can verify the following lemma from~\eqref{eq:capolcp}.

\begin{restatable}{lemma}{rcapolcp} \label{lem:rcapolcp}
  If $(q,r)\in Q^2$ is a good pair then $r^{\omega+1} = r^\omega q r^\omega$.
\end{restatable}

\begin{proof}
  By hypothesis and definition of good pairs, $\npr{q} = \npr{r}$  and $(\mpr{q},\mpr{r})$ is a \pol{\Cs}-pair for $\alpha$. The first equality yields $\npr{r^{\omega+1}} = \npr{r^\omega q r^\omega}$. Also, since $\alpha$ satisfies~\eqref{eq:cacp}, it also satisfies~\eqref{eq:capolcp} by Lemma~\ref{lem:eqformu}. Therefore, $\mpr{r^{\omega+1}} = \mpr{r^\omega q r^\omega}$ since $(\mpr{q},\mpr{r})$ is a \pol{\Cs}-pair. Altogether, we get $r^{\omega+1} = r^\omega q r^\omega$ as desired.
\end{proof}

By Corollary~\ref{cor:pairmult}, good pairs are stable under multiplication: if $(q,r),(q,'r') \in Q^2$ are good pairs, then so is $(qq',rr')$. We use this fact implicitly. Also, we have the following lemma.

\begin{restatable}{lemma}{bupol} \label{lem:bupol}
  Let $q \in Q$. There exists $K \in \pol{\Cs}$ such that $(q,r)$ is a good pair for all $r \in \tau(K)$ and $\tau\inv(q) \subseteq K \subseteq \eta\inv(\npr{q})$.
\end{restatable}

\begin{proof}
  Let $T \subseteq M$ be the set consisting of all $t \in M$ such that $(\mpr{q},t)$ is \emph{not} a \pol{\Cs}-pair for $\alpha$. By definition, for each $t \in T$, there exists a language $K_t \in \pol{\Cs}$ separating $\alpha\inv(\mpr{q})$ from $\alpha\inv(t)$. We define $K = \eta\inv(\npr{t}) \cap \bigcap_{t \in T} K_t$. We have $K \in \pol{\Cs}$ since $\eta$ is a \Cs-morphism and \pol{\Cs} is closed under intersection by Theorem~\ref{thm:polc}. One may now verify that $K$ satisfies the desired properties.
\end{proof}

Finally, the next lemma follows from Corollary~\ref{cor:polcpairs}.
\begin{restatable}{lemma}{rgpairs} \label{lem:rpairs}
  Let $e \in E(Q)$ and $r \in Q$. If $\npr{e} = \npr{r}$, then $(e,ere)$ is a good pair.
\end{restatable}

\begin{proof}
  Clearly, $\npr{e} = \npr{ere}$. Moreover, by definition of $\eta$, the hypothesis that $\npr{e} = \npr{r}$  implies that $(\mpr{e},\mpr{r})$ is a \Cs-pair for $\alpha$. Thus, Corollary~\ref{cor:polcpairs} yields that $(\mpr{e},\mpr{ere})$ is a \pol{\Cs}-pair for $\alpha$. Altogether, we obtain that $(e,ere)$ is a good pair.
\end{proof}

\subsection{Proof: main argument}

We may now start the proof. Given $q_1,q_2 \in Q$ and  \Kb a finite set of languages, we say that \Kb is \emph{$(q_1,q_2)$-safe} to indicate that for all $K \in\Kb$ and $u,u' \in K$, we have $q_1\tau(u)q_2 = q_1\tau(u')q_2$. We say that \Kb is \emph{safe} to indicate that it is $(1_Q,1_Q)$-safe.

\begin{restatable}{proposition}{genmain} \label{prop:genmain}
  For every \Jrel-class $J \subseteq N$, there exists a safe \bpol{\Cs}-cover of $\eta\inv(J)$.
\end{restatable}

We first explain how to conclude the proof using Proposition~\ref{prop:genmain}.
For each \Jrel-class $J \subseteq N$, Proposition~\ref{prop:genmain} yields a safe \bpol{\Cs}-cover $\Kb_J$ of $\eta\inv(J)$. The union \Kb of all $\Kb_J$ is a safe \bpol{\Cs}-cover of $A^*$. By definition of safeness, each $K \in \Kb$ satisfies $K \subseteq \tau\inv(q)$ for some $q \in Q$. Since \Kb is a cover of $A^*$, it follows that for each subset $F \subseteq Q$, $\tau\inv(F)$ is the union of all languages $K \in \Kb$ such that $K\subseteq \tau\inv(F)$. By closure under union, this yields $\tau\inv(F) \in \bpol{\Cs}$. Hence, $\tau$ is a \bpol{\Cs}-morphism, which completes the main proof.

\smallskip

We now focus on the proof of Proposition~\ref{prop:genmain}. Let $J$ be a \Jrel-class of $N$. We define the \emph{index} of $J$ as the number of \Jrel-classes $J'$ of $N$ such that $J\Jords J'$. We use induction on the index of $J$ to build a safe \bpol{\Cs}-cover of $\eta\inv(J)$. The construction involves a sub-induction on a second parameter. We define an ordering on $Q \times J \times Q$ as follows: given $(q_1,d,q_2),(r_1,d',r_2) \in Q \times J \times Q$ and $q \in Q$, we let $(q_1,d,q_2) \smarrow{q} (r_1,d',r_2)$ when there exist $s_1,s_2,s'_1,s'_2 \in Q$ and $e \in E(Q)$ such that,
\begin{enumerate}
  \item $\npr{s_1es_2} = d$ and $\npr{e} = \npr{s'_1}d'\npr{s'_2}$.
  \item $q = s_1es_2$, $r_1 = q_1s_1es'_1$ and $r_2 = s'_2es_2q_2$.
\end{enumerate}
For $(q_1,d,q_2),(r_1,d',r_2)\in Q \times J \times Q$, we let $(q_1,d,q_2) \leq (r_1,d',r_2)$ if $(q_1,d,q_2) = (r_1,d',r_2)$ or $(q_1,d,q_2) \smarrow{q} (r_1,d',r_2)$ for some $q \in Q$. Finally, we let $(q_1,d,q_2) < (r_1,d',r_2)$ when $(q_1,d,q_2) \leq (r_1,d',r_2)$ and $(r_1,d',r_2) \neq (q_1,d,q_2)$.

\begin{restatable}{fact}{transclos} \label{fct:transclos}
  The relation $\leq$ on $Q \times J \times Q$ is a preorder.
\end{restatable}

\begin{proof}
	By definition, $\leq$ is reflexive. Let us now prove that it is transitive. Let $(q_1,d,q_2)$, $(r_1,d',r_2)$, $(r'_1,d'',r'_2)$ be such that $(q_1,d,q_2) \leq (r_1,d',r_2)$ and  $(r_1,d',r_2) \leq (r'_1,d'',r'_2)$. We prove that $(q_1,d,q_2) \leq (r'_1,d'',r'_2)$. We have $q,r,s_1,s_1,s'_1,s'_2$, $t_1,t_2,t'_1,t'_1 \in Q$ and $e,f \in Q$ such that,
	\begin{enumerate}
		\item $\npr{s_1es_2} = d^{\ }$,\quad $\npr{e} = \npr{s'_1}d\npr{s'_2}$,\\ $\npr{t_1ft_2} = d'$,\quad $\npr{f} = \npr{t'_1}d''\npr{t'_2}$.
		\item  $q = s_1es_2$,\quad $r^{}_1 = q^{}_1s^{}_1es'_1$,\quad $r^{}_2 = s'_2es^{}_2q^{}_2$,\\ $r = t_1ft_2$,\quad $r'_1 = r^{}_1t^{}_1ft'_1$,\quad $r'_2 = t'_2ft^{}_2r^{}_2$.
	\end{enumerate}
  Let $s''_1 = s'_1r'_1t^{}_1ft'_1$ and $s''_2 = r'_2ft^{}_2s'_2$. We have $\npr{s_1es_2} = d$ and,
	\[
    \begin{array}{rlll}
      \npr{e}= & \npr{s'_1}d'\npr{s'_2} & \\
      = &\npr{s'_1t^{}_1ft^{}_2s'_2} & \text{as $\npr{t_1ft_2} = d'$,} \\
      = &\npr{s'_1t^{}_1ffft^{}_2s'_2} & \text{as $f \in E(M)$,}\\
      = & \npr{s'_1t^{}_1ft'_1}d''\npr{t'_2ft^{}_2s'_2} &\text{as $\npr{f}= \npr{t'_1}d''\npr{t'_2}$,}\\
      = & \npr{s''_1}d''\npr{s''_2}.
    \end{array}
	\]
	Furthermore, we have the equalities $q = s_1es_2$, $r'_1 = r^{}_1t^{}_1ft'_1 = q^{}_1s^{}_1es'_1t^{}_1ft'_1 = q^{}_1s^{}_1es''_1$ and $r'_2 = t'_2ft^{}_2r^{}_2 = t'_2ft^{}_2r^{}_2s'_2es^{}_2q^{}_2 = s''_2es^{}_2q^{}_2$. Altogether, we obtain $(q_1,d,q_2) \smarrow{q} (r'_1,d'',r'_2)$. Thus, we get $(q_1,d,q_2) \leq (r'_1,d'',r'_2)$.
\end{proof}

\newcommand{\barS}{\ensuremath{\smash{\widebar{S}}}\xspace}

We now start the construction. We define the \emph{rank} of a triple $(q_1,d,q_2) \in Q \times J \times Q$ as the number of triples $(r_1,d',r_2) \in Q \times J \times Q$ such that $(q_1,d,q_2)\leq (r_1,d',r_2)$. We use induction on the rank of $(q_1,d,q_2)\in Q \times J \times Q$ to construct a $(q_1,q_2)$-safe \bpol{\Cs}-cover of $\eta\inv(d)$. It will follow from the case $q_1=q_2=1_Q$ that there exists a safe \bpol{\Cs}-cover of $\Kb_d$ of~$\eta\inv(d)$ for every $d \in J$. Thus, $\Kb = \bigcup_{d \in J} \Kb_d$ is the desired safe \bpol{\Cs}-cover of $\eta\inv(J)$ described in Propostition~\ref{prop:genmain}.

We fix $(q_1,d,q_2) \in Q \times J \times Q$. We build our $(q_1,q_2)$-safe \bpol{\Cs}-cover of $\eta\inv(d)$ in two steps. We say that $q \in Q$ \emph{stabilizes $(q_1,d,q_2)$} if $(q_1,d,q_2) \smarrow{q}(q_1,d,q_2)$. We define:
\[
  \barS = \{q \in Q \mid \text{$\npr{q} = d$ and $q$ does \emph{not} stabilize $(q_1,d,q_2)$}\}.
\]
In the first step, we use induction to build a $(q_1,q_2)$-safe \bpol{\Cs}-cover $\Kb_{\barS}$ of $\tau\inv(\barS)\subseteq \eta\inv(d)$ (the base case of the sub-induction is when $\barS = \emptyset$: it suffices to choose $\Kb_{\barS} = \emptyset$). The key argument to construct this $(q_1,q_2)$-safe \bpol{\Cs}-cover is Proposition~\ref{prop:core}, which yields a \pol\Cs cover. In the second (and simpler) step, we complete $\Kb_{\barS}$ with a single language of \bpol\Cs to take into accounts the elements of $\eta\inv(d)$ stabilizing $(q_1,d,q_2)$, thus obtaining the desired safe \bpol{\Cs}-cover $\Kb_d$ of $\eta\inv(d)$.

The first step differs depending on whether $J \subseteq N$ is a \emph{maximal} \Jrel-class (\emph{i.e.}, $1_N \in J$) or not. Here, we treat the most involved case: we assume that $J$ is \emph{not} maximal. Thus, $1_N \not\in J$. 

\begin{restatable}{remark}{themax} \label{rem:themax}
  If $\{\veps\} \in \Cs$, we may assume without loss of generality that the \Cs-morphism $\eta: A^*\to N$ recognizes~$\{\veps\}$. In this case, the maximal \Jrel-class of $N$ is $\{1_N\}$ and $\eta\inv(1_N)=\{\veps\}$. This trivializes the case where $J$ is maximal: $\{\{\veps\}\}$ is a safe \bpol{\Cs}-cover of $\eta\inv(1_N)$. This is a key point as we are mainly interested in input classes \Cs which are themselves of the form $\Cs = \bpol{\Ds}$. In this case, $\{\veps\} \in \Cs$.
\end{restatable}

\subsection{\texorpdfstring{First case: $J$ maximal. Building \protect{$\Kb_{\protect\barS}$}}{First case: J maximal. Building KS}}

We prove Proposition~\ref{prop:genmain} in the special case where $J \subseteq N$ is the \emph{maximal} \Jrel-class of $N$. We first build a \pol{\Cs}-cover \Hb of $\tau\inv(\barS)$ with Proposition~\ref{prop:core}. Then, we look at each language $H \in \Hb$ and build a \bpol{\Cs}-cover of $H$ independently. The reunion of these covers will be the desired \bpol{\Cs}-cover $\Kb_{\barS}$ of $\tau\inv(\barS)$.

This case is simpler than the next one, when $J$ is not maximal. Indeed, the maximal \Jrel-class of $N$ is necessarily $\eta$-alphabetic. Let $w \in A^*$ and $a_1,\dots,a_n \in A$ be the letters such that $w = a_1 \cdots a_n$. We let,
\[
  \uclos w = \eta\inv(1_N)\ a_1\ \eta\inv(1_N)\ \cdots\ a_n\ \eta\inv(1_N) \subseteq A^*.
\]
In particular, we let $\uclos \veps = \eta\inv(1_N)$. The following lemma is a corollary of Proposition~\ref{prop:core}.

\begin{lemma} \label{lem:covbase}
  There is a finite set $F \subseteq \tau\inv(\barS)$ such that $\{\uclos w \mid w \in F\}$ is a cover of $\tau\inv(\barS)$.
\end{lemma}

\begin{proof}
  Since $J$ is maximal, it is $\eta$-alphabetic. Therefore, since $\tau\inv(\barS) \subseteq \eta\inv(d)$ and $d \in J$, it follows from Proposition~\ref{prop:core} that there is a finite set $P$ of $\eta$-patterns such $\pi(\bar{p}) \in \tau\inv(\barS)$ for every $\bar{p} \in P$ and the set $\{\uclos_{\eta} \bar{p} \mid \bar{p} \in P\}$ is a cover of $\tau\inv(\barS)$. Consider a $\bar{p} \in P$ and let $\bar{p} = (w_1,g_1,\dots,w_n,g_n,w_{n+1})$. Since $\pi(\bar{p}) \in \tau\inv(\barS)$, we have $\eta(\pi(\bar{p})) = d \in J$ which means that $g_1,\dots,g_n \in J$ by definition of $\eta$-patterns. Moreover, since $J$ is the maximal \Jrel-class of $N$, the neutral element $1_N$ is the only idempotent in $J$. Hence, $g_1 = \cdots = g_n = 1_N$. Since $\veps \in \eta\inv(1_N)$, this yields $\uclos_{\eta} \bar{p} \subseteq \uclos \pi(\bar{p})$. Consequently, we have $\{\uclos \pi(\bar{p}) \mid \bar{p} \in P\}$ is also a cover of $\tau\inv(\barS)$. We may now choose $F = \{\pi(\bar{p}) \mid \bar{p} \in P\}$ which completes the proof.
\end{proof}

We fix the set $F \subseteq \tau\inv(\barS)$ described in Lemma~\ref{lem:covbase} for the remainder of the proof. For each $w \in F$, we construct a $(q_1,q_2)$-safe \bpol{\Cs}-cover $\Kb_w$ of $\uclos w$. In view of Lemma~\ref{lem:covbase}, it will then be immediate that $\Kb_{\barS} = \bigcup_{w \in F} \Kb_w$ is the desired $(q_1,q_2)$-safe \bpol{\Cs}-cover of $\tau\inv(\barS)$. We fix $w \in F$. Let $a_1,\dots,a_n \in A$ be the letters such that $w = a_1 \cdots a_n$. We have $\uclos w = \eta\inv(1_N) a_1 \eta\inv(1_N)\cdots a_n \eta\inv(1_N)$.  We first cover the language $\eta\inv(1_N)$ using induction on the rank of $(q_1,d,q_2)$ in the following lemma.

\begin{restatable}{lemma}{indrank} \label{lem:indrankb}
  There exists a tight \bpol{\Cs}-cover \Vb of $\eta\inv(1_N)$ that is $(r_1,r_2)$-safe for all $r_1,r_2 \in Q$ such that $(q_1,d,q_2) < (r_1,1_N,r_2)$.
\end{restatable}

\begin{proof}
  Let $X \subseteq Q^2$ be the set of all pairs $(r_1,r_2) \in M^2$ such $(q_1,d,q_2) < (r_1,1_N,r_2)$. For all $(r_1,r_2)\in X$, the rank of $(r_1,1_N,r_2)$ is strictly smaller than the one of $(q_1,d,q_2)$. Hence, we may use  induction on the rank of $(q_1,d,q_2)$ to get a $(r_1,r_2)$-safe \bpol{\Cs}-cover $\Vb_{r_1,r_2}$ of $\eta\inv(1_N)$.  Let $\ell = |X|$ and $\{(r_{1,1},r_{2,1}),\dots,(r_{1,\ell},r_{2,\ell})\} = X$. We define,
  \[
    \Vb = \{V_1 \cap \cdots \cap V_\ell \cap \eta\inv(1_N) \mid \text{$V_j \in \Vb_{r_{1,j},r_{2,j}}$ for all $j \leq \ell$}\}.
  \]
  Since $\eta$ is a \Cs-morphism and \bpol{\Cs} is closed under intersection, it is immediate that $\Vb$ is tight \bpol{\Cs}-cover of $\eta\inv(1_N)$. One may verify that it satisfies the desired property.
\end{proof}

Lemma~\ref{lem:bpconcatm} yields a \bpol{\Cs}-cover $\Kb_w$ of $\uclos w$ such that for every language $K \in \Kb_w$, we have $K \subseteq V_0a_1V_1 \cdots a_nV_n$ where $V_i \in \Vb_{i}$ for all $i$. It remains to prove that $\Kb_w$ is $(q_1,q_2)$-safe.  Let $K \in \Kb_w$. Given $u,u' \in K$, we prove that $q_1\tau(u)q_2 = q_1\tau(u')q_2$.

By definition, we have $K \subseteq V_0a_1V_1 \cdots a_nV_n$ with $V_i \in \Vb_{i}$ for~all~$i$. Since $u,u' \in K$, we obtain two words $v_i,v'_i \in V_i$ for all $i$, such that $u = v_0a_1v_1 \cdots a_nv_{n}$ and $u' = v'_0a_1v'_1 \cdots a_nv'_{n}$.  For every $i$ such that $0 \leq i \leq n$, we define,
\[
  x_i = v_0a_1 \cdots v_{i-1}a_i \text{ and } y_i = a_{i+1}v_{i+1} \cdots a_nv_n.
\]
In particular,  $x_0 = y_n = \veps$. Observe that $u = x_nv_ny_n$, $u' = x_1v'_1y_1$ and $x_iv'_iy_i = x_{i-1}v_{i-1}y_{i-1}$ for all~$i$. We prove that $q_1\tau(x_iv_iy_i)q_2 = q_1\tau(x_iv'_iy_i)q_2$ for all $i$. Transitivity will then yield $q_1\tau(x_nv_ny_n)q_2 = q_1\tau(x_1v'_1y_1)q_2$. This exactly expresses the equality $q_1\tau(u)q_2 = q_1\tau(u')q_2$, as desired. We now fix $i$ and prove that $q_1\tau(x_iv_iy_i)q_2 = q_1\tau(x_iv'_iy_i)q_2$. We write $s_1 = \tau(x_i)$ and $s_2 = \tau(v_i)$. The proof is based on the following lemma.

\begin{lemma} \label{lem:indbase}
  We have $(q_1,d,q_2) < (q_1s_1,1_N,s_2q_2)$.
\end{lemma}

By definition of $\Vb_i$ in Lemma~\ref{lem:indrankb}, it follows from Lemma~\ref{lem:indbase} that $\Vb_i$ is a $(q_1s_1,s_2q_2)$-safe \bpol{\Cs}-cover of $\eta\inv(1_N)$. Thus, as $V_i \in \Vb_i$ and $v_i,v'_i \in V_i$, we get $q_1s_1\tau(v_i)s_2q_2= q_1s_1\tau(v'_i)s_2q_2$, \emph{i.e.} $q_1\tau(x_iv_iy_i)q_2 = q_1\tau(x_iv'_iy_i)q_2$ as desired.

\begin{proof}[Proof of Lemma~\ref{lem:indbase}]
  We first show the (non-strict) inequality $(q_1,d,q_2) \leq (q_1s_1,1_N,s_2q_2)$. We may use $s'_1=s'_2=e=1_Q$ in the definition of $\leq$. Thus, it suffices to show that $\npr{s_1s_2} = d$. Recall that the sets $\Vb_j$ are \emph{tight} covers of $\eta\inv(1_N)$. Thus, $\eta(v_j) = \eta(v'_j) = 1_N$ for all $j$. It follows that $\eta(x_iy_i) = \eta(a_1 \cdots a_n) = \eta(w)$. Thus, since $w \in F \subseteq \tau\inv(\barS)$, the definition of \barS yields $\eta(x_iy_i) = d$. Altogether, we get $\npr{s_1s_2} = d$.

  It remains to prove that the inequality is strict. By contradiction, assume that this is not the case, \emph{i.e.}, that $(q_1s_1,1_N,s_2r_2) \leq (q_1,d,q_2)$. Let $q = \tau(w)$. We prove that $q$ stabilizes $(q_1,d,q_2)$. This is a contradiction since $w \in F$ which means that $q = \tau(w) \in \barS$. By hypothesis, we get $p_1,p_2,p'_1,p'_2$ and $f \in E(Q)$ such that,
  \begin{enumerate}
    \item $\npr{p_1fp_2} = 1_N$ and $\npr{f} = \npr{p'_1}d\npr{p'_2}$.
    \item $q = p_1fp_2$, $q_1 = q_1\alpha(x_i)p_1ep'_1$ and $q_2 = p'_2ep_2\alpha(y_i)q_2$.
  \end{enumerate}
  Let $t_1 = \tau(a_1 \cdots a_i)$, $t'_1 = p_1ep'_1(s_1p_1ep'_1)^\omega$, $t_2 = \tau(a_{i+1} \cdots a_n)$ and $t'_2 = (p'_2ep_2s_2)^\omega p'_2ep_2$. We prove that we have $q_1 = q_1t_1t'_1$ and $q_2 = t'_2t_2q_2$. Let us first explain why this implies that $q$ stabilizes $(q_1,d,q_2)$.

  We have $\npr{t_1t_2} =\eta(w) = d$ since $w \in F$. Furthermore, using the equalities $\npr{p_1fp_2} = 1_N$, $\npr{f} = \npr{p'_1}d\npr{p'_2}$ and $\npr{s_1s_2} = d$, we deduce,
  \[
    \begin{array}{lll}
      1_N & = & \npr{(p_1fp'_1s_1)^{\omega+1}}1_N\npr{(s_2p'_2fp_2)^{\omega+1}}, \\
          & = & \npr{p_1fp'_1(s_1p_1fp'_1)^\omega } \npr{s_1s_2} \npr{(p'_2fp_2s_2)^{\omega}p'_2fp_2}, \\
          & = &  \npr{t'_1}d\npr{t'_2}.
    \end{array}
  \]
  Furthermore, $q =  \tau(w) = \tau(x_iy_i) =s_1s_2$. Thus, if $q_1 = q_1t_1t'_1$ and $q_2 = t'_2t_2q_2$., we get $(q_1,d,q_2) \smarrow{q} (q_1,d,q_2)$ as desired.

  It remains to show that $q_1 = q_1t_1t'_1$ and $q_2 = t'_2t_2q_2$. By symmetry, we only detail the former. We already established that $\eta(v_j) = 1_N$ for all $j$. Hence, $\npr{\tau(v_j)} = 1_N = \npr{1_Q}$ and it follows from Lemma~\ref{lem:rpairs} that $(1_Q,\tau(v_j))$ is a good pair. Since $t_1 = \tau(a_1 \cdots a_i)$ and $s_1 = \tau(v_0a_1 \cdots v_{i-1}a_i)$, closure under multiplication yields that $(t_1p_1ep'_1,s_1p_1ep'_1)$ is a good pair. We now apply Lemma~\ref{lem:rcapolcp} to get,
  \[
    (s_1p_1ep'_1)^{\omega+1} = (s_1p_1ep'_1)^{\omega}t_1p_1ep'_1 (s_1p_1ep'_1)^{\omega}.
  \]
  We multiply by $q_1$ on the left. As $q_1 = q_1s_1p_1ep'_1$, this yields $q_1 = q_1t_1p_1ep'_1 (s_1p_1ep'_1)^{\omega}$ which exactly says that $q_1 = q_1t_1t'_1$, completing the proof.
\end{proof}

\subsection{\texorpdfstring{Second case: $J$ \emph{non}-maximal. Building \protect{$\Kb_{\protect\barS}$}}{Second case: J not maximal, building KS}}
\label{sec:first-step}

To build the $(q_1,q_2)$-safe \bpol{\Cs}-cover of $\tau\inv(\barS)$, the argument goes roughly as follows. In a first step, we build a \pol{\Cs}-cover \Hb of $\tau\inv(\barS)$ with Proposition~\ref{prop:core}. In a second step, we look independently at each language $H \in \Hb$ and build a $(q_1,q_2)$-safe \bpol{\Cs}-cover of~$H$. The union of all these covers is the desired \bpol{\Cs}-cover $\Kb_{\barS}$ of~$\tau\inv(\barS)$. The hard part is building the \bpol{\Cs}-cover of $H$. For this, the intuition is the following: by definition of \pol\Cs, the language $H\in\pol\Cs$ is a concatenation of basic languages in $\Cs$. We use induction to construct a \bpol\Cs-cover of each basic language. Then, we use the weak concatenation principles (Lemmas~\ref{lem:bpconcat} and ~\ref{lem:bpconcatm}) to combine these covers into the desired \bpol\Cs cover of~$H$.

We start with the first step, building a \pol{\Cs}-cover \Hb of $\tau\inv(\barS)$. Unfortunately, we cannot apply Proposition~\ref{prop:core} directly because~$J$ needs not be $\eta$-alphabetic. To fix this, we define a new alphabet \frB and an auxiliary morphism $\eta_\frB:\frB^*\to N$ such that $J$ is $\eta_\frB$-alphabetic. Let $T$ be the set of all triples $(u_1,a,u_2)\in A^*\times A\times A^*$ such that $J \Jords \eta(u_1)$, $J \Jords \eta(u_2)$ and $\eta(u_1au_2) \in J$. We now let $\frB = \{\tau(u_1au_2) \in Q\mid (u_1,a,u_2) \in T\}$ and define $\tau_\frB: \frB^* \to Q$ by $\tau_\frB(b) = b \in\frB \subseteq Q$ for each letter $b \in \frB$. We also define $\eta_\frB = \npr{\cdot}\circ\tau_\frB: \frB^*\to N$. By definition, $\eta_\frB(b)\in J$ for all $b \in \frB$: the \Jrel-class $J \subseteq N$ is $\eta_\frB$-alphabetic.

We define a morphism $\gamma: \frB^* \to Q \times 2^Q$ (which is a monoid for the natural componentwise multiplication on $Q \times 2^Q$). For every $b \in \frB$, we let $\gamma(b) = (b,R) \in Q \times 2^Q$ where $R$ is the set of all $r \in Q$ such that $(b,r)$ is a good pair.  Since good pairs are preserved under multiplication, we have the following fact.

\begin{restatable}{fact}{bprops} \label{fct:bprops}
  Let $z \in \frB^+$, and $(q,R) = \gamma(z)$. Then, $q = \tau_\frB(z)$ and $(q,r)$ is a good pair for every $r \in R$.
\end{restatable}

We are ready to use Proposition~\ref{prop:core}. By definition, $\npr{q} = d$ for each $q \in \barS$ and $d \in J$. Hence, $\tau_\frB\inv(\barS) \subseteq \eta_\frB\inv(d) \subseteq \eta_\frB\inv(J)$ and since $J$ is $\eta_\frB$-alphabetic, Proposition~\ref{prop:core} yields a finite set $P$ of $\gamma$-guarded $\eta_\frB$-patterns such that $\pi(\widebar{p}) \in \tau_\frB\inv(\barS)$ for all $\widebar{p} \in P$ and $\{\uclos_{\eta_\frB} \widebar{p} \mid \widebar{p} \in P\}$ is a cover of $\tau_\frB\inv(\barS)$.

\smallskip
We now use $P$ to build the \pol{\Cs}-cover \Hb of $\tau\inv(\barS)$. First, we associate a language $U_z \in \pol{\Cs}$ to each word $z \in \frB^+$. This is where we use induction on the index of $J$ in Proposition~\ref{prop:genmain}. In the statement, we say that a cover $\Lb$ of some language $U$ is \emph{tight} to indicate that $L \subseteq U$ for all $L \in \Lb$.

\begin{restatable}{lemma}{imageu} \label{lem:imageu}
  Let $z \in \frB^+$ and $(q,R)= \gamma(z)$. Then, there exists $U_z \in\pol{\Cs}$ such that $\tau\inv(q) \subseteq U_z \subseteq \eta\inv(\npr{q}) \subseteq A^*$, $\tau(U_z) \subseteq R$ and there exists a tight safe \bpol{\Cs}-cover of\/ $U_z$.
\end{restatable}

\begin{proof}
  We first consider the special case where $z \in \frB$: we have $z = \tau(u_1au_2) \in Q$ where $(u_1,a,u_2)\in T$. For $i \in \{1,2\}$, we let $q_i = \tau(u_i)$. Lemma~\ref{lem:bupol} yields languages $H_i \in \pol{\Cs}$ such that  $\tau\inv(q_i) \subseteq H_i \subseteq \eta\inv(\npr{q_i})$ and $(q_i,r)$ is a good pair for all $r \in \tau(H_i)$. Let $J_i\subseteq N$ be the \Jrel class of \npr{q_i}. Since $\npr{q_i} = \eta(u_i)$, we have $J \Jords J_i$ by definition of $T$. Hence, induction on the index of $J$ in Proposition~\ref{prop:genmain} yields a safe \bpol{\Cs}-cover $\Kb'_i$ of $\eta\inv(J_i)$. Let $\Kb_i = \{K \cap H_i \mid K \in \Kb'_i\}$. Since $H_i \subseteq \eta\inv(\npr{q_i})$ and $\npr{q_i} \in J_i$, it follows that $\Kb_i$ is a tight safe \bpol{\Cs}-cover of $H_i$. Let $U_z = H_1aH_2 \in \pol{\Cs}$. Since $q = q_1\tau(a)q_2$, we get $\tau\inv(q) \subseteq U_z \subseteq \eta\inv(\npr{q})$. Also, since $\Kb_i$ is a \bpol{\Cs}-cover of $H_i \in \pol{\Cs}$ for $i \in \{1,2\}$, Lemma~\ref{lem:bpconcatm} yields a \bpol{\Cs}-cover $\Kb$ of $U_z$ such that each $K \in \Kb$ satisfies $K \subseteq K_1aK_2$ for some $K_1 \in \Kb_1$ and $K_2 \in \Kb_2$. Since $\Kb_1$ and $\Kb_2$ are tight and safe, one can verify that \Kb is itself a tight safe \bpol{\Cs}-cover of $U_z = H_1aH_2$.

  We turn to the general case. Let $z_1,\dots,z_n \in \frB$ such that $z = z_1 \cdots z_n$ and $(q_i,R_i) = \gamma(z_i)$ for all $i$. The previous case yields $U_i \in \pol{\Cs}$ such that $\tau\inv(q_i) \subseteq U_i \subseteq \eta\inv(\npr{q_i})$ and $\tau(U_i) \subseteq R_i$, and a tight safe \bpol{\Cs}-cover $\Kb_i$ of $U_i$ for all $i$. We define $U_z = U_1 \cdots U_n$. It belongs to \pol{\Cs} by closure under concatenation.  Clearly, $\tau\inv(q) \subseteq U_z \subseteq \eta\inv(\npr{q})$ and $\alpha(U_z) \subseteq R$ since $q = q_1 \cdots q_n$ and $R = R_1 \cdots R_n$. Since $U_1,\dots,U_n \in \pol{\Cs}$, Lemma~\ref{lem:bpconcat} yields a \bpol{\Cs}-cover \Kb of $U_1 \cdots U_n$ such that for all $K \in \Kb$, there is $K_i \in \Kb_i$ for each $i \leq n$ such that $K \subseteq K_1 \cdots K_n$. As the covers $\Kb_i$ are tight and safe, the cover \Kb is also tight and~safe.
\end{proof}

Now, let $\widebar{p} = (z_1,g_1,\dots,z_n,g_n,z_{n+1})$ be an $\eta_\frB$-pattern such that $\widebar{p} \neq (\veps)$. In this case $z_i \in \frB^+$ for all $i \leq n$ by definition and  Lemma~\ref{lem:imageu} yields a language $U_{z_i} \in \pol{\Cs}$. We let,
\[
  H_{\widebar{p}} = U_{z_1}\ \eta\inv(g_1)\ \cdots\ U_{z_n}\ \eta\inv(g_n)\ U_{z_{n+1}} \subseteq A^*.
\]
Note that $H_{\widebar{p}}$ is defined for all $\widebar{p} \in P$.  Indeed, as $\tau_\frB(\pi(\widebar{p})) \in \barS$, we have $\eta_\frB(\pi(\widebar{p})) = d$. Thus $\widebar{p} \neq (\veps)$ since $d \Jords 1_N$. Finally, we let $\Hb = \{H_{\widebar{p}}\mid\widebar{p} \in P\}$. This is our \pol{\Cs}-cover of $\tau\inv(\barS)$.

\begin{restatable}{lemma}{iscover} \label{lem:iscover}
  The set \Hb is a cover of $\tau\inv(\barS)$.
\end{restatable}

\begin{proof}
  Let $w \in \tau\inv(\barS)$. We exhibit $H \in \Hb$ containing $w$. By definition of $\barS$, we have $\eta(w) = d \in J$. Since $J \Jords 1_N$, there exists a decomposition $w=w_1a_1 \cdots w_na_n w_{n+1}$ with $n \geq 1$, such that $J \Jords \eta(w_i)$ for all $i\leq n+1$ and $J \Jrel \eta(w_ia_i)$ for all $i\leq n$. We let $b_i = (\eta(w_ia_i),\tau(w_ia_i)) \in \frB$ for all $i < n$.  Let now $b_{n} = (\eta(w_na_nw_{n+1}),\tau(w_na_nw_{n+1})) \in \frB$. We consider the word $z = b_1 \cdots b_n \in \frB^+$. By construction, $\tau_\frB(z) = \tau(w) \in \barS$. Since the set $\{\uclos_{\eta_\frB} \widebar{p} \mid \widebar{p} \in P\}$ is a cover of $\tau_\frB\inv(\barS)$, we get $\widebar{p} \in P$ such that $z \in \uclos_{\eta_\frB} \widebar{p}$. We consider $H_{\widebar{p}} \in \Hb$ and prove that $w \in H_{\widebar{p}}$. Let $\widebar{p} = (z_1,g_1,\dots,z_\ell,g_{\ell},z_{\ell+1}) \in (\frB^+ \times E(N)) \times \frB^+$. Since $z \in \uclos_{\eta_\frB} \widebar{p}$, we know that $z= z_1y_1 \cdots z_ny_\ell z_{\ell+1}$ with $y_i \in \eta_\frB\inv(g_i)$ for all $i$. By definition of $z$ from $w$, we have a decomposition $w = u_1v_1 \cdots u_\ell v_{\ell}u_{\ell+1}$ with $\tau(u_i) = \tau_\frB(z_i)$ and $\tau(v_i)= \tau_\frB(y_i)$ for all $i$. By Lemma~\ref{lem:imageu}, we obtain $\tau\inv(\tau_\frB(z_i)) \subseteq U_{z_i}$ for all $i$, which yields $u_i \in U_{z_i}$. From the equality $\tau(v_i) = \tau_\frB(y_i)$, we deduce that $\eta(v_i) = \eta_\frB(y_i) = g_i$ for every $i$. Therefore, $v_i \in \eta\inv(g_i)$. Altogether, we obtain finally that $w \in U_{z_1}\eta\inv(g_1) \cdots U_{z_\ell}\eta\inv(g_\ell)U_{z_{\ell+1}}$. This exactly says that $w \in H_{\widebar{p}}$, completing the proof.
\end{proof}

This concludes the first step in the outline of the argument sketched at the beginning of Section~\ref{sec:first-step}. Let us proceed to the second: for each $H\in\Hb$, we build a $(q_1,q_2)$-safe \bpol{\Cs}-cover $\Kb_H$ of $H$. Lemma~\ref{lem:iscover} will then imply that $\Kb_{\barS} = \bigcup_{H \in \Hb} \Kb_H$ is the desired $(q_1,q_2)$-safe \bpol{\Cs}-cover of $\tau\inv(\barS)$. We fix $H \in \Hb$ for the proof. This yields $\widebar{p} = (z_1,g_1,\dots,z_n,g_n,z_{n+1}) \in P$ such that $\pi(\widebar{p})\in\tau_\frB\inv(\barS)$ and $H = U_{z_1} \eta\inv(g_1) \cdots U_{z_n} \eta\inv(g_n) U_{z_{n+1}}$. Therefore, $H$ is a concatenation of languages in \pol{\Cs} (recall that $\eta$ is a \Cs-morphism). We build a \bpol{\Cs}-cover of each language in the concatenation and then use Lemma~\ref{lem:bpconcat} to build a cover of the whole language $H$. First, Lemma~\ref{lem:imageu} provides a tight safe \bpol{\Cs}-cover $\Lb_i$ of $U_{z_i}$ for all $i$.

We cover the languages $\eta\inv(g_i)$ by induction on the rank of $(q_1,d,q_2)$. As $\pi(\widebar{p})\in\tau_\frB\inv(\barS)$, we have $\eta_\frB(\pi(\widebar{p})) = d \in J$. By definition of $\eta_\frB$-patterns, it follows that $g_1,\dots,g_n \in J$. Thus, we~may consider the triples $(r_1,g_i,r_2) \in Q \times J \times Q$.

\begin{restatable}{lemma}{indrank} \label{lem:indrank}
  For $1 \!\leq\! i \!\leq\! n$, there is a tight\/ \bpol{\Cs}-cover~$\Vb_i$ of~$\eta\inv(g_i)$ that is $(r_1,r_2)$-safe for all $r_1,r_2 \in Q$ satisfying $(q_1,d,q_2) < (r_1,g_i,r_2)$.
\end{restatable}

\begin{proof}
  Let $X \subseteq Q^2$ be the set of all pairs $(r_1,r_2) \in M^2$ such $(q_1,d,q_2) < (r_1,g_i,r_2)$. For each $(r_1,r_2)\in X$, the rank of $(r_1,g_i,r_2)$ is strictly smaller than the one of $(q_1,d,q_2)$. Hence, we may use  induction on the rank of $(q_1,d,q_2)$ to get a $(r_1,r_2)$-safe \bpol{\Cs}-cover $\Vb_{r_1,r_2}$ of $\eta\inv(g_i)$.  Let $\ell = |X|$ and $\{(r_{1,1},r_{2,1}),\dots,(r_{1,\ell},r_{2,\ell})\} = X$. We define,
  \[
    \Vb_i\! =\! \{V_1 \cap \cdots \cap V_\ell \cap \eta\inv(g_i) \mid \text{$V_j\! \in\! \Vb_{r_{1,j},r_{2,j}}$ for all $j\! \leq\! \ell$}\}.
  \]
  Since $\eta$ is a \Cs-morphism and \bpol{\Cs} is closed under intersection, it is immediate that $\Vb_i$ is tight \bpol{\Cs}-cover of $\eta\inv(g_i)$. One can verify that it satisfies the desired property.
\end{proof}

Lemma~\ref{lem:bpconcat} yields a \bpol{\Cs}-cover $\Kb_H$ of $H$ such that for every $K \in \Kb_H$, there exist languages $L_i \in \Lb_i$ and $V_i \in \Vb_{i}$ for all $i$ such that $K \subseteq L_1V_1 \cdots L_nV_nL_{n+1}$. We prove that $\Kb_H$ is $(q_1,q_2)$-safe.  For $K \in \Kb_H$ and $u,u' \in K$, we show that $q_1\tau(u)q_2 = q_1\tau(u')q_2$.

By definition, we have $K \subseteq L_1V_1 \cdots L_nV_nL_{n+1}$ where $L_i\! \in\! \Lb_i$ and $V_i \in \Vb_{i}$ for all $i$. Since $u,u' \in K$, we get $u^{}_i,u'_i\in L_i$ and $v^{}_i,v'_i \in V_i$ for all $i$ such that $u = u_1v_1 \cdots u_nv_nu_{n+1}$ and $u' = u'_1v'_1 \cdots u'_n v'_n u'_{n+1}$.  Let $u'' = u_1v'_1\cdots u_nv'_nu_{n+1}$. Since the covers $\Lb_i$ are safe, we have $\tau(u_i) = \tau(u'_i)$ for all $i$. Thus, $\tau(u') = \tau(u'')$. It remains to prove that $q_1\tau(u)q_2 = q_1\tau(u'')q_2$. For each $i$ such that $1 \leq i \leq n$, we define,
\[
  x_i = u_1v_1u_2 \cdots v_{i-1}u_i \text{ and } y_i = u^{}_{i+1}v'_{i+1} \cdots u_nv'_nu^{}_{n+1}.
\]
In particular, we let $x_1 = u_1$ and $y_n = u_{n+1}$. Observe that $u = x_nv_ny_n$, $u'' = x^{}_1v'_1y^{}_1$ and $x^{}_iv'_iy^{}_i = x^{}_{i-1}v^{}_{i-1}y^{}_{i-1}$ for all~$i$. We prove that $q^{}_1\tau(x^{}_iv^{}_iy^{}_i)q^{}_2 = q^{}_1\tau(x^{}_iv'_iy^{}_i)q^{}_2$ for all~$i$. By transitivity, this will yield $q^{}_1\tau(x^{}_nv^{}_ny^{}_n)q^{}_2 = q^{}_1\tau(x^{}_1v'_1y^{}_1)q^{}_2$, \emph{i.e.} $q_1\tau(u)q_2 = q_1\tau(u'')q_2$, as desired. We now fix $i$ and prove that $q_1\tau(x_iv_iy_i)q_2 = q_1\tau(x^{}_iv'_iy^{}_i)q_2$. We write $s_1 = \tau(x_i)$ and $s_2 = \tau(v_i)$. The proof is based on the following lemma, which is where we need Equation~\eqref{eq:capolcp} (as we apply Lemma~\ref{lem:rcapolcp}, whose proof relies on this equation).

\begin{restatable}{lemma}{getstrict} \label{lem:getstrict}
  We have $(q_1,d,q_2) < (q_1s_1,g_i,s_2q_2)$.
\end{restatable}

By definition of $\Vb_i$ in Lemma~\ref{lem:indrank}, it follows from Lemma~\ref{lem:getstrict} that $\Vb_i$ is a $(q_1s_1,s_2q_2)$-safe \bpol{\Cs}-cover of $\eta\inv(g_i)$. Thus, as $V_i \in \Vb_i$ and $v^{}_i,v'_i \in V_i$, we get $q^{}_1s^{}_1\tau(v^{}_i)s^{}_2q^{}_2\!=\! q^{}_1s^{}_1\tau(v'_i)s^{}_2q^{}_2$, \emph{i.e.}, $q^{}_1\tau(x^{}_iv^{}_iy^{}_i)q^{}_2 = q^{}_1\tau(x^{}_iv'_iy^{}_i)q^{}_2$ as desired.

\begin{proof}[Proof of Lemma~\ref{lem:getstrict}]
  Recall that $\widebar{p} = (z_1,g_1,\dots,z_n,g_n,z_{n+1})$. Let $(r_j,R_j)=\gamma(z_j)$ for all~$j$. By Fact~\ref{fct:bprops}, $(r_j,r)$ is a good pair for all $r \in R_j$. Since $\widebar{p}$ is $\gamma$-guarded, we get $(e_j,F_j) \in E(Q \times 2^Q)$ for every $1\leq j \leq n$ that  satisfy $\eta_\frB\inv(g_j) \cap \gamma\inv((e_j,F_j)) \neq \emptyset$, $(r_je_j,R_jF_j) = (r_j,R_j)$ and $(e_jr_{j+1},F_{j}R_{j+1}) = (r_{j+1},R_{j+1})$.  In particular, $g_j = \npr{e_j}$. By Fact~\ref{fct:bprops}, $(e_j,r)$ is a good pair for all $r \in F_j$. Let  $\wh{r_1} = r_1 \cdots r_i$ and $\wh{r_2} = r_{i+1} \cdots r_n$.

  First, we show that $F_i \subseteq F_ie_iF_i$. Let $m \in F_i$. Since $F_iF_i = F_i$, we obtain $m_1,\dots,m_{|Q|} \in F_i$ such that $m = m_1 \cdots m_{|Q|}$. The pigeonhole principle yields $0 \leq k < \ell \leq |Q|$ such that $m_1 \cdots m_k = m_1 \cdots m_\ell$. Define $f = (m_{k+1} \cdots m_\ell)^\omega$, which is an idempostent of $Q$. We can now write $m = m_1 \cdots m_kf m_{\ell+1} \cdots m_{|M|}$.  Since  $F_i \in E(2^Q)$, we get $f \in F_i$ and $m \in F_ifF_i$. Thus, $(e_i,f)$ is a good pair and Lemma~\ref{lem:rcapolcp} yields $f = fe_if$. Hence,  $m \in F_ife_ifF_i \subseteq F_ie_iF_i$.

  We now prove that $(q_1,d,q_2) \leq (q_1s_1,g_i,s_2q_2)$. We show that $s_1 = s'_1e_is''_1$ and $s_2 = s''_2e_is'_2$ for two elements $s'_1,s'_2 \in Q$ such that $\npr{s'_1} = \npr{\wh{r_1}}$ and $\npr{s'_2} = \npr{\wh{r_2}}$, and $s''_1,s''_2 \in F_i$. This will give $\npr{s'_1e_is'_2} = \npr{\wh{r_1}e_i\wh{r_2}} = \npr{\wh{r_1}\wh{r_2}}$. Since $\wh{r_1}\wh{r_2} = \tau_\frB(\pi(\widebar{p})) \in \barS$, we obtain $\npr{s'_1e_is'_2} = d$. Also, $\npr{s''_1} = \npr{s''_2} = \npr{e_i} = g_i$ since $s''_1,s''_2 \in F_i$. Thus, $\npr{e_i} = \npr{s''_1}g_i\npr{s''_2}$. Altogether, we get $(q_1,d,q_2) \leq (q_1s_1,g_i,s_2q_2)$. By symmetry, we only decompose~$s_1$. By definition, $s_1 = \tau(x_i)$ where $x_i = u_1v_1u_2 \cdots v_{i-1}u_i$. We have $u_j \in U_{z_j}$ for all $j$ since $u_j \in L_j$ and $\Lb_j$ is a tight cover of $U_{z_j}$. Thus, $\tau(u_j) \in R_j$ for all $j$ by Lemma~\ref{lem:imageu}. Moreover, $R_i = R_iF_i$ and since $F_i \subseteq F_ie_iF_i$, this yields $\tau(u_i) = p_ie_is''_1$ with $p_i \in R_i$ and $s''_1 \in F_i$. Therefore, we have $s_1 = s'_1e^{}_is''_1$ with $s'_1 = \tau(u_1v_1u_2 \cdots v_{i-1})p_i$. We also have $\eta(v_j) = g_j = \npr{e_j}$ for all $j$ since $v_j \in V_j$ and $\Vb_j$ is a tight cover of $\eta\inv(g_j)$. Finally, Fact~\ref{fct:bprops} yields that for all $j$, if $r \in R_j$, then $(r_j,r)$ is a good pair. This implies that $\npr{r} = \npr{r_j}$. Altogether, we get $\npr{s'_1} = \npr{r_1e_1 \cdots r_{i-1}e_{i-1} r_i} = \npr{r_1 \cdots r_i}= \npr{\wh{r_1}}$.

  It remains to prove that the inequality is strict. By contradiction, assume that we have $(q_1s_1,g_i,s_2q_2) \leq (q_1,d,q_2)$. Let $q= \tau_\frB(\pi(\widebar{p}))$. We show that $q$ stabilizes $(q_1,d,q_2)$. Since  $q \in \barS$ (as $\widebar{p} \in P$), this is a contradiction. By hypothesis, we get $t_1,t_2,t'_1,t'_2 \in Q$ and $f \in E(Q)$ such that $\npr{t_1ft_2} = g_i$, $\npr{f} = \npr{t'_1}d\npr{t'_2}$, $q_1 = q_1s_1t_1ft'_1$ and $q_2 = t'_2ft_2s_2q_2$. We will now construct $r'_1,r'_2 \in Q$ such that $q_1\wh{r_1}e_ir'_1 = q_1$, $r'_2e_i\wh{r_2}q_2 = q_2$, $\npr{r'_1} = \npr{t_1ft'_1(\wh{r_1}t_1ft'_1)^\omega}$ and $\npr{r'_2} = \npr{(t'_2ft_2\wh{r_2})^{\omega}t'_2ft_2}$. Let us first explain why this implies that $q$ stabilizes $(q_1,d,q_2)$. As seen above, $\npr{\wh{r_1}e_i\wh{r_2}} = d$. Moreover, as $\npr{t_1ft_2} = g_i$, $\npr{f} = \npr{t'_1}d\npr{t'_2}$ and $g_i = \npr{e_i}$, we~have,
  \[
    \begin{array}{lll}
      g_i & = & \npr{(t_1ft'_1\wh{r_1})^{\omega+1}}g_i\npr{(\wh{r_2}t'_2ft_2)^{\omega+1}}, \\
          & = & \npr{t_1ft'_1(\wh{r_1}t_1ft'_1)^\omega } \npr{\wh{r_1}e_i\wh{r_2}} \npr{(t'_2ft_2\wh{r_2})^{\omega}t'_2ft_2}, \\
          & = &  \npr{r'_1}d\npr{r'_2}.
    \end{array}
  \]
  Moreover, $q =  \tau_\frB(\pi(\widebar{p})) = \tau(\wh{r_1}e_i\wh{r_2})$. Thus,  if $q_1\wh{r_1}e_ir'_1\! =\! q_1$ and $r'_2e_i\wh{r_2}q_2 = q_2$, we get $(q_1,d,q_2) \smarrow{q} (q_1,d,q_2)$.

  It remains to build $r'_1$ and $r'_2$. By symmetry, we focus on $r'_1$. As seen above, we have $s_1 = \tau(u_1v_1u_2 \cdots v_{i-1}u_i)$ with $(r_j,\tau(u_j))$ a good pair  and $\eta(v_j) = g_j$ for all $j$. We define $m_1 = r_1\tau(v_1) \cdots r_{i-1}\tau(v_{i-1})r_i$.  By closure under multiplication, we get that $(m_1t_1ft'_1,s_1t_1ft'_1)$ is a good pair. We get $(s_1t_1ft'_1)^{\omega+1} = (s_1t_1ft'_1)^{\omega}m_1t_1ft'_1 (s_1t_1ft'_1)^{\omega}$ from Lemma~\ref{lem:rcapolcp}. Since $q_1 = q_1s_1t_1ft'_1$, multiplying by $q_1$ on the left yields the equality $q_1 = q_1m_1t_1ft'_1 (s_1t_1ft'_1)^{\omega}$. We let $m'_1 = t_1ft'_1 (s_1t_1ft'_1)^{\omega}$. Thus, $q_1 = q_1m_1m'_1$. For all $j$, we have $\npr{\tau(v_j)} = g_j = \npr{e_j}$ and $(e_j,e_j\tau(v_j)e_j)$ is a good pair by Lemma~\ref{lem:rpairs}. Since $r_je_j = r_j$ and $e_jr_{j+1}=r_{j+1}$ for all $j$, closure under multiplication yields that $(\wh{r_1}m'_1,m_1m'_1)$ is a good pair as well. We use Lemma~\ref{lem:rcapolcp} to get $(m_1m'_1)^{\omega+1} = (m_1m'_1)^{\omega} \wh{r_1}m'_1 (m_1m'_1)^{\omega}$. Since $q_1 = q_1m_1m'_1$, multiplying by $q_1$ on the left yields $q_1 = q_1\wh{r_1}m'_1 (m_1p'_1)^{\omega}$. If we define $r'_1 = m'_1 (m_1m'_1)^{\omega}$, we get indeed $\npr{r'_1} = \npr{t_1ft'_1(\wh{r_1}t_1ft'_1)^\omega}$.
\end{proof}

\subsection{\texorpdfstring{Second step: full cover of $\eta\inv(d)$}{Second step: full cover of eta-1(d)}}

This step is the same whether $J$ is maximal or not. Indeed, in both cases, we have a $(q_1,q_2)$-safe \bpol{\Cs}-cover $\Kb_{\barS}$ of $\tau\inv(\barS)$. We let,
\[
  K' = \bigcup_{K \in \Kb_{\barS}} K,\quad K_{\bot} = \eta\inv(d) \setminus K' \quad \text{and} \quad \Kb = \{K_{\bot}\} \cup \Kb_{\barS}.
\]
Since $\eta$ is a \Cs-morphism and $K \in \bpol{\Cs}$ for all $K \in \Kb_{\barS}$, we have $K_{\bot} \in \bpol{\Cs}$. Therefore, \Kb is a \bpol{\Cs}-cover of $\eta\inv(d)$ by definition. We prove that it is $(q_1,q_2)$-safe.  Let $K \in \Kb$, $w,w' \in K$, $q = \tau(w)$ and $r = \tau(w')$. We show that $q_1 q q_2= q_1 r q_2$. If $K \in \Kb_{\barS}$, this is because $\Kb_{\barS}$ is $(q_1,q_2)$-safe.

Assume now that $K = K_{\bot}$. By definition of $K_{\bot}$, we know that $w,w' \in \eta\inv(d)$ and $w,w' \not\in \bigcup_{K \in \Kb_{\barS}} K$. Thus, $q,r \not\in \barS$ since $\Kb_{\barS}$ is a cover of $\tau\inv(\barS)$. In other words, we have $(q_1,d,q_2) \smarrow{q} (q_1,d,q_2)$ and $(q_1,d,q_2) \smarrow{r} (q_1,d,q_2)$. We get $s^{}_1,s^{}_2,s'_1,s'_2,t^{}_1,t^{}_2,t'_1,t'_1 \in Q$ and $e,f \in E(Q)$ such that,
\begin{enumerate}
  \item $\npr{s_1es_2} = d$,\quad $\npr{e} = \npr{s'_1}d\npr{s'_2}$,\\ $\npr{t_1ft_2} = d$,\quad $\npr{f} = \npr{t'_1}d\npr{t'_2}$.
  \item  $q = s_1es_2$,\quad $q^{}_1 = q^{}_1s^{}_1es'_1$,\quad $q^{}_2 = s'_2es^{}_2q^{}_2$,\\ $r = t_1ft_2$,\quad $q^{}_1 = q^{}_1t^{}_1ft'_1$,\quad $q^{}_2 = t'_2ft^{}_2q^{}_2$.
\end{enumerate}
This yields $\npr{e} = \npr{s'_1t^{}_1ft^{}_2s'_2}$ and  $\npr{f} = \npr{t'_1s^{}_1es^{}_2t'_2}$. Hence, we apply Lemma~\ref{lem:rswap} (which follows from~\eqref{eq:swap}) to get,
\[
  \begin{array}{c}
    (es'_1t^{}_1ft'_1s^{}_1e)^\omega (es^{}_2t'_2ft^{}_2s'_2e)^\omega = \\
    (es'_1t^{}_1ft'_1s^{}_1e)^\omega s'_1t^{}_1ft^{}_2s'_2(es^{}_2t'_2ft^{}_2s'_2e)^\omega
  \end{array}
\]
We multiply by $q_1s_1$ on the left and $s_2q_2$ on the right. Recall that $q^{}_1 = q^{}_1s^{}_1es'_1$, $q^{}_2 = s'_2es^{}_2q^{}_2$, $q^{}_1 = q^{}_1t^{}_1ft'_1$ and $q^{}_2 = t'_2ft^{}_2q^{}_2$. We deduce that $q_1s_1es_2q_2 = q_1t_1ft_2q_2$. Since we have $q = s_1es_2$ and $r = t_1ft_2$, we finally obtain $q_1qq_2 = q_1rq_2$, as desired.

\section{Handling full levels above two}
\label{sec:uptwo}
We specialize Theorem~\ref{thm:cargen} to characterize the classes \bpol{\Cs} when \Cs is itself of the form $\Cs = \bpol{\Ds}$. More precisely, we prove that in this case, Equation~\eqref{eq:swap} (which is based on \Cs-swaps) is equivalent to a property based on \Cs-\emph{sets}.

\begin{restatable}{proposition}{simpeq} \label{prop:simpeq}
  Let \Ds be a \vari and $\Cs = \bpol{\Ds}$. Let $\alpha: A^* \to M$ be a surjective morphism. Then, $\alpha$ satisfies~\eqref{eq:swap} if and only if it satisfies the following condition:
  \vspace{-0.05cm}
  \begin{equation}
    \begin{array}{c}
      (eqfre)^\omega (esfte)^\omega = (eqfre)^\omega qft (esfte)^\omega \text{ for every}\\
      \text{\Cs-set $\{q,r,s,t,e,f\} \subseteq M$ such that $e,f \in  E(M)$}.
    \end{array} \label{eq:csides}
  \end{equation}
\end{restatable}

Before we prove Proposition~\ref{prop:simpeq}, let us combine it with Theorem~\ref{thm:cargen}. This yields the following theorem.

\begin{restatable}{theorem}{carbbpol} \label{thm:carbbpol}
  Let \Ds be a \vari, let $\Cs = \bpol{\Ds}$ and let $\alpha: A^* \to M$ be a surjective morphism. Then, $\alpha$ is a $\bpol{\Cs}$-morphism if and only if it satisfies the two following equations:
  \begin{alignat*}{1}
    \begin{array}{c}
      (eset)^{\omega+1} = (eset)^{\omega}et(eset)^{\omega} \text{ for every $t \in M$}\\
      \text{and every \Cs-pair $(e,s) \in E(M) \times M$}.
    \end{array} \tag{\ref{eq:cacp}} \\
    \begin{array}{c}
      (eqfre)^\omega (esfte)^\omega = (eqfre)^\omega qft (esfte)^\omega \text{ for every}\\
      \text{\Cs-set $\{q,r,s,t,e,f\} \subseteq M$ such that $e,f \in  E(M)$}.
    \end{array}\tag{\ref{eq:csides}}
  \end{alignat*}
\end{restatable}

Given a \vari \Cs with decidable \emph{covering}, Lemmas~\ref{lem:septopairs} and~\ref{lem:covtosets} imply that the \Cs-pairs and the \Cs-sets associated to an input morphism can be computed. Hence, one may decide~\eqref{eq:cacp} and~\eqref{eq:csides} in this case.  We get the following corollary.

\begin{corollary} \label{cor:carbbpol}
  Let \Ds be a \vari and let $\Cs = \bpol{\Ds}$. If \Cs-covering is decidable, then so is \bpol{\Cs}-membership.
\end{corollary}

Corollary~\ref{cor:carbbpol} implies  that given a concatenation hierarchy, if a \emph{non-zero full} level has decidable \emph{covering}, then the next full level has decidable \emph{membership}. Together with Theorem~\ref{thm:grpgen} (which is proved in~\cite{pzconcagroup,PlaceZ22}), this yields the following.

\begin{corollary} \label{cor:cargrp}
  If\/ \Gs is a group \vari with decidable separation, the classes \bpolp{2}{\Gs} and \bpolp{2}{\Gs^+} have decidable membership.
\end{corollary}

Also, applying Theorem~\ref{thm:stdot} (which is proved in~\cite{pzbpolcj}) yields the decidability of membership for the levels \emph{three} in both the Straubing-Thérien hierarchy and the dot-depth.

\begin{corollary} \label{cor:cardd3}
  The classes \bpolp{3}{\stzer} and \bpolp{3}{\dotzer} have decidable membership.
\end{corollary}

\begin{restatable}{remark}{thmfail} \label{rem:thmfail}
  For an \emph{arbitrary} \vari \Cs, it remains true that~\eqref{eq:cacp} and~\eqref{eq:csides} form a necessary condition for being a \bpol{\Cs}-morphism. On the other hand, they need not be sufficient. For instance, consider the class \md of modulo languages. One can reformulate the equations with a notion tailored to \md: \emph{stable monoids}~\cite{bookstraub}. Specifically, one can prove that a morphism satisfies~\eqref{eq:cacp} and~\eqref{eq:csides} if and only if its stable monoid is \Jrel-trivial. It is well-known that this property is not sufficient for~membership in \bpol{\md} (see \emph{e.g.} \cite{ChaubardPS06}).

  On the other hand, Theorem~\ref{thm:carbbpol} does hold for some input classes \Cs which are \emph{not} of the form \bpol{\Ds}. For instance, this is the case for $\Cs = \at$ as seen in Corollary~\ref{cor:bpolat}. This is also the case when $\Cs = \stzer$ (for $\stzer=\{\emptyset,A^*\}$) or $\Cs = \grp$ (for the class \grp of all group languages).
\end{restatable}

\subsection{Proposition~\ref{prop:simpeq}: ``only if'' direction}
Let $\Cs= \bpol{\Ds}$ for a \vari~\Ds and let $\alpha: A^* \to M$ be a morphism. We start with the simpler ``only if'' direction, which does not require the hypothesis $\Cs=\bpol{\Ds}$. Assume that $\alpha$ satisfies~\eqref{eq:swap}. We show that~\eqref{eq:csides} holds as well. Let  $\{q,r,s,t,e,f\}\subseteq M$ be a \Cs-set such that $e,f  \in E(M)$. We show that,
\begin{equation} \label{eq:simpeqeasy}
  (eqfre)^\omega(esfte)^\omega = (eqfre)^\omega qft(esfte)^\omega.
\end{equation}
Lemma~\ref{lem:witness} yields a \Cs-morphism $\eta: A^* \to N$ such that the \Cs-swaps are exactly the $\eta$-swaps. Since $\{q,r,s,t,e,f\}\subseteq M$ is a \Cs-set, Lemma~\ref{lem:setmor} implies that it is an $\eta$-set. We obtain a set $F \subseteq A^*$ such that $\alpha(F)=\{q,r,s,t,e,f\}$ and $\eta(u)=\eta(u')$ for every $u,u' \in F$. Therefore, there exist $x \in N$ and, for each $p \in \{q,r,s,t,e,f\}$, a word $u_p \in F$ such that $\alpha(u_p) = p$ and $\eta(u_p) = x$. Let $k = \omega(N)$. We~define the following words:
\[
  \begin{array}{c}
  	u'_q = u_qu_f^{k-1},\ u'_r = u_ru_e^{k-1},\ u'_s = u_e^{k-1}u_s,\ u'_t = u_f^{k-1}u_t,\\ u'_e=u_e^k \text{ and } u'_f = u_f^k.
  \end{array}
\]
Since $e,f \in E(M)$, we have $\alpha(u'_q) = qf$, $\alpha(u'_r) = re$, $\alpha(u'_s) = es$, $\alpha(u'_t) = ft$, $\alpha(u'_e) = e$ and $\alpha(u'_f) = f$. Also, $\eta(u'_p) = x^{\smash{k}} \in E(N)$ for all $p \in \{q,r,s,t,e,f\}$ by definition of $k$. Thus, $\eta(u'_qu'_fu'_t) = \eta(u'_e)$ and $\eta(u'_ru'_eu'_s) = \eta(u'_f)$. Hence, $(qf,re,es,ft,e,f)$ is an $\eta$-swap and a \Cs-swap by definition of $\eta$.  We now apply~\eqref{eq:swap} to get~\eqref{eq:simpeqeasy} as desired.

\subsection{Proposition~\ref{prop:simpeq}: ``if'' direction, preliminaries}

We now turn to the more difficult ``if'' direction, where we need the hypothesis that \Cs is of the form \bpol{\Ds}. Assume now that $\alpha$ satisfies~\eqref{eq:csides}. We prove that~\eqref{eq:swap} holds as well. Therefore, we fix a \Cs-swap $(q,r,s,t,e,f) \in M^6$ such that $e,f \in E(M)$ and show that,
\begin{equation} \label{eq:simpeqhard}
	(eqfre)^\omega(esfte)^\omega = (eqfre)^\omega qft(esfte)^\omega.
\end{equation}
We adopt the same strategy as for~the ``only if'' direction: we build a \Cs-set from $(q,r,s,t,e,f)$ and then apply~\eqref{eq:csides} to it to get~\eqref{eq:simpeqhard}. First, we build a \Cs-morphism $\eta$ such that the $\eta$-sets are exactly the \Cs-sets and satisfying properties tied to the hypothesis $\Cs = \bpol{\Ds}$.

\begin{restatable}{lemma}{morbp} \label{lem:morbp}
  There exists a \pol{\Ds}-morphism $\eta: A^*\to (N,\leq)$ satisfying the two following properties:
  \begin{itemize}
    \item For all $p \in N$, the set $\alpha(\eta\inv(p)) \subseteq M$ is a \Cs-set for $\alpha$.
    \item For all $g,g' \in E(N)$, if $g \leq g'$, then $g' = g'gg'$.
  \end{itemize}
\end{restatable}

\begin{proof}
  By Corollary~\ref{cor:bpolc}, the class $\Cs=\bpol{\Ds}$ is a \vari. Hence, it follows from Lemma~\ref{lem:setmor} that there exists a \Cs-morphism $\beta_1: A^*\to T_1$ such that for every $S \subseteq M$, $S$ is a \Cs-set for $\alpha$ if and only if $S$ is a $\beta_1$-set. Additionally, since \pol\Ds is a \pvari by Theorem~\ref{thm:polc} and $\Cs = \bpol{\Ds}$, Lemma~\ref{lem:bpolm} yields a \pol{\Ds}-morphism $\beta_2: A^* \to (T_2,\leq)$ and a morphism $\delta: T_2 \to T_1$ such that $\beta_1 = \delta \circ \beta_2$. Finally, it follows from Lemma~\ref{lem:pairmor} that there exists a \Ds-morphism $\gamma: A^* \to Q$ such that for every pair $(x,y) \in N^2$, $(x,y)$ is a \Ds-pair for $\beta_2$ if and only if $(x,y)$ is a $\gamma$-pair. Since $\Ds\subseteq\pol{\Ds}$, $\gamma$ is also a \pol{\Ds}-morphism (recall that we view $Q$ as ordered monoid $(Q,=)$). We now consider the ordered monoid $(T_2 \times Q,\leq)$ equipped with componentwise multiplication and whose ordering is defined by $(t,q) \leq (t',q')$ if and only if $t \leq t'$ and $q = q'$ for all $(t,q),(t',q') \in T_2 \times Q$. Moreover, we let $\beta_3: A^* \to (T_2 \times Q,\leq)$ as the morphism which is defined by $\beta_3(w) = (\beta_2(w),\gamma(w))$ for every $w \in A^*$. Finally, we let $\eta: A^* \to (N,\leq)$ as the surjective restriction of $\beta_3$. Since \pol{\Ds} is a \pvari by Theorem~\ref{thm:polc}, it is clear that $\eta$ is a \pol{\Ds}-morphism. Indeed, the definition implies that every language recognized by $\eta$ is built by making unions and intersections of languages recognized by the \pol{\Ds}-morphisms $\beta_2$ and $\gamma$. It remains to verify that $\eta$ satisfies the two assertions.

  For the first assertion, let $p \in N$. By definition of $\eta$, all words in $\eta\inv(p)$ have the same image under $\beta_2$ and therefore under $\beta_1$ as well since $\beta_1 = \delta \circ \beta_2$. Hence, $\alpha(\eta\inv(p))$ is a $\beta_1$-set for $\alpha$. It follows from the definition of $\beta_1$ that $\alpha(\eta\inv(p)) \subseteq M$ is a \Cs-set for $\alpha$.

  We turn to the second assertion. Let $g,g' \in E(N)$ such that $g \leq g'$. We prove that $g'=g'gg'$. Since $g' \in E(N)$ and $g \leq g'$, it is immediate that $g'gg'\leq g'$. Hence, it suffices to show that $g'\leq g'gg'$. Since $\eta$ is surjective by construction, there exist $u,u'\in A^*$ such that $\eta(u) = g$ and $\eta(u') = g'$. It remains to prove that $\eta(u')\leq \eta(u'uu')$. Since for all $w\in A^*$, we have $\eta(w)=(\beta_2(w),\gamma(w))$, this amounts to proving that $\beta_2(u')\leq \beta_2(u'uu')$ and $\gamma(u')= \gamma(u'uu')$. Since $\eta(u')\in E(N)$, we have $\gamma(u') \in E(Q)$ and $\beta_2(u')\in E(T_2)$. From $\eta(u) \leq \eta(u')$, we deduce $\gamma(u) = \gamma(u')$, yielding $\gamma(u') = \gamma(u'uu')$. Moreover, $\gamma(u) = \gamma(u')$ also implies that $(\beta_2(u'),\beta_2(u))$ is a $\gamma$-pair for $\beta_2$. Therefore, it is a \Ds-pair for $\beta_2$ by definition of $\gamma$. Since $\beta_2(u')\in E(T_2)$ and $\beta_2$ is a \pol{\Ds}-morphism, we deduce that $\beta_2(u') \leq \beta_2(u'uu')$ from Theorem~\ref{thm:polcar}, completing the proof.
\end{proof}

We use the \pol{\Ds}-morphism $\eta: A^* \to (N,\leq)$ to construct a \Cs-morphism $\beta: A^* \to R$. Since $(q,r,s,t,e,f)$ is a \Cs-swap, this will imply that it is also a $\beta$-swap. We shall then exploit the definition of $\beta$ to build our \Cs-set from $\eta$ using the first assertion in Lemma~\ref{lem:morbp}.

For every $v \in A^*$, we let $\eupp{v} = \{v' \in A^* \mid \eta(v) \leq \eta(v')\}$ which a language recognized by $\eta$. Hence, $\eupp{v} \in \pol{\Ds}$ since $\eta$ is a \pol{\Ds}-morphism. Let $k = |M|^2 \times |N|^2 + 2$. We define $\Hb = \{\eupp{v_1} \cdots \eupp{v_\ell}\mid \ell\leq k \text{ and } v_1,\dots,v_\ell\in A^*\}$. Since \pol{\Ds} is closed under concatenation by Theorem~\ref{thm:polc}, it follows that \Hb is a \emph{finite} set of languages in $\pol{\Ds} \subseteq \Cs$. Thus, Proposition~\ref{prop:genocm} yields a \Cs-morphism $\beta: A^* \to R$ recognizing every $H \in \Hb$. Since $v\in\eupp{v}$ for all $v\in A^*$, one may verify that for all $w,w' \in A^*$, we have $\beta(w)=\beta(w') \Rightarrow \eta(w)=\eta(w')$. We shall use this fact repeatedly.

\begin{remark}
	The construction of $\eta$ and $\beta$ is where we need the hypothesis that $\Cs = \bpol{\Ds}$.
\end{remark}

\subsection{Proposition~\ref{prop:simpeq}: ``if'' direction, main proof}

We are ready to start the construction of a \Cs-set from our \Cs-swap $(q,r,s,t,e,f)$. By definition of the \Cs-swaps, the hypothesis that $\beta: A^* \to R$ is a \Cs-morphism yields $u_q,u_r,u_s,u_t,u_e,u_f \in A^*$ such that $\beta(u_e),\beta(u_f) \in E(R)$, $\beta(u_e) = \beta(u_qu_fu_t)$, $\beta(u_f)=\beta(u_ru_eu_s)$ and $\alpha(u_x) = x$ for every $x \in \{q,r,s,t,e,f\}$. We define $p_q = \eta(u_q)$, $p_r= \eta(u_r)$, $p_s=\eta(u_s)$, $p_t= \eta(u_t)$, $g_e= \eta(u_e)$ and $g_f = \eta(u_f)$.

Recall that $\beta(w)=\beta(w') \Rightarrow \eta(w)=\eta(w')$ for all $w,w' \in A^*$. Hence, by hypothesis on $u_q,u_r,u_s,u_t,u_e,u_f \in A^*$, the definitions imply that $g_e,g_f \in E(N)$, $g_e=p_qg_fp_t$ and  $g_f=p_rg_ep_s$.

\begin{restatable}{fact}{etaprop} \label{fct:etaprop}
  We have the following equalities:
  \begin{itemize}
    \item $g_e  = g_ep_sg_fp_rg_e = (g_ep_sg_fp_tg_e)^\omega$ and
    \item  $g_f =  g_fp_tg_ep_qg_f = (g_fp_tg_ep_sg_f)^\omega$.
  \end{itemize}
\end{restatable}

\begin{proof}
	First, observe that since $g_e,g_f \in E(N)$, $g_e=p_qg_fp_t$ and  $g_f=p_rg_ep_s$, we have the equalities $g_e=g_ep_qg_fp_tg_e$ and $g_f=g_fp_rg_ep_sg_f$.

  We now prove that $g_e=(g_ep_sg_fp_tg_e)^\omega$. By Lemma~\ref{lem:htogroup}, it suffices to prove that $g_e \Hrel g_ep_sg_fp_tg_e$. It is clear that $g_ep_sg_fp_tg_e \Hord g_e$. Furthermore, since $g_e=g_ep_qg_fp_tg_e$ and $g_f = g_fp_rg_ep_sg_f$, we get $g_e=g_ep_q(g_fp_rg_ep_sg_f)p_tg_e=g_ep_qg_fp_r(g_ep_sg_fp_tg_e)$. Therefore, $g_e \Jord g_ep_sg_fp_tg_e$ and Lemma~\ref{lem:jlr} yields $g_e \Hrel g_ep_sg_fp_tg_e$ as desired. The proof of the equality $g_f = (g_fp_tg_ep_sg_f)^\omega$ is symmetrical.

  We turn to proving that $g_e = g_ep_sg_fp_rg_e$ and $g_f = g_fp_tg_ep_qg_f$. By symmetry, we only prove the first equality.  Observe that since $g_f = p_rg_ep_s$, we have $g_ep_sg_fp_rg_eg_ep_sg_fp_rg_e =  g_ep_sg_fp_rg_e$. Thus, $g_ep_sg_fp_rg_e \in E(N)$ and by Lemma~\ref{lem:htogroup} it suffices to prove that $g_ep_sg_fp_rg_e \Hrel g_e$. Clearly, $g_ep_sg_fp_rg_e \Hord g_e$. Furthermore, since we have $g_f = p_rg_ep_s$ and $g_f = (g_fp_tg_ep_sg_f)^\omega$, we know that $g_f\Jord g_ep_sg_fp_rg_e$. Finally, since we have $g_e = p_qg_fp_t$, we get $g_e \Jord g_f \Jord g_ep_sg_fp_rg_e$. It then follows from Lemma~\ref{lem:jlr} that $g_e \Hrel g_ep_sg_fp_rg_e$, which completes the proof.
\end{proof}

We now use the morphism $\beta$ to  decompose $u_e$ and $u_f$. We shall then use the factors to construct our \Cs-set.

\begin{restatable}{lemma}{thedec} \label{lem:thedec}
  There exist six words $x_e,y_e,z_e,x_f,y_f,z_f \in A^*$ and $g' \in E(N)$ satisfying the following conditions:
  \begin{itemize}
    \item $\alpha(y_e) \in E(M)$, $e = \alpha(x_ey_ez_e)$ and $g_e = \eta(x_ey_ez_e)$
    \item $\alpha(y_f) \in E(M)$, $f = \alpha(x_fy_fz_f)$ and $g_f = \eta(x_fy_fz_f)$.
    \item $g' = \eta(y_e) = \eta(y_f)$ and $g' = g'g_fg'$.
    \item $g_f \leq  \eta(x_f)$, $g_f \leq \eta(z_f)$.
    \item $g_ep_qg_f \leq \eta(x_e)$ and $g_fp_tg_e \leq \eta(z_e)$.
  \end{itemize}
\end{restatable}

\begin{proof}
  We let $w_{1,0} = \cdots = w_{k,0} = u_f$ and use induction on an integer $h\geq 1$ to build $2k$ words $v_{1,h},\dots,v_{k,h},w_{1,h},\dots,w_{k,h} \in A^*$ satisfying the following conditions:
  \begin{itemize}
    \item $u_e = v_{1,h} \cdots v_{k,h}$ and $u_f=w_{1,h} \cdots w_{k,h}$.
    \item $\eta(w_{i,h-1}) \leq \eta(v_{i,h}) \leq \eta(w_{i,h})$ for $2 \leq i < k$.
    \item $g_ep_q \eta(w_{1,h-1}) \leq \eta(v_{1,h})$ and $g_fp_r\eta(v_{1,h}) \leq \eta(w_{1,h})$.
    \item $\eta(w_{k,h-1})p_tg_e \leq \eta(v_{k,h})$ and $\eta(v_{k,h})p_sg_f \leq \eta(w_{k,h})$.
  \end{itemize}
  Let $h\geq 1$. Assume that $w_{1,h-1},\dots,w_{k,h-1}$ are defined. We have $\beta(w_{1,h-1} \cdots w_{k,h-1}) = \beta(u_f)$. Since $\beta(u_e) \in E(R)$ and $\beta(u_e)= \beta(u_qu_fu_t)$, we get $\beta(u_e)= \beta(u_eu_qw_{1,h-1} \cdots w_{k,h-1}u_tu_e)$. Furthermore, $\beta$ recognizes the following language by definition,
  \[
    H = \eupp{u_eu_qw_{1,h-1}}\eupp{w_{2,h-1}} \cdots \eupp{w_{k-1,h-1}} \eupp{w_{k,h-1}u_tu_e}.\]
  Since $u_eu_qw_{1,h-1} \cdots w_{k,h-1}u_tu_e \in H$, we get $u_e \in H$. Therefore, there exist words $v_{1,h},\dots,v_{k,h}$ such that $u_e=v_{1,h} \cdots v_{k,h}$, $v_{1,h} \in \eupp{u_eu_qw_{1,h-1}}$, $v_{i,h} \in \eupp{w_{i,h-1}}$ for every $2 \leq i < k$, and $v_{k,h} \in \eupp{w_{k,h-1}u_tu_e}$. We obtain $g_ep_q\eta(w_{1,h-1})\leq \eta(v_{1,h})$,  $\eta(w_{i,h-1}) \leq \eta(v_{i,h})$ for $2 \leq i < k$ and $\eta(w_{k,h-1})p_tg_e \leq \eta(v_{k,h})$. Since $\beta(u_f) \in E(R)$ and $\beta(u_f)= \beta(u_ru_eu_s)$ by hypothesis. A symmetrical argument yields words $w_{1,h},\dots,w_{k,h}$ such that $u_f = w_{1,h} \cdots w_{k,h}$, $g_fp_r\eta(v_{1,h}) \leq \eta(w_{1,h})$, $\eta(v_{i,h}) \leq \eta(w_{i,h})$ for $2 \leq i < k$ and $\eta(v_{k,h})p_sg_f  \leq \eta(w_{k,h})$. This concludes the construction of the $2k$ words satisfying the desired construction.

  Let $d = \omega(M) \times \omega(N)$. Since $N$ is finite, there are two indices $m < m'$ and $\eta(w_{i,dm}) = \eta(w_{i,dm'})$ for $2 \leq i < k$. We fix these two numbers $m,m' \geq 1$. Since $k - 2  = |M|^2 \times |N|^2$, the pigeonhole principle yields $1 \leq i < j < k$ such that,
  \[
    \begin{array}{lll}
      \alpha(w_{2,dm} \cdots w_{i,dm}) & = & \alpha(w_{2,dm} \cdots w_{j,dm}), \\
      \alpha(v_{2,dm+1} \cdots v_{i,dm+1}) & = & \alpha(v_{2,dm+1} \cdots v_{j,dm+1}), \\
      \eta(w_{2,dm} \cdots w_{i,dm}) & = & \eta(w_{2,dm} \cdots w_{j,dm}), \\
      \eta(v_{2,dm+1} \cdots v_{i,dm+1}) & = & \eta(v_{2,dm+1} \cdots v_{j,dm+1}).
    \end{array}
  \]
  We now define,
  \[
    \begin{array}{ll}
      x_e =  v_{1,dm+1} \cdots v_{i,dm+1} &  x_f = w_{1,dm} \cdots w_{i,dm}, \\
      y_e =  (v_{i+1,dm+1} \cdots v_{j,dm+1})^d &  y_f =  (w_{i+1,dm} \cdots w_{j,dm})^d, \\
      z_e = v_{j+1,dm+1} \cdots v_{k,dm+1}  &  z_f = w_{j+1,dm} \cdots w_{k,dm}.
    \end{array}
  \]
  It remains to prove that the conditions in the lemma are satisfied. For the first two, we have $\alpha(y_e),\alpha(y_f) \in E(M)$ by definition of $d$. In addition, by definition of $i$ and $j$, we know that we have
  \[
    \alpha(x_ey_ez_e) = \alpha(v_{1,dm+1} \cdots v_{k,dm+1}) = \alpha(u_e) = e,
  \]
  as well as
  \[
    \alpha(x_fy_fz_f) =\alpha(w_{1,dm} \cdots w_{k,dm}) = \alpha(u_f) = f.
  \]
  Furthermore, $\eta(x_ey_ez_e) = \eta(u_e) = g_e$ and $\eta(x_fy_fz_f) = \eta(u_f) = g_f$.

  We turn to the third assertion. Define $g' = \eta(y_e)$. By choice of $d$,  we have $g' \in E(N)$. Moreover, for $2 \leq \ell < k$, we have $\eta(w_{\ell,0}) \leq \eta(w_{\ell,dm}) \leq \eta(v_{\ell,dm+1}) \leq \eta(w_{\ell,dm'})$. Hence, since $w_{i,0} = u_f$, $g_f = \eta(u_f)$ and $\eta(w_{\ell,dm}) = \eta(w_{\ell,dm'})$ for $2 \leq \ell < k$, we obtain that $g_f \leq \eta(w_{\ell,dm}) = \eta(v_{\ell,dm+1})$. One can multiply these inequalities for $i+1\leq\ell\leq j$. Since $g_f\in E(N)$, this yields $g_f \leq \eta(y_e) = \eta(y_f)$ by definition. Hence, $g' = \eta(y_e) = \eta(y_f)$ and $g_f \leq g'$. We get $g' = g'g_fg'$ from Lemma~\ref{lem:morbp}.

  For the remaining assertions, we only show~that  $g_f \leq  \eta(x_f)$ and $g_ep_qg_f \leq \eta(x_e)$. The proof of the other inequalities is symmetrical. By definition, we have $w_{1,0} = u_f$, $g_ep_q \eta(w_{1,h-1}) \leq \eta(v_{1,h})$ and $g_fp_r\eta(v_{1,h}) \leq \eta(w_{1,h})$ for all $h \geq 1$. Since $g_{\smash{f}} \in E(N)$, we get $g_ep_q(g_fp_rg_ep_qg_f)^{h-1} \leq \eta(v_{1,h})$ and $(g_fp_rg_ep_qg_f)^{h} \leq \eta(w_{1,h})$ by a simple induction. Also,~$g_f = (g_fp_rg_ep_qg_f)^\omega$ by Fact~\ref{fct:etaprop}. Thus,  we have $g_f \leq \eta(w_{1,dm})$ and $g_ep_qg_f \leq \eta(v_{1,dm+1})$ by definition of $d$. Finally, we have $g_f \leq \eta(w_{\ell,dm}) = \eta(v_{\ell,dm+1})$ for $2 \leq i < k$. Hence, $g_f \leq  \eta(x_f)$ and $g_ep_qg_f \leq \eta(x_e)$.
\end{proof}

We fix the words $x_e,y_e,z_e,x_f,y_f,z_f \in A^*$  and the element $g' \in E(N)$ given by Lemma~\ref{lem:thedec}. Let $h = \omega(M) \times \omega(N)$. We let,
\[
  \begin{array}{ll}
    y_q  =  y_ez_e u_q x_fy_f, & y_r  = y_fz_f (u_ru_eu_qu_f)^{h-1} u_rx_ey_e, \\
    y_t = y_fz_f u_t  x_ey_e, & y_s  =  y_ez_e u_s (u_fu_tu_eu_s)^{h-1} x_fy_f.
  \end{array}
\]
Let $\hat{x} = \alpha(y_x)$ for $x \in \{q,r,s,t,e,f\}$. Recall that in order to  apply Equation~\eqref{eq:csides}, we need to construct a \Cs-set. We will show below that  $\{\hat{q},\hat{r},\hat{s},\hat{t},\hat{e},\hat{f}\}$ is such a \Cs-set.
Note that $\hat{e},\hat{f} \in E(M)$ by Lemma~\ref{lem:thedec}. Since $\alpha(y_f),\alpha(y_e) \in E(M)$ and $\alpha(x_fy_fz_f) = \alpha(u_f) = f$ by Lemma~\ref{lem:thedec}, one can verify that,
\[
  \begin{array}{lll}
    \hat{e}\hat{q}\hat{f}\hat{r}\hat{e} & = & \alpha(y_ez_e)qf (reqf)^{h-1} r \alpha(x_ey_e), \\
                                        & = & \alpha(y_ez_e) (qfre)^{h-1} qfr\alpha(x_ey_e).
  \end{array}
\]
Symmetrically, one can verify that,
\[
  \begin{array}{lll}
    \hat{e}\hat{s}\hat{f}\hat{t}\hat{e} & = & \alpha(y_ez_e) s(ftes)^{h-1} ft\alpha(x_ey_e), \\
                                        &= & \alpha(y_ez_e) sft(esft)^{h-1}\alpha(x_ey_e).
  \end{array}
\]
Since $\alpha(x_ey_ez_e) = e$ by Lemma~\ref{lem:thedec}, it follows that,
\[
  \alpha(x_e) (\hat{e}\hat{q}\hat{f}\hat{r}\hat{e})^\omega (\hat{e}\hat{s}\hat{f}\hat{t}\hat{e})^\omega \alpha(z_e) =  (eqfre)^\omega (esfte)^\omega.
\]
Moreover, the definitions yield $\hat{q}\hat{f}\hat{t} = \alpha(y_ez_e)qft\alpha(x_ey_e)$. Hence, using again the equality  $\alpha(x_ey_ez_e) = e$ of Lemma~\ref{lem:thedec}, we get,
\[
  \alpha(x_e) (\hat{e}\hat{q}\hat{f}\hat{r}\hat{e})^\omega \hat{q}\hat{f}\hat{t} (\hat{e}\hat{s}\hat{f}\hat{t}\hat{e})^\omega \alpha(z_e)  =  (eqfre)^\omega qft (esfte)^\omega.
\]
We prove that $(\hat{e}\hat{q}\hat{f}\hat{r}\hat{e})^\omega (\hat{e}\hat{s}\hat{f}\hat{t}\hat{e})^\omega = (\hat{e}\hat{q}\hat{f}\hat{r}\hat{e})^\omega \hat{q}\hat{f}\hat{t} (\hat{e}\hat{s}\hat{f}\hat{t}\hat{e})^\omega$. In view of the above, this will imply that~\eqref{eq:simpeqhard} holds as desired (namely that $(eqfre)^\omega (esfte)^\omega = (eqfre)^\omega qft (esfte)^\omega$). Since $\alpha$ satisfies~\eqref{eq:csides}, it suffices to prove that $\{\hat{q},\hat{r},\hat{s},\hat{t},\hat{e},\hat{f}\} \subseteq M$ is a \Cs-set. We show that $\{\hat{q},\hat{r},\hat{s},\hat{t},\hat{e},\hat{f}\}$ is a subset of $\alpha(\eta\inv(g'))$, which is a \Cs-set by Lemma~\ref{lem:morbp}. Thus, $\{\hat{q},\hat{r},\hat{s},\hat{t},\hat{e},\hat{f}\}$ is itself a \Cs-set by Fact~\ref{fct:setinc}.

By definition, we have to prove that for $y \in \{y_q,y_r,y_s,y_t,y_e,y_f\}$, we have $\eta(y) = g'$. By Lemma~\ref{lem:thedec}, we already know this for $y=y_e$ and $y=y_f$. By symmetry, we only treat the cases $y = y_q$ and $y = y_r$. It suffices to prove that $\eta(y) \Hrel g'$ and $g' \leq \eta(y)$. Indeed, since $g' \in E(N)$, Lemma~\ref{lem:htogroup} and $\eta(y) \Hrel g'$ imply $(\eta(y))^\omega = g'$. Moreover, $g' \leq \eta(y)$ yields $g' \leq (\eta(y))^{\omega-1}$. Multiplying by $\eta(y)$ yields $g'\eta(y) \leq (\eta(y))^{\omega} = g'$. Since $g'\eta(y) = \eta(y)$ (by definition of $y_q,y_r$ as $\eta(y_e) = \eta(y_f) = g'$), we get $\eta(y) \leq g'$. Altogeter, we obtain $g' = \eta(y)$, as desired.

We first show that $\eta(y_q) \Hrel g'$ and $g' \leq \eta(y_q)$. We know that $y_q =  y_ez_e u_q x_fy_f$ and Lemma~\ref{lem:thedec} yields $g_f \leq  \eta(x_f)$ and $g_fp_tg_e \leq \eta(z_e)$. We get $g'g_fp_tg_ep_qg_fg' \leq \eta(y_q)$. Thus, since $g_fp_tg_ep_qg_f = g_f$ by Fact~\ref{fct:etaprop} and $g'g_fg' = g'$ by Lemma~\ref{lem:thedec}, we get $g' \leq \eta(y_q)$. Moreover, since $g' = \eta(y_e) = \eta(y_f)$, it is clear that $\eta(y_q) \Hord g'$. Thus, we show that $g' \Jord \eta(y_q)$. This will imply $\eta(y_q) \Hrel g'$ by Lemma~\ref{lem:jlr}. By definition, $y_q = y_ez_e u_q x_fy_f$. Hence, Lemma~\ref{lem:thedec} yields,
\[
  \eta(x_ey_qz_f)=\eta(x_e(y_ez_e u_q x_fy_f)z_f)=\eta(u_eu_qu_f) = g_ep_qg_f.
\]
Since $g_fp_tg_ep_qg_f =g_f$ by  Fact~\ref{fct:etaprop}, we get $g_f \Jord \eta(y_q)$. Also, $g' \Jord g_f$ since $g'=g'g_f g'$. We get $g' \Jord \eta(y_q)$ as desired.

We now prove that $\eta(y_r) \Hrel g'$ and that $g' \leq \eta(y_r)$. We have $y_r = y_fz_f (u_ru_eu_qu_f)^{h-1} u_rx_ey_e$. Lemma~\ref{lem:thedec} yields $g_f \leq  \eta(z_f)$ and $g_ep_qg_f \leq \eta(x_e)$. Thus, $g'(g_fp_rg_ep_qg_f)^{h}g' \leq \eta(y_r)$. Since $h$ is a multiple of $\omega(N)$, this yields $g'g_fg' \leq \eta(y_r)$ by Fact~\ref{fct:etaprop}. Thus, $g' \leq \eta(y_r)$ by Lemma~\ref{lem:thedec}. Also, the definition of $y_r$ and the fact that $g' = \eta(y_e) = \eta(y_f)$ yield $\eta(y_r) \Hord g'$. Thus, it suffices to show that $g' \Jord \eta(y_r)$. This implies $\eta(y_r) \Hrel g'$ by Lemma~\ref{lem:jlr}. By definition, $y_r = y_fz_f (u_ru_eu_qu_f)^{h-1} u_rx_ey_e$. Thus, Lemma~\ref{lem:thedec} yields,
\[
  \eta(x_fy_rz_e)= (g_fp_rg_ep_qg_f)^{h} p_rg_e.
\]
Moreover, $(g_fp_rg_ep_qg_f)^h = g_f$ by Fact~\ref{fct:etaprop}. Thus,  $g_f\Jord\eta(y_r)$. Finally, since  $g'=g'g_fg'$, we also have $g' \Jord g_f$. We get $g' \Jord \eta(y_r)$, which completes the proof.

\section{Conclusion}
\label{sec:conc}
In this paper, we studied the \bpolo operator used in the construction of concatenation hierarchies. We proved a generic characterization of the classes \bpol{\Cs} that depends on an \emph{ad hoc} property of the class~\Cs. While it is not clear when this property on \Cs is decidable in general, a first consequence is a new effective characterization of levels two in the dot-depth and Straubing-Thérien hierarchies, which is simpler than the previously known ones.

When \Cs is itself of the form \bpol{\Ds}, our characterization simplifies and yields a reduction from \bpol{\Cs}-membership to \Cs-covering. The latter problem was known to be decidable for level one in several hierarchies whose bases consist of group languages such as the group and modulo hierarchies. Combined with our results, this implies that level two is decidable in these hierarchies. Finally, it was known that level \emph{two} has decidable covering in the Straubing-Thérien hierarchy and the dot-depth. Hence, we obtain that level \emph{three} is decidable in these hierarchies.

\bibliographystyle{plain}

\bibliography{main}

\begin{thebibliography}{10}

\bibitem{almeidabc1proof}
Jorge Almeida.
\newblock Implicit operations on finite {\bf {j}}-trivial semigroups and a
  conjecture of {I.} {S}imon.
\newblock {\em Journal of Pure and Applied Algebra}, 69:205--218, 1990.

\bibitem{MR1647225}
Jorge Almeida.
\newblock Implicit operations and {K}nast's theorem.
\newblock In {\em Semigroups ({L}uino, 1992)}, pages 1--16. World Sci. Publ.,
  River Edge, NJ, 1993.

\bibitem{AK2009}
Jorge Almeida and Ondrej Klíma.
\newblock A counterexample to a conjecture concerning concatenation
  hierarchies.
\newblock {\em Information Processing Letters}, 110(1):4--7, 2009.

\bibitem{AK2010}
Jorge Almeida and Ondrej Klíma.
\newblock New decidable upper bound of the 2nd level in the
  {S}traubing-{T}hérien concatenation hierarchy of star-free languages.
\newblock {\em Discrete Mathematics {\&} Theoretical Computer Science},
  12(4):41--58, 2010.

\bibitem{arfi87}
Mustapha Arfi.
\newblock Polynomial operations on rational languages.
\newblock In {\em Proceedings of the 4th Annual Symposium on Theoretical
  Aspects of Computer Science}, STACS'87, pages 198--206. Springer, 1987.

\bibitem{arfi91}
Mustapha Arfi.
\newblock Op\'erations polynomiales et hi\'erarchies de concat\'enation.
\newblock {\em Theoretical Computer Science}, 91(1):71 -- 84, 1991.

\bibitem{Ash91}
Christopher~J. Ash.
\newblock Inevitable graphs: a proof of the type {II} conjecture and some
  related decision procedures.
\newblock {\em International Journal of Algebra and Computation},
  1(1):127--146, 1991.

\bibitem{conjdd2-blanchet1}
F.~Blanchet-Sadri.
\newblock On dot-depth two.
\newblock {\em RAIRO - Theoretical Informatics and Applications - Informatique
  Th\'eorique et Applications}, 24(6):521--529, 1990.

\bibitem{conjdd2-blanchet2}
F.~Blanchet-Sadri.
\newblock On a complete set of generators for dot-depth two.
\newblock {\em Discrete Applied Mathematics}, 50(1):1--25, 1994.

\bibitem{BrzoDot}
Janusz~A. Brzozowski and Rina~S. Cohen.
\newblock Dot-depth of star-free events.
\newblock {\em Journal of Computer and System Sciences}, 5(1):1--16, 1971.

\bibitem{BroKnaStrict}
Janusz~A. Brzozowski and Robert Knast.
\newblock The dot-depth hierarchy of star-free languages is infinite.
\newblock {\em Journal of Computer and System Sciences}, 16(1):37--55, 1978.

\bibitem{ChaubardPS06}
Laura Chaubard, Jean{-}{\'E}ric Pin, and Howard Straubing.
\newblock First order formulas with modular predicates.
\newblock In {\em Proceedings of the 21th {IEEE} Symposium on Logic in Computer
  Science ({LICS}'06)}, pages 211--220, 2006.

\bibitem{Cowan_1993}
David Cowan.
\newblock Inverse monoids of dot-depth two.
\newblock {\em International Journal of Algebra and Computation},
  03(04):411--424, 1993.

\bibitem{cmmptsep}
Wojciech Czerwi{\'n}ski, Wim Martens, and Tom{\'a}{\v{s}} Masopust.
\newblock Efficient separability of regular languages by subsequences and
  suffixes.
\newblock In {\em Proceedings of the 40th International Colloquium on Automata,
  Languages, and Programming}, ICALP'13, pages 150--161. Springer, 2013.

\bibitem{green}
James~Alexander Green.
\newblock On the structure of semigroups.
\newblock {\em Annals of Mathematics}, 54(1):163--172, 1951.

\bibitem{henckell:hal-00019815}
Karsten Henckell, Stuart Margolis, Jean-{\'E}ric Pin, and John Rhodes.
\newblock {Ash's type II theorem, profinite topology and Malcev products}.
\newblock {\em {International Journal of Algebra and Computation}}, 1:411--436,
  1991.

\bibitem{krpgbg}
Karsten Henckell and John Rhodes.
\newblock The theorem of knast, the pg=bg and type {II} conjectures.
\newblock {\em Monoids and semigroups with applications, Word Scientific},
  pages 453--463, 1991.

\bibitem{higginsbc1proof}
Peter Higgins.
\newblock A proof of simon's theorem on piecewise testable languages.
\newblock {\em Theoretical computer science}, 178(1):257--264, 1997.

\bibitem{klimabc1proof}
Ondrej Klima.
\newblock Piecewise testable languages via combinatorics on words.
\newblock {\em Discrete Mathematics}, 311(20):2124--2127, 2011.

\bibitem{knast83}
Robert Knast.
\newblock A semigroup characterization of dot-depth one languages.
\newblock {\em RAIRO - Theoretical Informatics and Applications},
  17(4):321--330, 1983.

\bibitem{KufleitnerL12}
Manfred Kufleitner and Alexander Lauser.
\newblock Around dot-depth one.
\newblock {\em International Journal of Foundations of Computer Science},
  23(6):1323--1340, 2012.

\bibitem{KufleitnerW15}
Manfred Kufleitner and Tobias Walter.
\newblock One quantifier alternation in first-order logic with modular
  predicates.
\newblock {\em {RAIRO} Theor. Informatics Appl.}, 49(1):1--22, 2015.

\bibitem{MargolisP85}
Stuart~W. Margolis and Jean{-}{\'E}ric Pin.
\newblock Products of group languages.
\newblock In {\em {FCT}'85}. Springer, 1985.

\bibitem{nash-williams63}
Crispin St. John~Alvah Nash-Williams.
\newblock On well-quasi-ordering finite trees.
\newblock {\em Mathematical Proceedings of the Cambridge Philosophical
  Society}, 59(4):833–835, 1963.

\bibitem{PPOrder}
Dominique Perrin and Jean-\'Eric Pin.
\newblock First-order logic and star-free sets.
\newblock {\em Journal of Computer and System Sciences}, 32(3):393--406, 1986.

\bibitem{pinvarbook}
Jean-{\'E}ric Pin.
\newblock {\em Varieties Of Formal Languages}.
\newblock Plenum Publishing Co., 1986.

\bibitem{jeppgbg}
Jean-{\'E}ric Pin.
\newblock {PG = BG}, a success story.
\newblock {\em NATO Advanced Study Institute, Semigroups, Formal Languages and
  Groups}, pages 33--47, 1995.

\bibitem{jep-intersectPOL}
Jean-{\'E}ric Pin.
\newblock An explicit formula for the intersection of two polynomials of
  regular languages.
\newblock In {\em {DLT 2013}}, volume 7907 of {\em Lect. Notes Comp. Sci.},
  pages 31--45. Springer, 2013.

\bibitem{jep-dd45}
Jean-{\'E}ric Pin.
\newblock {\em The dot-depth hierarchy, 45 years later}, chapter~8, pages
  177--202.
\newblock World Scientific, 2017.

\bibitem{pingoodref}
Jean-{\'E}ric Pin.
\newblock Mathematical foundations of automata theory.
\newblock In preparation \url{https://www.irif.fr/~jep/PDF/MPRI/MPRI.pdf},
  2019.

\bibitem{pin-straubing:upper}
Jean-{\'E}ric Pin and Howard Straubing.
\newblock Monoids of upper triangular boolean matrices.
\newblock In {\em Semigroups. Structure and Universal Algebraic Problems},
  volume~39, pages 259--272. North-Holland, 1985.

\bibitem{pwmalcev}
Jean-\'Eric Pin and Pascal Weil.
\newblock Profinite semigroups, {Mal'cev} products and identities.
\newblock {\em Journal of Algebra}, 182(3):604--626, 1996.

\bibitem{pwdelta2}
Jean-{\'E}ric Pin and Pascal Weil.
\newblock Polynomial closure and unambiguous product.
\newblock {\em Theory of Computing Systems}, 30(4):383--422, 1997.

\bibitem{Pin_2001}
Jean-\'Eric Pin and Pascal Weil.
\newblock A conjecture on the concatenation product.
\newblock {\em {RAIRO Informatique Théorique}}, 35(6):597--618, 2001.

\bibitem{PlaceZ22}
Thomas Place and Marc Zeitoun.
\newblock A generic polynomial time approach to separation by first-order logic
  without quantifier alternation.
\newblock In {\em 42nd {IARCS} Annual Conference on Foundations of Software
  Technology and Theoretical Computer Science, {FSTTCS} 2022, December 18-20,
  2022, {IIT} Madras, Chennai, India}, volume 250 of {\em LIPIcs}, pages
  43:1--43:22. Schloss Dagstuhl - Leibniz-Zentrum f{\"{u}}r Informatik.

\bibitem{pzqalt}
Thomas Place and Marc Zeitoun.
\newblock Going higher in the first-order quantifier alternation hierarchy on
  words.
\newblock In {\em Proceedings of the 41st International Colloquium on Automata,
  Languages, and Programming}, ICALP'14, pages 342--353. Springer, 2014.

\bibitem{pzcovering2}
Thomas Place and Marc Zeitoun.
\newblock The covering problem.
\newblock {\em Logical Methods in Computer Science}, 14(3), 2018.

\bibitem{pzupol}
Thomas Place and Marc Zeitoun.
\newblock Separating without any ambiguity.
\newblock In {\em Proceedings of the 45th International Colloquium on Automata,
  Languages, and Programming}, ICALP'18, pages 137:1--137:14, 2018.

\bibitem{PZ:generic18}
Thomas Place and Marc Zeitoun.
\newblock Generic results for concatenation hierarchies.
\newblock {\em Theory of Computing Systems (ToCS)}, 63(4):849--901, 2019.
\newblock Selected papers from CSR'17.

\bibitem{pzjacm19}
Thomas Place and Marc Zeitoun.
\newblock Going higher in first-order quantifier alternation hierarchies on
  words.
\newblock {\em Journal of the {ACM}}, 66(2):12:1--12:65, 2019.

\bibitem{pzconcagroup}
Thomas Place and Marc Zeitoun.
\newblock Separation and covering for group based concatenation hierarchies.
\newblock In {\em Proceedings of the 34th Annual {ACM/IEEE} Symposium on Logic
  in Computer Science}, {LICS}'19, pages 1--13, 2019.

\bibitem{pzbpolcj}
Thomas Place and Marc Zeitoun.
\newblock Separation for dot-depth two.
\newblock {\em Logical Methods in Computer Science}, 17(3), 2021.

\bibitem{pzupol2}
Thomas Place and Marc Zeitoun.
\newblock All about unambiguous polynomial closure.
\newblock {\em {TheoretiCS}}, 2(11):1--74, 2023.

\bibitem{pzgr}
Thomas Place and Marc Zeitoun.
\newblock Group separation strikes back.
\newblock In {\em {38th Annual ACM/IEEE Symposium on Logic in Computer Science,
  LICS'23}}, pages 1--13. IEEE Computer Society, 2023.

\bibitem{schutzsf}
Marcel~Paul Sch{\"u}tzenberger.
\newblock On finite monoids having only trivial subgroups.
\newblock {\em Information and Control}, 8(2):190--194, 1965.

\bibitem{simonphd}
Imre Simon.
\newblock {\em Hierarchies of Events of Dot-Depth One}.
\newblock PhD thesis, University of Waterloo, Department of Applied Analysis
  and Computer Science, Waterloo, Ontario, Canada, 1972.

\bibitem{simonthm}
Imre Simon.
\newblock Piecewise testable events.
\newblock In {\em Proceedings of the 2nd GI Conference on Automata Theory and
  Formal Languages}, pages 214--222. Springer, 1975.

\bibitem{stockphd}
Larry~J. Stockmeyer.
\newblock {\em The complexity of decision problems in automata theory and
  logic}.
\newblock PhD thesis, 1974.
\newblock PHD.

\bibitem{StrauConcat}
Howard Straubing.
\newblock A generalization of the sch{\"u}tzenberger product of finite monoids.
\newblock {\em Theoretical Computer Science}, 13(2):137--150, 1981.

\bibitem{StrauVD}
Howard Straubing.
\newblock Finite semigroup varieties of the form {V {\textasteriskcentered} D}.
\newblock {\em Journal of Pure and Applied Algebra}, 36:53--94, 1985.

\bibitem{StrauDD2Conf}
Howard Straubing.
\newblock Semigroups and languages of dot-depth 2.
\newblock In {\em Proceedings of the 13th International Colloquium on Automata,
  Languages, and Programming, {ICALP'86}}, volume 226 of {\em Lecture Notes in
  Computer Science}, pages 416--423. Springer, 1986.

\bibitem{StraubingDD2Journal}
Howard Straubing.
\newblock Semigroups and languages of dot-depth two.
\newblock {\em Theoretical Computer Science}, 58(1):361--378, 1988.

\bibitem{bookstraub}
Howard Straubing.
\newblock {\em Finite Automata, Formal Logic and Circuit Complexity}.
\newblock Birkhauser, Basel, Switzerland, 1994.

\bibitem{stbc1proof}
Howard Straubing and Denis Th\'erien.
\newblock Partially ordered finite monoids and a theorem of {I}. {S}imon.
\newblock {\em Journal of Algebra}, 119(2):393--399, 1988.

\bibitem{Straubing_1992}
Howard Straubing and Pascal Weil.
\newblock On a conjecture concerning dot-depth two languages.
\newblock {\em Theoretical Computer Science}, 104(2):161--183, 1992.

\bibitem{TheConcat}
Denis Th{\'e}rien.
\newblock Classification of finite monoids: The language approach.
\newblock {\em Theoretical Computer Science}, 14(2):195--208, 1981.

\bibitem{Therien88}
Denis Th{\'{e}}rien.
\newblock Categories et langages de dot-depth un.
\newblock {\em {RAIRO} Theor. Informatics Appl.}, 22(4):437--445, 1988.

\bibitem{permauto}
Gabriel Thierrin.
\newblock Permutation automata.
\newblock {\em Theory of Computing Systems}, 2(1):83--90, 1968.

\bibitem{ThomEqu}
Wolfgang Thomas.
\newblock Classifying regular events in symbolic logic.
\newblock {\em Journal of Computer and System Sciences}, 25(3):360--376, 1982.

\bibitem{phdpw}
Pascal Weil.
\newblock {\em Inverse monoids and the dot-depth hierarchy}.
\newblock PhD thesis, University of Nebraska, 1988.

\bibitem{Weil_1989}
Pascal Weil.
\newblock Inverse monoids of dot-depth two.
\newblock {\em Theoretical Computer Science}, 66(3):233--245, 1989.

\bibitem{Weil_1993}
Pascal Weil.
\newblock Some results on the dot-depth hierarchy.
\newblock {\em Semigroup Forum}, 46(1):352--370, 1993.

\bibitem{Zetzsche18}
Georg Zetzsche.
\newblock Separability by piecewise testable languages and downward closures
  beyond subwords.
\newblock In {\em Proceedings of the 33rd Annual {ACM/IEEE} Symposium on Logic
  in Computer Science}, LICS'18, pages 929--938, 2018.

\end{thebibliography}

\end{document}